\documentclass[12pt]{article}
\usepackage{amssymb}
\usepackage{a4}
\usepackage{authblk}
\usepackage{amsmath}
\usepackage{graphicx}
\usepackage{subcaption}

\usepackage{comment}
\usepackage{hyperref}
\allowdisplaybreaks[4]

\evensidemargin \oddsidemargin
\def\1{{\bf 1}}
\def\id{\mbox{id\,}}
\def\ot{\!\otimes\!}
\def\ots{\otimes_{\star}}
\def\F{{\cal F}}
\def\bF{\mbox{$\overline{\cal F}$}}

\def\f{{\scriptscriptstyle {\cal F}}}
\def\bbf{{\scriptscriptstyle \overline{\cal F}}}

\def\cross{\mbox{$\rule{0.7pt}{1.3ex}\!\times $}}
\def\cocross{\mbox{$\times \!\rule{0.3pt}{1.1ex}\,$}}
\def\ra{\rangle}
\def\la{\langle}

\def \A{{\cal A}}
\def \B {{\cal B}}

\def \C {{\cal C}}
\def \C {{\cal C}}

\def \E {{\sf E}}

\def \Ms {{\scriptscriptstyle M}}
\def \Mcs {{\scriptscriptstyle M_c}}

\def \O {{\cal O}}

\def \L {{\cal L} }
\def \M {{\cal M}}
\def \N {{\cal N}}
\def \Q {{\cal Q}}
\def \QM {{\cal Q}_{\scriptscriptstyle M_c}}
\def \QMst {{\cal Q}_{{\scriptscriptstyle M_c}\star}}
\def \X {{\cal X} }

\def\R{\mbox{$\cal R$}}
\def\r{{\scriptscriptstyle {\cal R}}}
\def\bR{\mbox{$\overline{\cal R}$}}
\def\br{{\scriptscriptstyle \overline{\cal R}}}

\def\hH{\mbox{$\hat{H}$}}

\newcommand{\trc}{\triangleright}

\newcommand{\ltlc}{\stackrel{\scriptscriptstyle \triangleleft}{}}
\def\g{\mathfrak{g}}
\def\k{\mathfrak{k}}
\def\h{\mathfrak{h}}

\newcommand{\gm}{{\bf g}}
\newcommand{\gmp}{\gm_{\scriptscriptstyle \perp}}
\newcommand{\gmps}{\gm_{\scriptscriptstyle \perp\star}}
\newcommand{\XiM}{\Xi_{\scriptscriptstyle M}}
\newcommand{\XiMst}{\Xi_{\scriptscriptstyle M\star}}

\newcommand{\XM}{{\cal X}^{\scriptscriptstyle M}}
\newcommand{\Xip}{\Xi_{\scriptscriptstyle \perp}}
\newcommand{\Xips}{\Xi_{{\scriptscriptstyle \perp}\star}}
\newcommand{\Up}{U_{\scriptscriptstyle \perp}}

\newcommand{\Vp}{V_{\scriptscriptstyle \perp}}

\newcommand{\Np}{N_{\scriptscriptstyle \perp}}

\newcommand{\Omp}{\Omega_{\scriptscriptstyle \perp}}
\newcommand{\Omps}{\Omega_{\scriptscriptstyle \perp \star}}
\newcommand{\Pp}{\mathrm{pr}_{\scriptscriptstyle \perp}}
\newcommand{\Pps}{\mathrm{pr}_{{\scriptscriptstyle{\perp}}\star}}
\newcommand{\Pt}{\mathrm{pr}_{ t}}
\newcommand{\Pts}{\mathrm{pr}_{t\star}}

\def\TT{{\mathcal T}}
\def\Xis{{\Xi_\star }}
\def\Om{\Omega}
\def\Oms{\Omega_\star}
\def\rR{\mathsf{R}}
\def\ric{\mathsf{Ric}}
\def\tT{\mathsf{T}}

\def\b#1{{\mathbb #1}}
\newcommand{\CC}{\mathbb{C}}
\newcommand{\RR}{\mathbb{R}}

\newcommand{\NN}{\mathbb{N}}

\def\nn{\nonumber \\}

\newcommand{\be}{\begin{equation}}
\newcommand{\ee}{\end{equation}}
\newcommand{\bea}{\begin{eqnarray}}
\newcommand{\eea}{\end{eqnarray}}
\newcommand{\ba}{\begin{array}}
\newcommand{\ea}{\end{array}}
\newtheorem{theorem}{Theorem}
\newtheorem{prop}[theorem]{Proposition}

\def\sq{\mbox{\rlap{$\sqcap$}$\sqcup$}}
\newenvironment{proof}[1]{\vspace{5pt}\noindent{\bf Proof #1}\hspace{6pt}}%
{\hfill\sq}
\newcommand{\bp}{\begin{proof}}
\newcommand{\ep}{\end{proof}\par\vspace{10pt}\noindent}

\begin{document}

\title{Twisted submanifolds of $\RR^n$}

\author{Gaetano Fiore\footnote{
Dip. di Matematica e Applicazioni, Universit\'a di Napoli ``Federico II'',
\& I.N.F.N., Sezione di Napoli,\\
Complesso Universitario MSA,
Via Cintia, 80126 Napoli, Italia.\ 
Email: gaetano.fiore@na.infn.it}
~and
Thomas Weber\footnote{
Dip. di Matematica, Universit\'a di Bologna,
Piazza di Porta S. Donato, 5, 40126 Bologna, Italia.\\
Email: thomasmartin.weber@unibo.it}}

\date{}

\maketitle
\abstract{We propose a general procedure to construct  noncommutative deformations
of an embedded submanifold $M$ of $\RR^n$ determined by a set of smooth  equations $f^a(x)=0$. We use the framework of Drinfel'd  twist 
deformation of differential geometry of [Aschieri et al.,
Class. Quantum Gravity 23 (2006), 1883]; the commutative pointwise product
is replaced by a (generally noncommutative) $\star$-product determined by a  Drinfel'd twist. The twists we employ are based on the Lie algebra $\Xi_t$
of vector fields that are tangent to {\it all} 
the submanifolds that are level sets  of the $f^a$ (tangent infinitesimal diffeomorphisms); the twisted Cartan calculus is automatically equivariant under twisted $\Xi_t$. We can
consistently project a connection from the twisted $\RR^n$ to the twisted 
$M$ if the twist is based on a suitable Lie subalgebra $\mathfrak{e}\subset\Xi_t$. 
If we endow $\RR^n$ with a metric, then twisting and projecting to the normal and tangent vector fields commute, and we can project the Levi-Civita connection
consistently  to the twisted $M$, 
provided the twist is based on the Lie subalgebra $\k\subset\mathfrak{e}$
of the  Killing  vector fields of the metric;  a twisted Gauss theorem
follows, in particular. Twisted algebraic manifolds
can be characterized in terms of generators
and $\star$-polynomial relations. We present in some detail twisted cylinders embedded  in twisted Euclidean  $\RR^3$
and twisted hyperboloids  embedded in twisted  Minkowski $\RR^3$ [these are twisted (anti-)de Sitter spaces $dS_2,AdS_2$].
}


\medskip\noindent
{\bf Declarations.} \ \
 i) Funding: We do not acknowledge any specific funding for this work. \

\noindent
ii) Conflicts of interest/Competing interests: On behalf of all authors, the corresponding author states that there is no conflict of interest. \

\noindent
iii) Availability of data and material: not applicable. \

\noindent
iv) Code availability: not applicable.


\section{Introduction}

The notion of a  submanifold $N$ of a manifold $M$ is a fundamental concept in (differential) geometry, playing   a crucial role in various branches of mathematics and physics.
A metric, connection, ..., on $M$ uniquely induces (see e.g. \cite{Kobayashi1996}) 
a metric, connection, ..., on $N$.
In the last few decades various deep physical and mathematical reasons have stimulated the 
generalization of differential geometry to so-called Noncommutative Geometry (NCG)  \cite{Connes,Lan97,Madore99,Majid2000,GraFigVar00}. 
In particular, NCG has been advocated as a suitable framework for formulating
a fundamental (or at least an effective) theory of quantum spacetime allowing the
quantization of gravity (see e.g. \cite{DopFreRob95,Aschieri2006}) and/or for unifying fundamental interactions (see e.g. \cite{ConLot91,ChaConvan}).
It is therefore natural and 
important to investigate whether and to what extent a notion of a submanifold is possible and fruitful in the NCG framework.  Surprisingly, this question has received little 
systematic attention so far 
(rather isolated exceptions are the papers  \cite{Masson1995,Giunashvili,TWeber2019}).
On several noncommutative (NC) spaces one can make sense of special classes of NC submanifolds,
but some features of the latter may depart from their commutative counterparts. For instance,
from the noncommutative algebra ``of functions on the quantum group $U_q(n)$", which 
is generated by the $n^2$ matrix elements of a $n\times n$ unitary matrix, one can obtain
the one $\A$ on the quantum group $SU_q(n)$  by imposing that the so-called $q$-determinant  (a suitable central element) be 1, as in the commutative  ($q=1$) limit; but the so-called quantum group bicovariant differential calculus on $\A$ (i.e. the corresponding $\A$-bimodule $\Omega$ of 1-forms) remains of dimension  $n^2$ instead of  $n^2-1$ \cite{Woronowicz1989}.
The same phenomenon occurs e.g. obtaining the $SO_q(n)$-covariant 
quantum Euclidean spheres $S^{n-1}_q$   from the  $SO_q(n)$-covariant 
quantum Euclidean spaces $\RR^n_q$, by imposing that the [central and
$SO_q(n)$-invariant] square distance $r^2$
from  the origin be 1; said differently, the 1-form $dr^2$ cannot be imposed to vanish, and actually the graded commutator \
$\left[\frac 1{q^2-1}r^{-2} dr^2,\,\cdot\,\right]$ \ acts as  the exterior derivative  \cite{Fio06JPCS,Ste96JMP,FioMad00,CerFioMad01}.

In the present work we wish to address the above question systematically within the framework
of deformation quantization \cite{BayFlaFroLicSte} in the particular approach based on  Drinfel'd twisting \cite{Drinfeld1983} of Hopf algebras.
We restrict our attention to the noncommutative generalization of embedded submanifolds of $\RR^n$, because
by the Whitney embedding theorems \cite{Lee2012}
one can always embed a smooth manifold $M$ in  
$\RR^n$ with  a sufficiently high dimension $n$. 
More precisely, we shall assume that $ M\subset \RR^n$
consists  of points of $x\in\RR^n$  fulfilling a set of equations  
\be
f^a(x)=0,\qquad a=1,2,...,k<n,                  \label{DefIdeal}
\ee
where $f\equiv(f^1,...,f^k):\RR^n\rightarrow\RR^k$ are smooth 
functions such that the Jacobian matrix $J=\partial f/\partial x$ is of rank $k$ 
on all $\RR^n$;
or, more generally, that $f$ is well-defined and $J$ is of rank $k$ on an open
subset ${\cal D}_f\subset\RR^n$, and $M$
consists  of the points of  ${\cal D}_f$  fulfilling (\ref{DefIdeal}).
In all our examples here ${\cal E}_f:=\RR^n\setminus{\cal D}_f$ will be
empty or of zero measure.
By replacing in (\ref{DefIdeal}) \ $f^a(x)\mapsto f^a_c(x):=f^a(x)-c^a$, \ with  $c \equiv(c^1,...,c^k)\in f\left({\cal D}_f\right)
$, \ we
 obtain a whole $k$-parameter family of embedded manifolds $M_c$ ($M_0=M$) of dimension $n\!-\!k$,  that are level sets of $f$. 
Embedded submanifolds $N\subset M$ can be obtained by
adding more equations to (\ref{DefIdeal})\footnote{Given a smooth 
manifold $M$ one could introduce such a $N\subset M$ also 
{\it patchwise}, i.e. as a level set of eqs. of the type (\ref{DefIdeal}) 
on each chart of an atlas of $M$ fulfilling suitable matching conditions in intersecting charts.}.
The $*$-algebra $\XM$ of smooth  complex-valued functions on $M$
can be expressed as the quotient of  the  
$*$-algebra  $\X=C^\infty({\cal D}_f)$ of  smooth 
 functions on ${\cal D}_f$ over the ideal $\C\subset\X$ of smooth functions 
vanishing on $M$:
\be
\XM:=\X/\C\equiv\big\{\, [\alpha]:=\alpha+\C \:\: |\:\: \alpha\in\X\big\};            \label{quotient}
\ee
In the appendix we prove that $\C$ is generated by the left-hand sides (lhs) of (\ref{DefIdeal}):
\begin{theorem}
\ \ $\C=
\bigoplus_{a=1}^k \X f^a =\bigoplus_{a=1}^k f^a\X$, \ i.e. for all
$h\in\C$ there exist  $h^a\in\X$ such that
\be
h(x)=\sum_{a=1}^k h^a(x)f^a(x)=\sum_{a=1}^k f^a(x)h^a(x).        \label{DecoCeq}
\ee
\label{DecoC}
\end{theorem}
Similarly, $\X^{\scriptscriptstyle N}$ can be obtained as the quotient  of  $\XM$ over 
the ideal generated by further equations of the type (\ref{DefIdeal}), 
or equivalently  as the quotient  of $\X$  over the larger ideal generated by all such equations. Identifying vector fields with first order differential operators, we denote as \
$\Xi:=\{X=X^i\partial_i \:\: |\:\:  X^i\in \X\}$ \ the Lie
algebra of  smooth  vector fields $X$ on ${\cal D}_f$; here and below we abbreviate $\partial_i\equiv \partial /\partial x^i$. Those vector fields $X\in\Xi$
such that $X(f^a)$ belong to $\C$ for all $a$, or equivalently
such that $X(h)$  belongs to $\C$ if $h$ does (i.e. vanishes when restricted to $M$) make up a Lie subalgebra $\Xi_\C$, which is also a $\X$-bimodule; 
those such that $X(h)$  belongs to $\C$ 
for all $h\in\X$ make up a smaller
Lie subalgebra $\Xi_{\C\C}$, which is actually an ideal in  $\Xi_\C$ and itself a  $\X$-bimodule. 
It decomposes as $\Xi_{\C\C}=\bigoplus_{a=1}^k f^a\Xi$.
The Lie algebra $\XiM$ of vector fields tangent to $M$ can be identified with
that of derivations of $\XM$, namely with the quotient
\be
\XiM:=\Xi_\C/\Xi_{\C\C}\equiv\big\{\, [X]:=X+\Xi_{\C\C} \:\: |\:\: X\in\Xi_\C\big\}.                                     \label{quotient'}
\ee

If $f^a(x)$ are polynomial functions fulfilling suitable irreducibility conditions and we set $\X=\mbox{Pol}^\bullet(\RR^n)$, the $*$-algebra  of complex-valued  polynomial
 functions on $\RR^n$ (instead of $\X=C^\infty({\cal D}_f)$), then again the $*$-algebra $\XM$ of complex-valued polynomial  functions on $M$
can be expressed as the quotient  $\XM=\X/\C$, where $\C\subset\X$ is the ideal of polynomial functions 
vanishing on $M$,  $\Xi:=\{X=X^i\partial_i \:\:\: |\:\:\: X^i\in \X\}$ is the Lie
algebra of   polynomial vector fields $X$ on $\RR^n$, etc.
$\C$ can be decomposed again in the form (\ref{DecoCeq}), with $\X=\mbox{Pol}^\bullet(\RR^n)$ \cite{FioFraWebquadrics}.

Often one is interested in noncommutative {\it deformations} of differential geometry on a manifold, i.e. in {\it families} of NCGs depending on a formal parameter  $\nu$ and reducing to the original one if we formally set $\nu=0$.
Deformation quantization  \cite{BayFlaFroLicSte,Kon97} provides a general framework  to deform $\X$ into a noncommutative algebra $\X_\star$  over  $\CC[[\nu]]$ 
(the ring of formal power series in $\nu$ with coefficients in $\CC$): as a module over  
$\CC[[\nu]]$
$\X_\star$  coincides with $\X[[\nu]]$, but the commutative pointwise product $\alpha \beta$ of $\alpha,\beta\in\X$ ($\CC[[\nu]]$-bilinearly extended to $\X[[\nu]]$)
is deformed into a possibly noncommutative (but still associative) product,
\be
\alpha\star \beta=\alpha \beta+\sum\nolimits_{l=1}^\infty \nu^l B_l(\alpha,\beta),
\ee
where $B_l$ are suitable bidifferential operators of degree $l$ at most.
We wish to deform $\XM$ into a noncommutative algebra $\XM_\star$ in the form of a quotient
\be
\XM_\star:=\X_\star/\C_\star\equiv\big\{\, [\alpha]:=\alpha+\C_\star \:\: |\:\: \alpha\in\X_\star\big\} \label{quotientstar},
\ee
with 
$\C_\star$ a two-sided ideal of  $\X_\star$.
To make also $\XM_\star=\XM[[\nu]]$ hold as an equality of 
$\CC[[\nu]]$-modules we  require that $\C_\star= \C[[\nu]]$, 
i.e. that \ $c\star \alpha,\, \alpha\star c\in \C[[\nu]]$  for all $\alpha\in\X$, $c\in\C$, \  so that \ $(\alpha+c)\star(\alpha'+c')-\alpha\star \alpha'\in \C[[\nu]]$  for all $\alpha,\alpha'\in \X[[\nu]]$ and $c,c'\in\C[[\nu]]$; \ 
as a result, taking the quotient and deforming the product commute: \  $(\X/\C)_\star=\X_\star/\C_\star$. \
This is fulfilled if\footnote{In fact, for
all $c\equiv \sum_{a=1}^kf^a c^a\in \C$ ($c^a\in\X$) (\ref{cond1}) implies 
$c=\sum_{a=1}^kf^a \star c^a$, so that for all $\alpha\in\X$, by the associativity of $\star$, 
$ c\star \alpha=(\sum_{a=1}^kf^a \star c^a)\star \alpha=\sum_{a=1}^kf^a \star (c^a\star \alpha)
\stackrel {(\ref{cond1})}{=}\sum_{a=1}^kf^a (c^a\star \alpha)\in \C[[\nu]]$; and similarly for $\alpha\star c$. 
Note that it is not sufficient to require that \ $\alpha\star f^a\!-\!\alpha f^a$, $f^a\star \alpha\!-\!f^a\alpha$, 
or equivalently  $B_l(\alpha,f^a), B_l(f^a,\alpha)$, \ belong to $\C$ to obtain the same results. As a more general condition ensuring $c\star \alpha,\, \alpha\star c\in \C[[\nu]]$ one could require that for all $a=1,..,k$ and  $\alpha\in\X$ the product  $f^a \alpha=\alpha f^a$ can be expressed as a combination of $\star$-products: \ $f^a \alpha=f^b\star A^a_b(\alpha)=A'{}^a_b(\alpha)\star f^b$.}
\be
\alpha\star f^a=\alpha f^a=f^a\star \alpha\qquad\Leftrightarrow\qquad B_l(\alpha,f^a)=0=B_l(f^a,\alpha)\qquad \forall l\in\mathbb{N}
\label{cond1}
\ee 
for all $\alpha\in\X$, $a=1,..,k$ [(\ref{cond1}) implies that the $f^a$ are central in $\X_\star$,  again].

In  \cite{Drinfeld1983} Drinfel'd has introduced  a general deformation quantization procedure of 
universal enveloping algebras $U\g$ (seen as Hopf algebras) of Lie groups  $G$
and of their module algebras, based on {\it twisting}; a {\it twist} is a suitable
element (a 2-cocycle, see section \ref{TwistSym})
\be
\F=\1\ot\1+\sum_{l=1}^\infty \nu^l \sum_{I_l} \F^{I_l}_{1} \otimes\F^{I_l}_{2} \in U\g\otimes U\g [[\nu]]                                   \label{twist}
\ee
(here $\otimes=\otimes_{\CC[[\nu]]}$).
It acts on the tensor product of any two $U\g$-modules or module algebras, 
in particular 
algebras of functions on any $G$-manifold,
including some 
symplectic\footnote{However this quantization procedure does not apply to every Poisson manifold:
there are several symplectic manifolds, e.g. the symplectic $2$-sphere and the
symplectic Riemann surfaces of genus $g>1$, which do not admit a $\star$-product induced by
a Drinfel'd twist (c.f. \cite{Thomas2016,FrancescoThomas2017}). Nevertheless, if one
is not taking into account the Poisson structure, every $G$-manifold (i.e. 
smooth manifolds $G$ acts on) can be quantized
via the above approach.} manifolds \cite{Aschieri2008}.
In \cite{Aschieri2006}   the authors consider the Lie algebra $\g=\XiM$ of 
smooth vector fields on a generic smooth manifold $M$ (this is the  Lie algebra of the infinite-dimensional Lie group of 
diffeomorphisms of $M$) and  the $U\XiM$-module algebra 
$\XM=C^\infty(M)$; $ \F_{1}^{I_l},\F_{2}^{I_l}$ seen as differential operators acting on $\XM$ have order  $l$ at most and no zero-order term. 
The deformed product reads
\be
\alpha\star \beta:=\alpha \beta+  \sum_{l=1}^\infty \nu^l \sum_{I_l}\bF_{1}^{I_l}   (\alpha) \:\:\bF_{2}^{I_l} (\beta)\,,
 \label{starprod}
\ee
where  \ $\bF=\1\ot\1+\sum_{l=1}^\infty \nu^l \sum_{I_l} \bF_{1}^{I_l}\otimes\bF_{2}^{I_l}$ \ is the inverse of the twist. In the sequel we will abbreviate 
\ $\F=\F_1 \otimes\F_2$, \ $\bF=\bF_1 \otimes\bF_2$ \
(Sweedler notation with suppressed summation symbols). 
In the presence of several copies of $\F$ we write $\F_1\ot\F_2$ and $\F'_1\ot\F'_2$ etc., in order
to distinguish the summations.
Actually, Ref. \cite{Aschieri2006} twists not only $U\XiM$ into a new Hopf algebra $U\XiM^\f$
and $\XM$ into a $U\XiM^\f$-module algebra $\XM_\star$,  but also the
 $U\XiM$-equivariant $\XM$-bimodule of differential forms on $M$, their tensor powers,
the Lie derivative, and the geometry on $M$ (metric, connection, curvature, torsion,...) - if present -, into 
deformed counterparts. 

Here, as  in  \cite{Masson1995}, we shall take the algebraic characterization (\ref{quotient}-\ref{quotient'}) as the starting point  for defining submanifolds in NCG, but use a twist-deformed differential calculus on it\footnote{The derivation-based approach to differential calculi of Dubois-Violette and Michor 
\cite{DVM1996}, which was used in \cite{Masson1995}, does not encompass several differential calculi (e.g. quantum group covariant ones), or requires algebra extensions to succed (see e.g. \cite{CerFioMad01}). 
The approaches to the differential calculus  \`a la  Connes \cite{Connes}
and Woronowicz \cite{Woronowicz1989} 
(which include  the one  considered here) are  more general:
the bimodule of noncommutative differential 1-forms is the primary object  whereby the whole calculus  can be derived by imposing the Leibniz rule and nilpotency of the exterior derivative. As a result, the dual module consists of noncommutative vector fields which are no longer derivations.}. 
\be
\Xi_t:=\{X\in\Xi \:\:\: |\:\:\: X (f^1)= 0,\;...,\:X (f^k) =0  \}\subset\Xi_\C           \label{defeqXis}
\ee 
is the Lie subalgebra of vector fields  tangent to {\it all} submanifolds $M_c$ (because they fulfill $X(f^a_c)=0$ for all $c\in f({\cal D}_f)$) at all points; it  is  an  $\X$-bimodule,
as well. 
The key observation is that, 
applying this deformation procedure to \ $\X=C^\infty({\cal D}_f)$ \ with a 
twist \ $\F\in U\Xi_t\otimes U\Xi_t [[\nu]]$, \ we 
satisfy (\ref{cond1}) and therefore obtain   a deformation $\X_\star$ of $\X$ such that 
\ $\X^{\scriptscriptstyle M_c}_\star\!=\!\X^{\scriptscriptstyle M_c}[[\nu]]=\X_\star/\C_{\star}^c$,
\ for all $c\in f({\cal D}_f)$; \ moreover, $\Xi_{\scriptscriptstyle M_c\star}\!=\!\Xi_{\scriptscriptstyle M_c}[[\nu]]=\Xi_{\C^c\star}/\Xi_{\C\C^c\star}$, \ see section \ref{TwistDiffGeomSubman}. 
In other words, we obtain a noncommutative deformation, in the sense of deformation quantization
 and in the form of quotients as in (\ref{quotient}-\ref{quotient'}), of the $k$-parameter family of embedded manifolds $M_c\subset\RR^n$. 
Actually, for every $X\in\Xi_\C$ there is an element in the equivalence
class $[X]$ that belongs to $\Xi_t$, namely its tangent projection $X_t$ (see Proposition \ref{prop05}). 
$\X_\star,\Xis,...$ are $U\Xi^\f$-equivariant,
while $\X^{\scriptscriptstyle M_c}_\star,\XiM{}_{\star},\Xi_{t\star},...$ are $U\Xi_t^\f$-equivariant.
If $\F$ is unitary or real, then $U\Xi^\f$ and $\X_\star$ admit $*$-structures (involutions), making them a Hopf $*$-algebra and a $U\Xi^\f$-module $*$-algebra,
respectively; thereby $U\Xi_t^\f$ is a Hopf $*$-subalgebra and $\X^{\scriptscriptstyle M_c}_\star,\Xi_{t\star},...$ are a $U\Xi_t^\f$-module
$*$-algebra and $U\Xi_t^\f$-equivariant Lie $*$-algebras, respectively.
By the same procedure one can obtain noncommutative deformations of differential 
geometry on submanifolds $N\subset M\subset\RR^n$.

The plan of the paper will be as follows. 

In section \ref{Preli} we present  preliminary material, first on twisting (section \ref{TwistAlgStruc}), then 
on its application \cite{Aschieri2006,AschieriCastellani2009,Aschieri2009} to the differential geometry on a generic manifold  (section \ref{TwistedNCGonM}).

In section \ref{TwistDiffGeomSubman} we deal with twist deformations of embedded manifolds $M\subset\RR^n$ in the smooth context.
In section \ref{DiffGeomSubman} we pave the way for these deformations recalling or deriving basic facts about
differential geometry on a submanifold $M$ of $\RR^n$, i.e. how the Cartan calculus and any connection, metric, etc., on $\RR^n$ induces corresponding data on $M$,
how to concretely build bases of the bimodules $\Xi_t,\Xip$ of tangent and normal vectors fields (i.e. sections in the tangent and normal bundle),
the corresponding projections $\Pt,\Pp$, etc.
In section \ref{SecTwistingVF} we first show that the
whole twisted Cartan calculus on  $\X_\star$ is projected to the one on $\XM_\star$, in the same way as for its undeformed counterpart, and that projection commutes with twisting,
for all twists \ $\F\in U\Xi_t\otimes U\Xi_t [[\nu]]$. Then we show that the same can be done
for: i) a connection $\nabla$, using a twist $\F\in U\mathfrak{e}\otimes U\mathfrak{e} [[\nu]]$, where $\mathfrak{e}$ is the corresponding {\it equivariance Lie algebra} (a Lie subalgebra of $\Xi_t$); ii)
the metric, and the associated Levi-Civita connection, using a twist $\F\in U\k\otimes U\k [[\nu]]$, where $\k\subseteq\mathfrak{e}$ is the Lie subalgebra of the corresponding {\it Killing vector fields}. Under the latter assumptions one can build a twisted
version not only of the first, but also of the second fundamental form, and prove
a twisted version of Gauss theorem.  Twisted $\Xi_t,\Xip,\Pt,\Pp$ stay essentially undeformed; we find suitable  $\k$-invariant bases for them.
Here we limit ourselves to developing  (pseudo)Riemannian geometry
for our physical interests, but other geometric
structures (say projective, affine, conformal,...) could be explored as well.
To build concrete examples of twisted submanifolds 
one can look for $M\subset\RR^n$ such that $\Xi_t$ contains a  finite-dimensional Lie subalgebra $\g$, because the simplest known Drinfel'd twists are 
based on such a $\g$; a nontrivial $\g$ surely exists if $M$ is symmetric under some Lie group.

In particular, one can apply \cite{FioFraWebquadrics}
 this procedure  to algebraic submanifolds $M\subset \RR^n$, e.g.
quadrics (i.e. level sets of a polynomial function $f(x)=0$ of degree 2); for the latter there exists a Lie subalgebra $\g$ (of dimension at least 2) of both $\Xi_t$ and 
 the Lie algebra ${\sf aff}(n)$ of the affine group ${\sf Aff}(\RR^n)=\RR^n\cocross GL(n)$ of $\RR^n$. If we choose a twist 
$\F\in U\g\otimes U\g[[\nu]]$ all the structures
can be formulated in terms of generators
and $\star$-polynomial relations. More precisely,
the algebra  $\X=\mathrm{Pol}^\bullet(\RR^n)$
of polynomial functions (with complex coefficients) in the set of Cartesian coordinates $x^1,...,x^n$ is deformed so that 
every $\star$-polynomial of degree $k$ in $x$ equals an ordinary polynomial of the same degree in $x$, and vice versa. 
The same occurs with the $\X_\star$-bimodules  
and algebras $\Omega^\bullet_\star$ of differential forms, that of differential operators, etc. In \cite{FioFraWebquadrics} the authors discuss in detail deformations of all families
of quadric surfaces embedded in $\RR^3$ that are induced by twists of
the abelian \cite{Reshetikhin1990} or Jordanian \cite{Ogievetsky1992,Ohn1992} type.  
In section \ref{Examples} of the present work, as illustrations of the approach, we just present cocycle twist deformations  of elliptic (in particular, circular) cylinders 
(first family) as well as hyperboloids and cone (second family)  embedded in $\RR^3$; they are induced by unitary abelian  or Jordanian twists.
Endowing $\RR^3$ with the Euclidean (resp. Minkowski) 
metric gives the circular cylinders (resp. hyperboloids and cone)
a  Lie algebra $\mathfrak{k}$ 
of isometries of dimension at least 2; choosing a twist $\F\in U\k\otimes U\k [[\nu]]$
we thus find twisted (pseudo)Riemannian $M_c$ (with the metric given by the
twisted first fundamental form) that are symmetric  under 
the Hopf algebra $U\mathfrak{k}^{\f}$ (the ``quantum group of isometries"),
and the twisted Levi-Civita connection on all $M_c$ equals the projection of the twisted  Levi-Civita connection on $\RR^3$ (the exterior derivative), 
while the twisted curvature can be expressed 
in terms of the twisted second fundamental form through the twisted Gauss theorem.
Actually, the metric, Levi-Civita connection, intrinsic and extrinsic curvatures of the circular
cylinders and hyperboloids, as elements in the appropriate tensor spaces, remain undeformed; the twist enters only their action on twisted tensor products of vector fields.
The twisted hyperboloids  can be seen as twisted (anti-)de Sitter spaces $dS_2,AdS_2$.

In section \ref{Outlook} we summarize  our results, add further remarks, mention 
possible mathematical developments, physical applications, issues
worth further investigations.

\section{Preliminaries}
\label{Preli}

\subsection{Twisted algebraic structures}
\label{TwistAlgStruc}

\subsubsection{Twisting Hopf algebras $H\!:=\!U\g$}
\label{TwistSym}

As known, the universal enveloping  algebra (UEA) $H\!:=\!U\g$ of the
Lie algebra $\g$ of any Lie group $G$ is a Hopf algebra. 
First, we briefly recall what this means. Let
$$
\ba{lll}
\varepsilon(\1)=1,\qquad \quad &\Delta(\1)=\1\ot\1,\qquad \quad & S(\1)=\1,\\[8pt]
\varepsilon(g)=0,\qquad \quad & \Delta(g)=g\ot\1+\1\ot g,\qquad
\quad &  S(g)=-g,\qquad \qquad \mbox{if }g\in\g; \ea
$$
$\varepsilon,\Delta$ are extended to all of $H$ as
algebra maps, $S$ as  an antialgebra map: \be \ba{lll}
\varepsilon:H\to\b{C},\quad\qquad  & \Delta:H\to H\ot H,\quad \qquad  
 & S:H\to H,\\[8pt]
\varepsilon(ab)=\varepsilon(a)\varepsilon(b),\quad \qquad  &
 \Delta(ab)=\Delta(a)\Delta(b),\quad\qquad  & S(ab)=S(b)S(a)
.\ea \label{deltaprop} 
\ee
The extensions of $\varepsilon,\Delta,S$ are unambiguous, because
$\varepsilon(g)=0$,
$\Delta\big([g,g']\big)=\big[\Delta(g),\Delta(g')\big]$,
$S\big([g,g']\big)=\big[S(g'),S(g)\big]$ if $g,g'\in\g$. The maps
$\varepsilon,\Delta,S$ are the abstract operations by which one
constructs the trivial representation, the tensor product of any two
representations and the contragredient of any representation,
respectively. $H\!=\!U\g$ equipped with  $\varepsilon,\Delta,S$
is a Hopf algebra; this means that a number of properties 
(see e.g. \cite{Chari1995,Majid2000,ES2010})
are fulfilled, in particular 
$(\Delta\otimes\mathrm{id})\circ\Delta=(\mathrm{id}\otimes\Delta)\circ\Delta$
(coassociativity),
$(\epsilon\otimes\mathrm{id})\circ\Delta=\mathrm{id}
=(\mathrm{id}\otimes\epsilon)\circ\Delta$
(counitality), 
$\mu\circ(S\otimes\mathrm{id})\circ\Delta
=\eta\circ\epsilon
=\mu\circ(\mathrm{id}\otimes S)\circ\Delta$
(antipode property)
[$\mu: H\otimes H\to H$ denotes the product in $H$, $\mu(a\otimes b)=ab$,
and $\eta:\CC\to H$ is defined by $\eta(\alpha)=\alpha\1$].
$H$ is cocommutative, i.e.  $\tau\circ\Delta=\Delta$, where $\tau$ is the flip
operator: $\tau(a\otimes b)=b \otimes a$.

If $G$ is a real form of a Lie group then there exists also a $*$-structure
on $H\!=\!U\g$, i.e. an involution $*:H\to H$ such that   for all $a,b\in H$
and $\alpha,\beta\in\CC$
\bea
\1^*=\1, \qquad (\alpha a+\beta b)^*= \bar\alpha a^*+\bar\beta b^*, 
\qquad (ab)^*=b^*a^*, \label{*prop1} \\[6pt]
\varepsilon(a^*)=[\varepsilon(a)]^*\qquad
  \Delta(a^*)=[\Delta(a)]^{*\ot *},\qquad  
S\left\{\left[S(a^*)\right]^*\right\}=a. \label{*prop2}
\eea
$H$ equipped with  $*,\varepsilon,\Delta,S$
is a Hopf $*$-algebra.

\medskip
Secondly, we recall how to deform a Hopf algebra using
a {\it twist} \cite{Drinfeld1983} (see also
\cite{Tak90,Chari1995}). Let $\hH=H[[\nu]]$.
Given a twist, i.e. an element $\F=\1\ot\1+\mathcal{O}(\nu)\in(H\ot H)[[\nu]]$
fulfilling 
\bea &&
(\epsilon\ot\id)\F=
(\id\ot\epsilon)\F=\1,                 \label{twistcond}\\[8pt]
&&(\F\ot\1)[(\Delta\ot\id)(\F)]=(\1\ot\F)[(\id\ot\Delta)(\F)]=:\F^3,
                                           \label{cocycle}
\eea 
we shall call $H_s\!\subseteq\!H$ the smallest Hopf
subalgebra such that $\F\!\in\! (H_s\ot H_s)[[\nu]]$, and 
\be
\beta:=\F_1 S\left(\F_2\right)
\in H_s\qquad\Rightarrow\qquad \beta^{-1}=S\left(\bF_1\right)\bF_2.    \label{defbeta} 
\ee
Extending the product, $\Delta,\varepsilon,S$ linearly to the formal
power series in $\nu$ and setting 
\be 
\Delta_\f(a) :=\F
\Delta(a) \bF, \qquad  S_\f(a):=\beta \, S(a)\beta^{-1}
,\qquad \R:=\F_{21}\bF,\label{inter-2}
 \ee
one finds that the analogs of conditions (\ref{deltaprop}), as well as
analogs of the coassociativity, counitality and antipode property are
satisfied and therefore $H^\f=(\hH,\Delta_\f,\varepsilon,S_\f)$ is a
Hopf algebra deformation of $(H,\Delta , \varepsilon, S)$. While the
latter was cocommutative,  $H^\f$ is triangular
noncocommutative (or quasi-cocommutative), i.e.
$\tau\!\circ\!\Delta_\f(a)\!=\!\R\Delta_\f(a)\bR$, where
$\bR\!=\!\R_{\scriptscriptstyle 21}$ is the inverse of the so-called
{\it universal R-matrix} or {\it triangular structure} $\R$.
Correspondingly, $\Delta_\f ,
S_\f$ replace $\Delta, S$ in  the tensor
product of any two representations and the contragradient of any
representation, respectively.
Drinfel'd has shown \cite{Drinfeld1983} that any triangular deformation of
the Hopf algebra $H$ can be obtained in this way (up to
isomorphisms).

To obtain a new Hopf $*$-algebra we need some further assumption on
the twist. Without loss of generality $\nu$ can be assumed real.
If $\F$ is {\it real} (i.e. $\F^{*\ot
*}=(S\ot S)[\F_{21}]$), then $\beta^*=\beta$ and
$\R^{*\ot*}=(\beta\ot\beta)^{-1}\bR(\beta\ot\beta)
=(\beta\ot\beta)\bR(\beta\ot\beta)^{-1}$.
Introducing the new $*$-structure
\be
g^{*_\f}:=\beta g^*\beta^{-1}
\ee
makes $(H^\f,*_\f)$ into a triangular Hopf $*$-algebra, i.e. also
(\ref{*prop1}-\ref{*prop2}) and $\R^{*_\f\ot *_\f}=\bR$ are satisfied. 
$*_\f$ can be transformed back to $*$ by the Hopf algebra
isomorphism  (\ref{eq20}), which transforms the product of $\hat H$ 
into the $\star$-product induced by $\F$. 
Another possibility is that $\F$ is {\it unitary} (i.e. $\F^{*\ot *}=\bF$).
Then \ $\beta^*=S(\beta^{-1})$, \ $\R^{*\ot*}\!=\!\bR$, \ and 
 \ $(H^\f,*)$ \ itself is a triangular Hopf $*$-algebra.

Eq. (\ref{cocycle}), (\ref{inter-2}) imply the
generalized intertwining relation
$\Delta_\f^{(n)}(a)\!=\!\F^n\Delta^{(n)}(a)(\F^n)^{-1}$
 for the iterated coproduct. By definition
$$
\Delta_\f^{(n)}: \hH\to \hH^{\ot n},\qquad\Delta^{(n)}: H[[\nu]]\to
(H)^{\ot n}[[\nu]],\qquad\F^n\in (H_s)^{\ot n}[[\nu]]
$$
reduce to $\Delta_\f,\Delta,\F$ for $n=2$, whereas for $n>2$ they
can be defined recursively as \be \ba{l}
\Delta_\f^{(n\!+\!1)}=(\id^{\ot^{n\!-\!1}}\ot\Delta_\f)
\circ\Delta_\f^{(n)},\qquad\Delta^{(n\!+\!1)}=
(\id^{\ot{(n\!-\!1)}}\ot\Delta)\circ\Delta^{(n)},\\[8pt]
\F^{n\!+\!1}=(\1^{\ot{(n\!-\!1)}}\ot\F)[(\id^{\ot{(n\!-\!1)}}
\ot\Delta)\F^n].                \label{iter-n} \ea \ee
The result for
$\Delta_\f^{(n)},\F^n$ are the same if in definitions (\ref{iter-n})
we iterate the coproduct on a different sequence of tensor factors
[coassociativity of $\Delta_\f$; this follows from the coassociativity of
$\Delta$ and the cocycle condition (\ref{cocycle})];
for instance, for $n\!=\!3$ this amounts to (\ref{cocycle}) and
$\Delta_\f^{(3)}\!=\!(\Delta_\f\ot\id)\!\circ\!\Delta_\f$. For
any $a\in H[[h]]= \hH$ we shall use the Sweedler notations
(summations are understood)
$$
\Delta^{(n)}(a)=
a_{(1)} \otimes a_{(2)} \otimes ...\otimes a_{(n)},\qquad\qquad
\Delta_\f^{(n)}(a)=
a_{\widehat{(1)}} \otimes a_{\widehat{(2)}} \otimes ...\otimes a_{\widehat{(n)}}.
$$

\medskip  
We consider the following examples of twists:
\begin{enumerate}
\item[i.)]
Let $n\in\mathbb{N}$, \
$P:=\sum_{i=1}^ne_i\otimes e_{n+i}\in\mathfrak{g}\otimes\mathfrak{g}$, \
with pairwise commuting elements \ $e_1,...,e_{2n}\in\mathfrak{g}$.
Then
$$
\mathcal{F}=\exp(i\nu P)\in( U\g \otimes U\g )[[\nu]]
$$
is a Drinfel'd twist on $ U\g $ (\cite{Reshetikhin1990}).
We refer to it as an
abelian (or {\it Reshetikhin}) twist on $ U\g $. 
It is unitary if $P^{^*\otimes^*}=P$; 
this is e.g. the case if the $e_i$ are anti-Hermitian or Hermitian.
It is immediate to check that the twist with $P$ replaced by $P'=\frac 12\sum_{i=1}^n
(e_i\otimes e_{n+i}-e_{n+i}\otimes e_i)$ is both unitary and real, 
leads to the same $\R$ and makes $\beta=\1$,
whence $S_\f=S$, and the $*$-structure remains undeformed also for 
$H$-$*$-modules and module algebras, see (\ref{star'}).

\item[ii.)]
Let $H,E\in\mathfrak{g}$ be elements of a Lie algebra such that
$
[H,E]=2E.
$
Then 
$$
\mathcal{F}
=\exp\left[\frac{1}{2}H\otimes\log(\1+i\nu E)\right]
\in( U\g \otimes U\g )[[\nu]]
$$
defines a Drinfel'd twist, called Jordanian twist
\cite{Ogievetsky1992,Ohn1992}.
If $H$ and $E$ are anti-Hermitian, $\mathcal{F}$ is unitary.
More sophisticated twists can be obtained using this 
as a prototype \cite{Borowiec2005,Pachol2017,BorMelMelPac18}.
\end{enumerate}

There are numerous other examples of Drinfel'd twists. We refer to \cite{Jonas2017}
for an explicit (recursive) construction of twists on UEA via a Fedosov method and a classification
(of equivalence classes) of twists in terms of the Chevalley-Eilenberg cohomology of
the Lie algebra.

\subsubsection{Twisting $H$-modules and $H$-module algebras}
\label{TwistMod}

We recall that, given  a Hopf ($*$-)algebra $H$ over $\b{C}$,
a left $H$-module $(\M,\trc)$ is a vector space $\M$ over $\b{C}$
equipped with a left action, i.e. a $\b{C}$-bilinear map
$(g,a)\!\in\! H\!\times\!\M\mapsto g\trc a\!\in\!\M$
such that (\ref{leibniz})$_1$ and $1\trc a=a$ hold.
An element $a\in\M$ of a left $H$-module is said to be
$H$-invariant if $g\trc a=\epsilon(g)a$ for all $g\in H$.
Equipped also with
an antilinear involution $*$  fulfilling (\ref{leibniz})$_2$,
$(\M,\trc,*)$ is a left $H$-$*$-module. A
left  $H$-module ($*$-)algebra is a ($*$-)algebra $\A$ over $\b{C}$
equipped with a left $H$-($*$-)module structure,
such that (\ref{leibniz})$_3$ 
and $g\trc 1=\epsilon(g)1$ hold:
\be
(gg')\!\trc\! a=g
\!\trc\! (g'\! \!\trc\! a)\!,\qquad (g\!\trc\! a)^*\!=[S(g)]^*\!\trc\! a^*\!,
\qquad g\!\trc\! (ab)=\!
\left(g_{(1)}\!\trc\! a\right)\! \left(g_{(2)}\!\trc\! b\right).\qquad        \label{leibniz}
\ee
If $g\!\in\!\g$, formula (\ref{leibniz})$_3$ becomes the Leibniz rule.
An $\A$-bimodule $\M$ for a left $H$-module algebra $\A$ is said to be an $H$-equivariant
$\A$-bimodule if $\M$ is a left $H$-module such that
\begin{equation}
    g\trc(a\cdot s\cdot b)
    =(g_{(1)}\trc a)\cdot(g_{(2)}\trc s)\cdot(g_{(3)}\trc b)
\end{equation}
for all $g\in H$, $a,b\in\A$ and $s\in\M$, where we denoted the $\A$-module actions on $\M$ by
$\cdot$. If in addition, $\A$ is a left $H$-module $*$-algebra and there is a $*$-involution on $\M$,
we call $\M$ an $H$-equivariant $\A$-$*$-bimodule if $(a\cdot s\cdot b)^*=b^*\cdot s^*\cdot a^*$.
We remark that any $\A$-($*$-)subbimodule of an $H$-equivariant $\A$-($*$-)bimodule is an
$H$-equivariant $\A$-($*$-)bimodule if it is closed under the Hopf algebra action.
A map $\phi\colon\M\rightarrow\M'$ between left $H$-modules is said to be $H$-equivariant
if it commutes with the Hopf algebra actions, i.e. $g\trc\phi(s)=\phi(g\trc s)$ for all
$g\in H$ and $s\in\M$.
%
%
%
Extending the action $\trc$ $\b{C}[[\nu]]$-bilinearly one can make
any $H$-module $(\M,\trc)$ into an $\hH$-module $(\M[[\nu]],\trc)$.
If $\F$ is a real (resp. unitary) twist on $H$, the undeformed $*$-involution
(resp. the $*$-involution $*_\f$) structures $H^\f$ as a triangular
Hopf $*$-algebra.
If $\F$ is real (resp. unitary) and $(\M,\trc,*)$ is an $H$-$*$-module  then 
(see e.g. \cite{Fiore2010,Aschieri2006}) $(\M[[\nu]],\trc,*_\star)$ is an $H^\f$-$*$-module, where
\be
    a^{*_\star}:= a^*\qquad\qquad 
    \mbox{(resp. }  a^{*_\star}:=S(\beta)\trc a^* \mbox{)}.      \label{star'}
\ee

Given an $H$-module ($*$-)algebra $\A$ and choosing
$\M=\A$, the twist gives also a systematic way to make $\A[[\nu]]$
into an $H^\f$-module ($*$-)algebra $\A_\star$, by endowing it with
a new product, 
\be
a\star a':=
\left(\bF_1  \trc a\right) \left(\bF_2\trc a'\right),                                 \label{starprod}
\ee
the so called  $\star$-product. In fact, $\star$ is associative by (\ref{cocycle}),  fulfills \  $(a\!\star\!
a')^{*_\star}\!=\!a'{}^{*_\star}\!\!\star\! a^{*_\star}$ and
\be
\ba{l} g\trc (a\!\star\!a')=
\big(g_{\widehat{(1)}} \trc  a\big)  \star \big(g_{\widehat{(2)}} \trc a'\big). 
\ea 
\ee
If $aa'=\pm a'a$, 
i.e. $a,a'$ (anti)commute, then
\be
a'\star a \: = \: \pm\: (\R_2\trc a)\star(\R_1\trc a').                     \label{braid1}
\ee
Consequently, twists leading to the same $\R$ [e.g. the abelian twists
$\exp{(i\nu P)},\exp{(i\nu P')}$ of the previous section] lead to the same commutation relations in $\A_\star$.
More generally, for any $H$-equivariant $\A$-($*$-)bimodule $\M$ of a left $H$-module ($*$-)algebra
$\A$, the twisted module actions
\begin{equation}
    a\star s
    =(\overline{\F}_1\trc a)\cdot(\overline{\F}_2\trc s)
    \quad\text{ and }\quad
    s\star a
    =(\overline{\F}_1\trc s)\cdot(\overline{\F}_2\trc a),
\end{equation}
where $a\in\A$ and $s\in\M$, structure $\M$ as an $H^\f$-equivariant $\A_\star$-($*$-)bimodule
$\M_\star$ (with $*$-involution (\ref{star'}) on $\M_\star$). We refer to 
\cite{Aschieri2014,Fiore2010} for proofs of the previous statements.

\medskip
Given two $H$-modules $(\M,\trc)$, $(\N,\trc)$,
the tensor product $(\M \ot \N,\trc)$
is an $H$-module if we define $g\trc (a\ot b):= 
\left(g_{(1)}\trc a\right) \ot\left( g_{(2)}\trc b\right)$.
As above, this is extended to an $H^\f$-($*$-)module
$(\M\otimes\N[[\nu]],\trc)$.
Introducing the ``$\star$-tensor product''
\cite{Aschieri2006} 
\be 
\ba{l} (a\ot_\star
b):=\bF (\trc\ot\trc)(a\ot b)\equiv 
\bF_1\trc a \otimes \bF_2\trc b 
\ea \label{startensor}
\ee 
(an invertible
endomorphism, i.e. a change of basis, of $\M\!\ot\!\N[[\nu]]$), we find
\be
g\trc(a\ot_\star b)\:=\: \:
g_{\widehat{(1)}}\trc a \:\otimes_\star\:  g_{\widehat{(2)}}\trc b . \label{startransf}
\ee
Given two $H$-module
($*$-)algebras $\A,\B$, this applies in particular to $\M=\A$,
$\N=\B$. The tensor ($*$-)algebra $\A\ot\B$ [whose product is
defined by $(a\ot b)(a'\ot b')=(aa'\ot bb')$] is an $H$-module
($*$-)algebra under the action $\trc$. By introducing the
$\star$-product (\ref{starprod}) $\A\ot\B$  is deformed into an
$H^\f$-module ($*$-)algebra $(\A\ot\B)_\star$, with
 product and $*$-structure  related to those of
$\A_\star,\B_\star$ by 
\bea
&& (a\ot_\star b)\star (a'\ot_\star b')=
a\star (\R_2\trc a')\: \ot_\star\, (\R_1\trc b) \star b',\label{braid}\\[8pt]
&&\!\! \ba{ll}
(a\ot_\star b)^{*}=
\R_2\trc a^{*} \:\:\ot_\star \:  \R_1\trc b^{*} \qquad &\mbox{if }\F \mbox{ is real},\\[6pt]
(a\ot_\star b)^{*_\star}=
\R_2\trc a^{*_\star} \:\:\ot_\star \:  \R_1\trc b^{*_\star}\qquad &\mbox{if }\F \mbox{ is unitary},
\ea
\eea
where $\R_1\ot \R_2$ (again a summation is understood) is the
decomposition of $\R$ in $H^\f\ot H^\f$. 
 From (\ref{braid}) we recognize that
$(\A\ot\B)_{\star}$ is isomorphic to
the {\it braided  tensor product (algebra)}
\cite{Majid2000,Chari1995} of $\A_\star$ with $\B_\star$; here the
braiding is involutive and therefore spurious, as
$\bR=\R_{21}$. So $(\A\ot \B)_{\star}$ encodes both the
usual $\star$-product within  $\A,\B$ and the $\star$-tensor product
(or braided tensor product) between the two.  (On the contrary, setting
\ $ (a\otimes b):=
\F_1\trc a \otimes_\star \F_2\trc b $ \
`unbraids' the braided tensor product, cf. \cite{FioSteWes03}). 
By (\ref{cocycle}) the
$\star$-tensor product is associative, and the previous results hold
also for iterated $\star$-tensor products.

\medskip
The algebra $(H[[\nu]],\star)$ itself is an  $H^\f$-module algebra, and one can build a triangular Hopf algebra $H_\star=(H[[\nu]],\star,\eta,\Delta_\star,\epsilon,S_\star,\R_\star)$  isomorphic  to $H^\f=(H[[\nu]],\mu,\eta,\Delta_\f,\epsilon,S_\f,\R)$, with isomorphism $D:H_\star\to H^\f$ and inverse  given by 
\cite{Majid1994,Aschieri2006} (cf. also \cite{Fio98JMP,Fiore2010})
\be\label{eq20}
D(\xi):=(\bF_1\trc \xi)\bF_2=\F_1\,\xi\, S\!\left(\F_2\right)\,\beta^{-1},\qquad
D^{-1}(\phi)=\bF_1 \,\phi\,\beta\, S\!\left(\bF_2\right).
\ee
Namely, \
$D(\xi\star \xi')=D(\xi)D(\xi')$, \ and \
$\Delta_\star,S_\star,\R_\star$ \ are related to \ $\Delta_\f,S_\f,\R$ \ by the relations
\be \label{HF-HstarREL}
\Delta_\star=(D^{-1}\otimes D^{-1})\circ  \Delta_\f \circ D,
\qquad S_\star=  D^{-1}\circ S_\f \circ D,
\qquad \R_\star= (D^{-1}\otimes D^{-1})(\R).
\ee
One can think of $D$ also as a change of generators within $H[[\nu]]$.

If $\F$ is real  then $(H[[\nu]],\star,*)$ is a left $H^\f$-module $*$-algebra, and
$D:(H_\star,*)\to(H^\f,*_\f)$ is an isomorphism of  triangular
Hopf $*$-algebras (cf. \cite{Majid2000}~Proposition~2.3.7, \cite{Aschieri2006}).

If $\F$ is  unitary  then $(H[[\nu]],\star,*_\star)$ is a left $H^\f$-module $*$-algebra,
and \ $D:(H_\star,*_\star)\to(H^\f,*)$ \ is an isomorphism of  triangular
Hopf $*$-algebras, see Proposition \ref{IsomorHopf} 
in the Appendix.

\subsection{Twisted differential geometry}
\label{TwistedNCGonM}

Ref. \cite{Aschieri2006} applies the previous machinery to $H=U\Xi$, where $\Xi$ is the
Lie algebra of the Lie group of diffeomorphisms of $M$, and $\A$ is the algebra 
$\X=C^\infty(M)$ of smooth functions on 
$M$, or more generally an $\X$-bimodule of tensor fields on $M$.
Tensor fields of rank $(p,r)$ ($p,r\in\NN_0$)  on $M$ 
can be described as elements in the tensor product 
\be
\TT^{p,r}:=\underbrace{\Omega\otimes\ldots\otimes\Omega}_{\mbox{$p$ times}}\otimes\underbrace{\Xi\otimes\ldots\otimes\Xi}_{\mbox{$r$ times}}
\ee
of the $\X$-bimodules   $\Omega\equiv\Omega^1$, $\Xi$ of differential 1-forms and vector fields on $M$, respectively. 
Here and below $\otimes$  stands for $\otimes_{\X}$ (rather than $\otimes_{\CC}$), 
namely $T\otimes f T'= T f\otimes T'$ for all $f\in\X$. 
We set  $\TT^{0,0}:=\X$.
 The tensor product is  associative; to avoid the need of reorderings we multiply 
$T\in \TT^{p,r}$ by 1-form tensor factors
only from the left if $r>0$, by vector field tensor factors only from the right if $p>0$.
The tensor product between a function $f\in\X\equiv\TT^{0,0}$ and another tensor field is as usual not 
explicitly written. All $\TT^{p,r}$ are $\X$-bimodules, e.g. 
$f( T\otimes T')=(f T)\otimes T'$, $( T\otimes T')f= T\otimes( T'f)$.
$\TT:=\bigoplus_{p,r\in\NN_0}\TT^{p,r}$ (endowed with the product $\otimes$)
is a $U\Xi$-module algebra:  the action $\trc: U\Xi\otimes \TT\rightarrow \TT\,$  is uniquely determined by
 \ ${\bf 1}_H\trc T=T$ and  
\be
X\trc T=\L_X(T), \qquad X\in\Xi, \label{LieDer}
\ee
where $\L$ is the Lie derivative.
It fulfills the Leibniz rule $g\trc ( T\otimes T')=g_{(1)}\trc T\otimes g_{(2)}\trc T'$.

By setting $\A=\TT$  we can apply the results of section \ref{TwistAlgStruc}, in particular define a
deformed tensor algebra $\TT_\star$ with associative $\star$-tensor product defined by eq. (\ref{startensor}).
$\TT_\star$ is a  $U\Xi^\f$-module algebra. \ All $\TT^{h,r}_\star$  are $\X_\star$-bimodules.   \ In $\TT_\star$ we have in particular
\bea
 T\otimes_\star h\star T'= T\star h\otimes_\star T'~,\qquad
h\star( T\otimes_\star T')=(h\star T)\otimes_\star T'~.
\eea 
The first formula shows that $\otimes_\star$ is actually    $\otimes_{\X_\star}$, the tensor product
over $\X_\star$. While  the usual product of a tensor field $ T$ with  a function $h$ from 
the left and from the right coincide\footnote{If $ T\in\Xi$ then the product $ T\cdot h$
of $T$ with $h$ from the right is  the
vector field  that on a function $g$ gives $( T \!\cdot\! h)(g)= T(g)h$. In section
\ref{QQM} we shall denote it by $T\ltlc  h$, so as to distinguish it from
the   operator $ T h= T(h) +h T$.}, 
in general this no longer occurs with the $\star$-product.

In a chart $U$ with coordinates $x^\mu$
any vector field $X$ can be  expressed in the $\partial_\mu$ basis as
$X=X^\mu\partial_\mu$. It  can be also uniquely expressed as
$X=X_\star^\mu\star\partial_\mu$, \
where $X^\mu_\star$ are functions defined on $U$. The same
occurs if $\{\partial_\mu\}$ is replaced by a more general 
(not necessarily holonomic or $\nu$-independent) frame $\{e_a\}$:
$X=X_\star^a\star e_a$.
Similarly,
every 1-form $\omega$ can be uniquely written as
$\omega=\omega_\mu  dx^\mu=\omega^\star_\mu\star dx^\mu$, 
with $\omega_\mu,\omega^\star_\mu$ functions defined on $U$, and where $\{dx^\mu\}$ is the usual 
dual frame of the vector field frame $\{\partial_\mu\}$.
More generally, in $U$ every tensor field $ T^{p,q}\in\TT^{p,q}$ can be uniquely  
written using functions $ T_{\star~\mu_1\ldots\mu_p}^{~\lambda_1\ldots\lambda_q}$
 defined  on $U$ as 
\be 
 T^{p,q}= T_{\star~\mu_1\ldots\mu_p}^{~\lambda_1\ldots\lambda_q}
\star dx^{\mu_1}\otimes_\star\ldots \otimes_\star dx^{\mu_p}
\otimes_\star\partial_{\lambda_1}\otimes_\star
\ldots\otimes_\star \partial_{\lambda_q}.
\label{79}
\ee 


Let us twist the algebra  $\A=\Omega^\bullet=\oplus_p\Omega^p$ 
of differential forms. 
We denote by $\Omega^\bullet_\star:= (\Omega^\bullet,\wedge_\star)$ 
the $\CC[[\nu]]$-module of forms equipped with the $\star$-deformed wedge product  
\be \label{formsfromthm}
\omega\wedge_\star\omega':=(\bF_1\trc\omega )\wedge (\bF_2\trc\omega')= \omega\otimes_\star\omega'
-\R_2\trc \omega'\otimes_\star\R_1\trc\omega~.
\ee  
This  can be seen as the tensor subspace of totally 
$\star$-antisymmetric (contravariant) tensor fields. The degree of the top form 
stays undeformed. Below we  drop the symbols $\wedge,\wedge_\star$.

The usual exterior derivative 
$d:\Om^\bullet\rightarrow \Om^{\bullet+1}$ satisfies the graded 
Leibniz rule $d(\alpha_p\star \beta)=d\alpha_p\star \beta+(-1)^p\alpha_p\star
d\beta $ and is therefore also the $\star $-exterior derivative. This is 
so, because the exterior derivative commutes with the Lie derivative,
i.e. with the Hopf algebra action.

One can endow \cite{Aschieri2006} the module underlying the algebra $U\Xi^\f\simeq U\Xi[[\nu]]$ itself with the $\star$-product; the new algebra
$U\Xi{}_\star$ endowed by suitable coproduct, counit, antipode becomes a Hopf algebra isomorphic to $U\Xi^\f$, whereby it is manifest
that the above differential calculus is bicovariant in the sense of Woronowicz \cite{Woronowicz1989}.
$\Xi$ is closed under the $\star$-Lie bracket
\be
[X,Y]_\star:=X\star Y-(\R_2\trc Y)\star (\R_1\trc X)=\left[\bF_1\trc X,\bF_2\trc Y\right]=
\L_{\bbf_1\trc X}\left(\bF_2\trc Y\right).
\ee

The action $\L^\star_X$  of $U\Xi{}_\star$ on $T$ ($\star$-Lie derivative) is defined by
\be
\L^\star_X(T)=(\bF_1\trc X)\trc\left[\bF_2\trc T\right], \qquad X\in U\Xi. \label{starLieDer}
\ee

\subsubsection{$\star$-Pairing between 1-forms and vector fields, twisted Cartan calculus}\label{pairingformvect}

 Denoting $\la~,~\ra$  the commutative pairing between vector fields and 1-forms,
the $\star$-pairing  is defined as \ $\la~,~\ra_\star:=\la~,~\ra\circ \bF(\trc\otimes\trc): \,\Xi_\star\otimes_\CC  \Omega_\star  \mapsto  \X_\star$, \ namely
\bea\label{lerest}
(X,\omega)~&\mapsto &\la X,\omega\ra_\star
:=\left\la\bF_1\trc X,\bF_2\trc \omega\right\ra~.
\eea
The $\star$-pairing  is  actually a map $\la~,~\ra_\star : 
\Xi_\star\otimes_\star \Omega_\star \mapsto  \X_\star$, as it satisfies the  $\X_\star$-linearity properties  
\be\label{linst}
\ba{c}
\la X, h\star\omega \ra_\star= \la X  \star h,\omega\ra_\star,
\\[6pt]
\la h\star X,\omega\star k\ra_\star=h\star\la X,\omega\ra_\star\star k~.
\ea
\ee
with $h,k\in\X_\star$. From \ $\la X,dh\ra=X(h)$, \ $g\trc dh=d(g\trc h)$  \ and (\ref{lerest}) it follows that
\be\label{defXst}
\la X,dh\ra_\star
=\left(\bF_1\trc X\right)  \left(\bF_2\trc h\right)=: X_\star(h).
\ee
$X_\star$ is a twisted derivation, i.e. fulfills the deformed Leibniz rule
\be 
 X_\star(h\star h')=X_\star(h)\star h'+[\R_2\trc h]\star [(\R_1\trc X)(h')];          \label{twistLeibnizrule}
\ee
the quickest way to prove the latter is by  the Leibniz rule for $d$ and (\ref{defXst}), (\ref{linst}), (\ref{braid1}).
The compatibility $X\trc \la Y,\omega\ra =\left\la X_{(1)}\trc Y,X_{(2)}\trc\omega\right\ra$
of $\la~,~\ra$ with the Lie derivative (which expresses the diffeomorphism-invariance of the pairing) implies
\be\label{lieagainder}
X\trc \la Y,\omega\ra_\star=\big\la X_{\widehat{(1)}}\trc Y,
X_{\widehat{(2)}}\trc \omega\big\ra_\star~.
\ee
In the commutative case,  for any local moving frame 
(vielbein) $\{e_i\}$ we can build  a dual frame of 1-forms $\{\omega^i\}$, \
$\la e_i,\omega^j\ra=\delta^j_i$, \  and conversely;
in particular $\la \partial_\mu,dx^\lambda\ra=\delta_\mu^\lambda$.
The  exterior derivative decomposes as $d=\omega^i  e_i$.
It is the same in the noncommutative case. The $\star$-dual frame $\{ \theta^i\}$   of $\{ e_i\}$,
\be
\la  e_i,\theta^j\ra_\star=\delta^j_i~,                 \label{dualframes}
\ee 
can be obtained from $\{\omega^i\}$ via a $\X_\star$-linear transformation that is the identity
at zero order in $\nu$ \cite{Aschieri2006}, and the  exterior derivative decomposes
also as $d=\theta^i\star e_{i\star}$.
Using the $\star$-pairing  we can associate to any $1$-form
$\omega$ the left $\X_\star$-linear map $\la~\,,\omega\ra_\star:\Xis\rightarrow \X_\star$, and to any vector field
$X$ the right $\X_\star$-linear map $\la X,~\,\ra_\star: \Omega_\star \rightarrow  \Omega_\star $. 
The maps \  $\mathrm{i}_X:=\la X,~\,\ra_\star$, 
 $\mathrm{i}_\omega:=\la~\,,\omega\ra_\star$ are the simplest twisted insertions (interior products) of a vector field in a form and of a 1-form in 
a multivector field, respectively.
Using the exterior derivative and the twisted insertion, Lie bracket, and Lie derivative 
one can develop  \cite{TWeber2019} a twisted Cartan calculus in complete analogy with the  usual one (see also the thesis \cite{Thomas2018} for more details). As one can extend the commutative pairing to higher tensor powers setting
\be
\la X_p\otimes...\otimes X_1,\omega_1\otimes...\otimes\omega_p\otimes \tau\ra
:=\la X_1,\omega_1\ra \,...\, \la X_p,\omega_p\ra\: \tau \label{ExtPairing}
\ee
for all $X_i\in\Xi$, $\omega_i\in\Omega$, so can one extend \ $\la X,~\,\ra_\star$ \ to the corresponding twisted tensor powers
using the same formula (\ref{lerest}). Properties (\ref{linst}), (\ref{lieagainder}) are preserved.
There is a $\star$-pairing
$\langle~,~\rangle'_\star\colon\Omega_\star\ot_\star\Xi_\star\to\X_\star$
with forms on the left and vector fields on the right. It is related to the
previous $\star$-pairing via $\langle\omega, X\rangle'_\star
=\langle\bR_1\trc X,\bR_2\trc\omega\rangle_\star$ for all $\omega\in\Omega_\star$
and $X\in\Xi_\star$. It is left and right $\X_\star$-linear and satisfies
$\langle\omega\star h,X\rangle'_\star=\langle\omega,h\star X\rangle'_\star$ for
all $h\in\X_\star$. As in the case of $\langle~,~\rangle_\star$ there is an
extension of $\langle~,~\rangle'_\star$ to higher twisted tensor powers.

\subsubsection{Covariant derivative, torsion, curvature, metric}\label{covderivative}

In \cite{AschieriCastellani2009,Aschieri2006}
a  {\it twisted  covariant derivative} (or, synonymously, {\it twisted  connection}) $\nabla^{\f}$ is defined as a collection of maps \ $\nabla^{\f}\!:\!\Xi_\star\! \otimes_{\CC[[\nu]]}\!  \TT^{p,q}_\star \!\to\!\TT^{p,q}_\star$, $X\!\otimes\! T\mapsto \nabla^{\f}_X T$ ($p,q\!\in\!\NN_0$) fulfilling the properties
\bea
&&\nabla^{\f}_X h=\L^\star_X(h),\label{eq04}\\[.35cm]
&&\nabla_{h\star X+h'\star Y}^{\f} T=h\star\nabla_X^{\f} T
+h'\star\nabla_Y^{\f} T,\label{eq03}\\[.35cm]
&&\nabla_{X}^{\f} (T+T')=\nabla_X^{\f} T+\nabla_X^{\f} T'~, \label{eq00}
\eea
\begin{equation}\label{Leibddsg}
\begin{split}
    \nabla^{\f}_X( T\ots T')
    =&
    (\overline{\mathcal{R}}_1\trc\nabla^{\f}_{\overline{\r}'_2\trc X}
    (\overline{\mathcal{R}}''_2\trc T))
    \ots((\overline{\mathcal{R}}_2\overline{\mathcal{R}}'_1
    \overline{\mathcal{R}}''_1)\trc T')
+(\overline{\mathcal{R}}_1\trc T)\ots
    \nabla^{\f}_{\overline{\r}_2\trc X} T',
\end{split}
\end{equation}
\begin{equation}\label{compleredd}
\begin{split}
    \nabla^{\f}_X\la Y,\omega\ra_\star 
    =&\la\overline{\mathcal{R}}_1\trc(\nabla^{\f}_{\overline{\r}'_2\trc X}
    (\overline{\mathcal{R}}''_2\trc Y)),
    (\overline{\mathcal{R}}_2\overline{\mathcal{R}}'_1\overline{\mathcal{R}}''_1)
    \trc\omega\ra_\star 
+\la\overline{\mathcal{R}}_1\trc Y,
    \nabla^{\f}_{\overline{\r}_2\trc X} \omega\ra_\star
\end{split}
\end{equation}
for all \ $X,Y\in\Xi_\star\equiv\TT^{0,1}_\star$, \ $h,h'\in \X_\star\equiv\TT^{0,0}_\star$, \ $ T,T'\in\TT_\star$. \ 
On functions the twisted covariant and Lie derivatives along $X$ coincide, by (\ref{eq04}).
Eq. (\ref{compleredd}) amounts to the compatibility of the action of $\nabla^{\f}_X$ on 1-forms with the pairing of the latter with vector fields.
$\nabla^{\f}$ is left $\X_\star$-linear in the
first argument, by (\ref{eq03}); it is only $\CC[[\nu]]$-linear in the second
argument, by  (\ref{eq00}) and  (\ref{Leibddsg}) with $T=c\in\CC[[\nu]]\subset\X_\star$.
Relation (\ref{Leibddsg}) for  \ $T=h\in\X_\star$  \
 becomes  [by (\ref{eq04})]  the deformed Leibniz rule  
\be
\nabla_X^{\f}(h\star T')
\,=\,\mathcal{L}_X^{\star}(h)\star T'+
(\bR_1\trc h)\star\nabla^{\f}_{\br_2\trc X}T';           \label{ddsDuhv}
\ee
this holds in particular on vector fields $T'=Y$.
Actually, the knowledge of $\nabla^{\f}$ just for $(p,q)=(0,0)$ (i.e. on functions)
and  $(p,q)=(0,1)$ (i.e. on vector fields),
determines its unique extension to all the $(p,q)\in\NN_0^2$:
eq. (\ref{compleredd}) determines the action of $\nabla^{\f}_X$ on 1-forms, while
(\ref{Leibddsg}) allows to extend  $\nabla^{\f}_X$ recursively to all the $\TT^{p,q}_\star$'s,
which consist of combinations of tensor products of 1-forms and vector fields. 

The torsion $\tT^\f_\star $ and the curvature $\rR^\f_\star $ associated to
a connection $\nabla^{\f}$ are  left $\X_\star$-linear maps  \ \
$\tT^\f_\star :\Xi_\star\otimes_\star\Xi_\star\to\Xi_\star$,  \ \
$\rR^\f_\star : \Xi_\star\otimes_\star\Xi_\star\otimes_\star\Xi_\star\to\Xi_\star$ 
\ \ defined by
\bea
\tT^\f_\star(X,Y)&:=&\nabla_{X}^{\f}Y-\nabla_{\r_2\trc Y }^{\f}(\R_1\trc X)
-[X,Y]_{\star}~,\\[.2cm]
\rR^\f_\star(X,Y,Z)&:=&\nabla_{X}^{\f}\nabla_{Y}^{\f}Z-
\nabla_{\r_2\trc Y }^{\f}\nabla_{\r_1\trc X}^{\f}Z-\nabla^{\f}_{[X,Y]_\star} Z~,
\eea
for all $X,Y,Z\in\Xi_\star$.
They fulfill \ $\tT^\f_\star(X,h\star Y)= \tT^\f_\star(X\star h ,Y)$ \ (and similarly for the curvature),  and
the antisymmetry property
\bea\label{Tantysymm}
\tT^\f_\star(X,Y)&=&-\tT^\f_\star(\R_2\trc Y ,\R_1\trc X)~,\\
\rR^\f_\star(X,Y,Z)&=&-\rR^\f_\star(\R_2\trc Y,\R_1\trc X,Z)~.
\label{Rantysymm}
\eea
Thus, one can regard torsion and curvature as elements of the following $\star$-tensor spaces
\be
\tT^\f\in\Omega_\star\wedge_\star\Omega_\star\ots\Xi_\star,\qquad\rR^\f\in\Omega_\star\ots\Omega_\star\wedge_\star\Omega_\star\ots\Xi_\star,     \label{starTRspaces}
\ee
acting on vector fields through the twisted pairing (\ref{lerest}) applied to higher tensor powers, see (\ref{ExtPairing}). We omit the $\star$ in the subscript of
the elements (\ref{starTRspaces}) in order to distinguish them from the corresponding
maps. In other words, for all $X,Y,Z\in\Xi_\star=\Xi[[\nu]]$
\begin{equation}
    \tT^\f_\star(X,Y)=\langle X\ot_\star Y,\tT^\f\rangle_\star
   \quad \text{and}\quad
    \rR^\f_\star(X,Y,Z)=\langle X\ot_\star Y\ot_\star Z,\rR^\f\rangle_\star.
\end{equation}

A {\it metric} is defined as a non-degenerate element 
$\gm$ in the module $\Oms\ots\Oms=\Omega\otimes\Omega[[\nu]]$
that is symmetric, i.e. invariant under the flip $ \tau(\alpha\otimes\beta):=\beta\otimes\alpha$.
Clearly the two decompositions $\gm=\gm^\alpha\otimes \gm_\alpha=\gm^A\ots \gm_A$ 
(sum over repeated indices) 
are related by $\gm^\alpha\otimes \gm_\alpha=\bF_1\trc\gm^A\otimes\bF_2 \trc\gm_A$.
$\gm$ determines the map $\gm_\star:\Xi_\star\ots \Xi_\star\to\X_\star$
defined by
\bea
\gm_\star(X,Y):=\left\la X,\left\la Y,\gm^A\right\ra_\star \star \gm_A\right\ra_\star;    \label{twistedmetric1}
\eea
this fulfills \ $\gm_\star(h\star X,Y)=h\star\gm_\star( X,Y)$
(left $\X_\star$-linearity in $X$) and \ $\gm_\star(X\star  h,Y)=\gm_\star( X,h\star Y)$. \
The twisted {\it Levi-Civita} (LC) connection $\nabla^{\f}$ is a connection fulfilling  $\tT^\f=0$
and $\nabla^{\f}_X\gm=0$ for all $X\in\Xi_\star$, or equivalently, for all $Y,Z\in\Xi_\star$
\begin{equation}
    \mathcal{L}^\star_X \big[\gm_\star(Y,Z) \big]
    =\gm_\star\!\left(\overline{\mathcal{R}}_1\trc(
    \nabla^{\f}_{\overline{\r}'_2\trc X}
    (\overline{\mathcal{R}}''_2\trc Y)),
    (\overline{\mathcal{R}}_2\overline{\mathcal{R}}'_1\overline{\mathcal{R}}''_1)
    \trc Z\right)\\
    +\gm_\star\!\left(\overline{\mathcal{R}}_1\trc Y,
    \nabla^{\f}_{\overline{\r}_2\trc X}Z\right).
\label{twistedMetricCov}
\end{equation}
If a twisted LC connection exists, it is unique by 
\cite{AschieriCastellani2009}~Theorem~5. For equivariant metrics there is an
existence and uniqueness theorem (c.f. \cite{TWeber2019}~Lemma~3.12) of an
equivariant twisted LC connection.
If $\F=\1\ot\1$  (whereby $\star$ becomes the ordinary product, and $\R=\1\ot\1$), the above formulae  reduce to the notions and properties of ordinary connection, torsion, curvature, metric, etc. In particular we recover the characterization
of a LC connection: 
\bea
\ba{ll}
\mbox{torsion-free, i.e. }\qquad  &\tT:=\nabla_{\!\!X} Y- \nabla_{\!Y}  X-[X,Y]=0 
\qquad \forall X,Y\in\Xi,\\[6pt]
\mbox{metric-compatible, i.e.}\:\:   & {\cal L}_Z \big[\gm( X,\!Y)\big]\!-\gm( \nabla_{\!\!Z}  X,\!Y)-
\gm( X,\!\nabla_{\!\!Z} Y)=0\quad \forall X,\!Y,\!Z\in\Xi.  
\ea \quad  \label{LC}
\eea

\medskip
 In the commutative case the Ricci tensor is a contraction of the curvature
tensor, $\ric_{jk}={\rR_{ijk}}^i$.
The twisted  Ricci map is defined as the following contraction of the
curvature:
\be\label{defofric1}
\ric^{\f}_\star(X,Y):=\la \theta^i, \rR^{\f}_\star(e_i,X,Y)\ra_\star'~,
\ee
where sum over $i$ and (\ref{dualframes}) are understood. \
$\la~,~\ra_\star'$ is 
a contraction between forms on the {\it left} and
vector fields on the {\it right}, see Section~\ref{pairingformvect}.
Recall that it is defined by the  pairing 
\bea
\la \omega\,,\,X\ra_\star'&=&\left\la\bF_1\trc \omega\,,\,\bF_2\trc X\right\ra=\left\la\R_2\trc X\,,\,\R_1\trc \omega\right\ra_\star~
\eea
and has  the $\X_\star$-linearity properties
\be\label{linearitypaiop}
\la h\star \omega, X\star k\ra_\star'=h\star\la \omega, X\ra_\star'\star k~,\qquad
\la \omega, h\star X \ra_\star'=
\la \omega\star h, X\ra_\star'~.
\ee
The twisted  Ricci map is well defined because
(\ref{defofric1}) is independent of the choice of the frame $\{e_i\}$ 
(and of the dual frame $\{\theta^i\}$),
and because it is defined as a  $\X_\star$-linear map: 
\bea
\ric^{\f}_\star(h\star X,Y)=h\star \ric^{\f}_\star(X,Y)~,\qquad
\ric^{\f}_\star(X,h\star Y)= \ric^{\f}_\star(X\star h ,Y)~.\label{rictrlinea3}
\eea 
Evaluating the Ricci map on the dual metric
$\gm^{-1}=\gm^{-1\alpha}\otimes\gm^{-1}{}_{\alpha}=\gm^{-1A}\ots\gm^{-1}{}_{A}$ yields the Ricci scalar:
\bea
\mathfrak{R}^{\f}=\ric^{\f}_\star\left(\gm^{-1A},\gm^{-1}{}_{A}\right). \label{RicciScalar}
\eea

\bigskip
We now show how to construct nontrivial twisted deformations $\nabla^{\f}$ of $\nabla$. First, we need some preliminary result in ordinary differential geometry.
It is easy to check that 
\be
\mathfrak{e}:=\left\{g\in\Xi \quad |\quad \left[g,\nabla_{\!X}Y\right]=\nabla_{[g,X]}Y+\nabla_{\!X}[g,Y]\quad\forall X,Y\in\Xi\right\}        \label{equivLieAlg}
\ee
is a Lie subalgebra  of $\Xi$; we shall name it the {\it  equivariance Lie algebra of $\nabla$}. 
It follows
\be\label{eq01}
g\trc \left(\nabla_{\!X}Y\right)=\nabla_{\!g_{(1)}\trc X}\left(g_{(2)}\trc Y\right)
\qquad\forall g\in U\mathfrak{e}.
\ee
\begin{prop}
Given a connection $\nabla$ on $M$ and the associated equivariance Lie algebra 
$\mathfrak{e}$, 
setting
\be
\nabla^{\f}_X T:=\nabla_{\bbf_1\trc X}(\bF_2\trc  T)              \label{twistedNabla}
\ee
with \ $\mathcal{F}\in U\mathfrak{e} \otimes U\mathfrak{e}[[\nu]]$  \ defines a twisted connection along $X\in\Xi_\star$.
It is $U\mathfrak{e}^\f$-equivariant, i.e.
\begin{equation}
    g\trc\nabla^\f_XY
    =\nabla^\f_{g_{\widehat{(1)}}\trc X}(g_{\widehat{(2)}}\trc Y)
\end{equation}
and satisfies the additional deformed  Leibniz rule (with functions multiplying
from the right)
\begin{equation}
    \nabla^\f_X(T\star h)
    =(\nabla^\f_{X}T)\star h
    +(\overline{\R}_1\trc T)\star(\mathcal{L}^\star_{\overline{\r}_{2}\trc X}(h))
\label{rightLeibniz}
\end{equation}
for all $g\in U\mathfrak{e}^\f$, $h\in\X_\star$, $T\in\mathcal{T}_\star$ and $X,Y\in\Xi_\star$.
Furthermore, eqs. (\ref{Leibddsg}), (\ref{compleredd}) boil down to
\bea
&& \nabla^{\f}_X( T\ots T')= \nabla^{\f}_X  T \ots   T' + \bR_1\trc T\ots 
\nabla^{\f}_{\br_2\trc X} T',                    \label{Leibddsg'} \\[.35cm]
&&\nabla^{\f}_X\la Y,\omega\ra_\star =\la\nabla^{\f}_X(Y),\omega\ra_\star + 
\la\bR_1\trc Y,\nabla^{\f}_{\br_2 \trc X} \omega\ra_\star.      \label{compleredd'}
\eea
\label{TwistedNabla}
\end{prop} 
Of course, nontrivial deformations of this kind are possible only if 
$\mathfrak{e}\neq\{0\}$.

\noindent
We recall that   $Z\!\in\!\Xi$ is a {\it Killing vector field} of a (pseudo)Riemannian manifold $(M,\gm)$ if 
\be
{\cal L}_Z \big[\gm( X,Y)\big]-\gm\big( [Z, X],Y\big)-\gm\big( X,[Z, Y]\big)=0\qquad \forall X,Y\in\Xi,     \label{Killing0}
\ee
or equivalently  if \ $\gm\big( \nabla_{\!\!X} Z,Y\big)+\gm\big( X,\nabla_{\!Y} Z\big)=0$\footnote{The latter condition is obtained taking the difference of (\ref{LC})$_2$,  (\ref{Killing0}),  using
the bilinearity of $\gm$ and (\ref{LC})$_1$}. \
The Killing vector fields close a Lie subalgebra $\k\subset\Xi$; this is the Lie algebra of the group of isometries of $(M,\gm)$ if $M$ is complete.

\begin{prop}\label{prop02} 
The Killing vector fields $\k\subset\Xi$ form a Lie subalgebra of the
equivariance Lie algebra $\mathfrak{e}$ of
the Levi-Civita connection $\nabla$ on a (pseudo)Riemannian manifold $(M,\gm)$. 
For all twists  $\mathcal{F}\in U\k\otimes U\k[[\nu]]$ the map $\gm_\star$ is also right 
$\X_\star$-linear  in the second argument and related to the
undeformed one $\gm:\Xi\otimes\Xi[[\nu]]\to \X[[\nu]]$, $\gm(X,Y):=\left\la X,\left\la Y,\gm^\alpha\right\ra  \gm_\alpha\right\ra$, by
\be
\gm_\star(X,Y)=\gm\left(\bF_1\trc X,\bF_2\trc Y\right),    \label{twistedmetric2}
\ee
and $\nabla^{\f}_X $ is the unique twisted Levi-Civita connection
corresponding to $\gm_\star$.
Torsion and curvature of the twisted Levi-Civita connection remain undeformed as elements of the  tensor spaces 
\bea
0=\tT^\f=\tT\in \Omega \wedge \Omega\otimes\Xi[[\nu]],\qquad\quad
\rR^\f=\rR\in\Omega\otimes\Omega \wedge \Omega\otimes\Xi[[\nu]]
\label{TRspaces}
\eea
and the associated maps $\tT^\f_\star,\rR^\f_\star$ are also right $\X_\star$-linear 
 in the last argument.
Eq. (\ref{twistedMetricCov}) boils down to
\be
\mathcal{L}^\star_X \big[\gm_\star(Y,Z) \big] =\gm_\star(\nabla^{\f}_XY,Z)
    +\gm_\star(\bR_1\trc Y,\nabla^{\f}_{\br_2\trc X}Z).
\label{twistedMetricCov'}
\ee
\label{TwistedLC} 
\end{prop}

The proofs of these propositions are in the appendix.
The existence and uniqueness of the twisted Levi-Civita connection
of Proposition~\ref{prop02} was proven in 
\cite{AschieriCastellani2009}~Theorem~6 and Theorem~7.
In NCG right function-linearity
of the curvature  in the last argument is in general not true, see e.g. \cite{DubMadMasMou96,FioMad00EPJ}.
In section \ref{Examples} we will find nontrivial $\k$ for suitably symmetric 
quadrics in $\RR^3$.

\section{Twisted smooth submanifolds of $\RR^n$}
\label{TwistDiffGeomSubman}

\subsection{Differential geometry of manifolds embedded in $\RR^n$}
\label{DiffGeomSubman}


We develop some theoretical tools for  the $(n\!-\!k)$-dimensional submanifolds 
$M_c\subseteq{\cal D}_f\subseteq\RR^n$ 
defined by equations (\ref{DefIdeal}).
Recall definitions (\ref{quotient'}),  (\ref{defeqXis}). We can identify $\XiM\subset\Xi$ with the
Lie subalgebra of smooth vector fields tangent to $M$ at all points, and
$\Xi_t\subset\Xi$ with the Lie subalgebra of smooth vector fields tangent to {\it all} 
$M_c$ ($c\in f({\cal D}_f)$) at all points, 
because $X(f^a)=0$ implies  $X(f_c^a)=0$. Decomposing $X=X^i\partial_i$
and abbreviating $f_i^a:=\partial_i (f^a)$,  $X\in\Xi_t$ amounts to $X^if^a_i=0$ for all $a=1,...,k$;
as the Jacobian matrix $J=(f^a_i)$ has by assumption rank $k$, dim$(\Xi_t)=n-k=:m$.
Henceforth  $\Omega$ will stand for the $\X$-bimodule of differential 1-forms on 
${\cal D}_f$, i.e. the dual one of $\Xi$.
We also define a $\X$-subbimodule $\Omega_{\scriptscriptstyle{\perp}}\subset \Omega$ of 1-forms by
\be
 \Omega_{\scriptscriptstyle{\perp}}:=\{\omega\in\Omega \:|\: \langle \Xi_t,\omega\rangle=0\}.
 \label{defOmegadiamond}
\ee 
Let $\Xi_t^\bullet=\bigwedge\!\!^\bullet\Xi_t$,  $\XiM^\bullet=\bigwedge\!\!^\bullet\XiM$,  
 $\Omega^\bullet_{\scriptscriptstyle{\perp}}=\bigwedge\!\!^\bullet\Omega_{\scriptscriptstyle{\perp}}$ be the corresponding exterior algebras. 
\begin{prop}
The $\Xi_t,\Xi_\C,\XiM$  defined in (\ref{quotient'}),  (\ref{defeqXis})
are Lie subalgebras of $\Xi$. \ $\Xi_{\C\C}$ is an ideal in $\Xi_\C$. \
$\Xi_t^\bullet$,  $\XiM^\bullet$, $\Omega^\bullet_{\scriptscriptstyle{\perp}}$ are
$U\Xi_t$-equivariant $\X$-bimodules. \ $\Omega_{\scriptscriptstyle{\perp}}$ can be explicitly decomposed as
\vskip-.3cm
\be
\Omega_{\scriptscriptstyle{\perp}}=\bigoplus_{a=1}^k\X\mathrm{d}f^a=\bigoplus_{a=1}^k\mathrm{d}f^a\X.  \label{decoOmegadiamond}
\ee
\label{Propdiamond}
\end{prop}
\vskip-.5cm

\bp{} $\omega=\omega_adf^a$ implies
$\langle X,\omega\rangle=\langle X,df^a\rangle\omega_a=X(f^a)\omega_a=0$.
Conversely, in any basis $\{X_\alpha=X_\alpha^i\partial_i\}_{\alpha=1}^m$ of $\Xi_t$
the $n\times m$ matrix $(X_\alpha^i)$ has rank $m$ and  fulfills $f^a_iX_\alpha^i=0$;
decomposing $\omega=\omega_idx^i$,
$\langle \Xi_t,\omega\rangle=0$ amounts to $\langle X_\alpha,\omega\rangle=X_\alpha^i\omega_i=0$
for all $\alpha$, and this linear system of $m$ independent equations admits
only  solutions $\omega_i =f^a_i  \omega_a$, $\omega_a\in\X$, whence $\omega=\omega_adf^a$. 
This proves (\ref{decoOmegadiamond}).
For all  $X,W\in\Xi_t$,  $Y\in\Xi_\C$, $Z\in\Xi_{\C\C}$,  $a=1,...,k$, $h\in\X$,
$\omega\in\Omega_{\scriptscriptstyle{\perp}}$ we find  
\bea
\ba{lll}
(X\trc W)(f^a)=[X,W](f^a)=X\big(W(f^a)\big)-W\big(X(f^a)\big)=0
&\quad\Rightarrow\quad & X\trc W\in\Xi_t,\\[4pt]
(X\trc Y)(f^a)=[X,Y](f^a)=X\big(Y(f^a)\big)\!-\!Y\big(X(f^a)\big)=X\big(Y(f^a)\big)\!\in\!\C
&\quad\Rightarrow\quad &  X\trc Y\in\Xi_\C,\\[4pt]
(X\trc Z)(h)=[X,Z](h)=X\big(Z(h)\big)-Z\big(X(h)\big)=X\big(Z(h)\big)\in\C 
&\quad\Rightarrow\quad &  X\trc Z\in\Xi_{\C\C},
\ea\nn[1pt]
X\trc\omega= X(\omega_a)\, df^a+\omega_a \,X\trc\big(d f^a\big)
= X(\omega_a)\, df^a+\omega_a \,d[X(f^a)]= X(\omega_a)\, df^a\in\Omega_{\scriptscriptstyle{\perp}}. 
\qquad\qquad\quad\nonumber
\eea
This implies in turn  that   $g\trc W\in\Xi_t$,  $g\trc Y\in\Xi_\C$,
 $g\trc Z\in\Xi_{\C\C}$, $g\trc\omega\in\Omega_{\scriptscriptstyle{\perp}}$
  for all $g\in U\Xi_t$, so that $\Xi_t,\Xi_\C,\Xi_{\C\C},\XiM,\Omega_{\scriptscriptstyle{\perp}}$ 
are $U\Xi_t$-equivariant $\X$-bimodules, and also $\Xi_t^\bullet,\XiM^\bullet,\Omega^\bullet_{\scriptscriptstyle{\perp}}$ are. \ep

\subsubsection{Metric, Levi-Civita connection, intrinsic and extrinsic curvatures}
\label{Metric}

We now discuss  $\Xi_t,\Omega_{\scriptscriptstyle{\perp}}$ as addends
in the decomposition of $\Xi,\Omega$ with respect to a metric.
Consider a (non-degenerate) metric \ $\gm\equiv\gm^\alpha\ot\gm_\alpha\in\Omega\otimes\Omega$ on ${\cal D}_f$ (actually, the following discussion is valid on any smooth manifold)
and its dual \ $\gm^{-1}=\gm^{-1\alpha}\ot\gm^{-1}_\alpha\in\Xi\otimes\Xi$. \ We recall that
\begin{equation}
{\cal G}:\Xi\to\Omega,    \qquad X\mapsto \omega_X=\langle X,\gm^\alpha\rangle\gm_\alpha 
\end{equation}
is an isomorphism of $\X$-bimodules with inverse given by
$\omega\mapsto X_\omega=\langle\gm^{-1\alpha},\omega\rangle\gm^{-1}_\alpha $. In fact
$X_{\omega_X}=\langle\gm^{-1\alpha},\omega_X\rangle\gm^{-1}_\alpha
=\langle\gm^{-1\alpha},\langle X,\gm^\beta\rangle\gm_\beta\rangle\gm^{-1}_\alpha
=X$ for all $X\in\Xi$ and
$\omega_{X_\omega}=\langle X_\omega,\gm^\alpha\rangle\gm_\alpha
=\langle\langle\gm^{-1\beta},\omega\rangle\gm^{-1}_\beta,\gm^\alpha\rangle\gm_\alpha
=\omega$ for all $\omega\in\Omega$. It follows that for all  $Y\in \Xi$, $\alpha\in\Omega$,
\bea
\gm(Y,X)=\la Y,\omega\ra,\qquad\qquad \gm^{-1}(\omega,\alpha)=\la X,\alpha \ra,
\label{g&pairing}
\eea
whenever $\omega ={\cal G}(X)$, or equivalently $X={\cal G}^{-1}(\omega)$.
Let us now introduce the $\X$-subbimodules 
\bea
\Xip:=\{X\in\Xi \:\:\:|\:\:\: \gm(X,\Xi_t)=0\}, \qquad
\Omega_t:=\{\omega\in\Omega \:\:\: |\:\:\:  \gm^{-1}(\omega,\Omp)=0  \},
\eea
and let ${\cal D}_f'\subseteq{\cal D}_f$ be the open subset where the restriction
\begin{equation}\label{D_f'}
\gmp^{-1}:=\gm^{-1}|_{\Omp\otimes\Omp}\colon\Omp\otimes\Omp\to\X
\end{equation}
is non-degenerate. 
If $\gm$ is Riemannian, then \ ${\cal D}_f'={\cal D}_f$. \ For simplicity, henceforth
we shall denote the restrictions of \ \ $\Xi,\,\Xi_t,\,\Xip\,\Omega,\,\Omp,\,\Omega_t$ to \ ${\cal D}_f'$ \ by the same symbols, and  
 by $\k\subset\Xi_t$ the Lie subalgebra of Killing vector fields of $\gm$ that are also tangent to the submanifolds $M_c\subset {\cal D}_f'$.

\begin{prop} 
The Lie algebra $\Xi$ of 
smooth vector fields and the $\X$-bimodule $\Omega$ of 1-forms
on ${\cal D}_f'$ split into the direct sums of  $\X$-subbimodules 
\bea
\Xi=\Xi_t\oplus\Xip,\qquad\quad\Omega=\Omega_t\oplus\Omp,
\label{decoXi'}
\eea
orthogonal with respect to the metric $\gm$ and $\gm^{-1}$ respectively.
 $\Xi_t$ is a Lie subalgebra of $\Xi$. 
$\Omega_t$  is  orthogonal to  $\Xip$
with respect to the pairing: \
$ \Omega_t=\{\omega\in\Omega \:\:\: |\:\:\:  \la \Xip,\omega\ra=0  \} $. \ 
Also the  restrictions  of $\gm^{-1}$ to the {\rm tangent} forms and of $\gm$ to 
the {\rm tangent} and {\rm normal} vector fields
\be
\gm^{-1}_t:=\gm^{-1}|_{\Omega_t\ot\Omega_t}\colon\Omega_t\ot\Omega_t\to\X,  \qquad \gmp:=\gm|_{\Xip\otimes\Xip}:\Xip\otimes\Xip\to\X,
\ee
\be
    \gm_t:=\gm|_{\Xi_t\otimes\Xi_t}\colon\Xi_t\otimes\Xi_t\to\X \label{eq21}
\ee
are non-degenerate.
The orthogonal projections
\ $\Pp: \Xi\to \Xip$, \ $\Pt:\Xi\to \Xi_t$,  \ $\Pp:\Omega\to \Omp$, \ $\Pt:\Omega\to \Omega_t$ are
uniquely extended as projections to the bimodules of multivector fields and higher rank forms through the rules
$\Pp(\omega\omega')=\Pp(\omega) \Pp(\omega')$, $\Pt(\omega\omega')=\Pt(\omega) \Pt(\omega')$,...:
\bea
\Pp: \Omega^p\to \Omega^p_\perp, \quad \Pt:  \Omega^p\to \Omega^p_t, \quad
\Pp: \bigwedge\nolimits^p\Xi\to \bigwedge\nolimits^p\Xip, \quad \Pt: \bigwedge\nolimits^p\Xi\to \bigwedge\nolimits^p\Xi_t.
\eea
\ $\Xi_t,\Xip,\Omega_t,\Omp$, their exterior powers and the projections $\Pp,\Pt$  are $U\mathfrak{k}$-equivariant. 
\label{XiOmegaDeco}
\end{prop}

\bp{}
On ${\cal D}_f'$ one can build unique projections $\Pp:\omega\in\Omega\to \omega_{\scriptscriptstyle{\perp}}\in\Omp$, 
$\Pt=\id-\Pp:\omega\in\Omega\to \omega_t\in\Omega_t$, $\Pp:X\in\Xi\to X_{\scriptscriptstyle{\perp}}\in\Xip$, 
$\Pt=\id-\Pp:X\in\Xi\to X_t\in\Xi_t$ such that the decompositions (\ref{decoXi'}) hold, 
see section \ref{Bases}.
By Proposition \ref{Propdiamond},
$\Xi_t,\Omega_{\scriptscriptstyle{\perp}}$ are in particular   
$U\mathfrak{k}$-equivariant $\X$-subbimodules. 
Also $\Xip,\Omega_t$ are $U\mathfrak{k}$-equivariant $\X$-subbimodules,
by the   $U\mathfrak{k}$-equivariance and  $\X$-linearity in both arguments of 
$\gm(\cdot,\cdot)$ and  of $\la\cdot,\cdot\ra$:  
if $X\in\Xi_{\scriptscriptstyle{\perp}}$ then
$\gm(\xi\trc X,Y)=\xi_{(1)}\trc\gm(X,S(\xi_{(2)})\trc Y)=0$ for all $\xi\in U\mathfrak{k}$ and $Y\in\Xi_t$ by the $U\mathfrak{k}$-equivariance of
$\Xi_t$; hence 
$\xi\trc X\in\Xi_{\scriptscriptstyle{\perp}}$ and $\Xip$ is $U\mathfrak{k}$-equivariant. Similarly one shows that $\Omega_t$ 
is $U\mathfrak{k}$-equivariant. Consequently, also
$\Pt,\Pp$ acting on $\Xi,\Omega$,  as well as  their extensions  to 
$\bigwedge\nolimits^\bullet\Xi,\Omega^\bullet$, are
$\X$-linear in all arguments and  $U\mathfrak{k}$-equivariant; for instance, the $U\mathfrak{k}$-equivariance on $\Omega$ follows from
 $\mathrm{pr}_t(\xi\trc \omega)
=\mathrm{pr}_t(\xi\trc \omega_t+\xi\trc \omega_{\scriptscriptstyle{\perp}})
=\xi\trc \omega_t
=\xi\trc\mathrm{pr}_t(\omega)$ for all $\xi\in U\mathfrak{k}$ and 
$\omega=\omega_t+\omega_{\scriptscriptstyle{\perp}}\in\Omega$.
Now, note that by (\ref{g&pairing})$_1$ ${\cal G},{\cal G}^{-1}$ map 
$\Xip,\Omp$ into each other and  $\Xi_t,\Omega_t$ into each other. In fact,
\begin{itemize} 

\item $X\in \Xip$ implies $\la \Xi_t,\omega\ra=\gm(\Xi_t,X)=0$, whence
 $\omega\in\Omp$; and vice versa.

\item $X\in \Xi_t$ implies $\la \Xip,\omega\ra=\gm(\Xip,X)=0$, whence
 $\omega\in\Omega_t$; and vice versa.

\end{itemize}
Then, by (\ref{g&pairing})$_2$,
if $\alpha\in\Omega_t$, then  for all $X\in\Xip$  it is
$\la X,\alpha\ra=\gm^{-1}(\omega_X,\alpha)=0$, because
$\omega_X\in\Omp$;
conversely,  if $\la \Xip,\alpha\ra=0$, then 
for all $\omega\in\Omp$ it is
$\gm^{-1}(\omega,\alpha)=\la X_\omega,\alpha \ra=0$, because $X_\omega\in\Xip$.
So we have proved that  \
$ \Omega_t=\{\omega\in\Omega \:\:\: |\:\:\:  \la \Xip,\omega\ra=0 \}$. \
Next, let $\omega\in\Omega_t$; then $X_\omega\in\Xi_t$.
By (\ref{g&pairing})$_2$, $\gm_t^{-1}(\omega,\Omega_t)=0$ implies  $\la X_\omega,\Omega_t\ra=0$ and therefore also $\la X_\omega,\Omega \ra=0$, whence by  
the non-degeneracy of the pairing, $X_\omega=0$, and in turn $\omega=0$,
namely $\gm_t^{-1}$ is non-degenerate. 
Since  $\gm$ is non-degenerate, for all $X\in\Xi_t$ there is $Y\in\Xi$,
and hence also $Y_t\in\Xi_t$, such that $0\!\neq\!\gm(X,Y)\!=\!\gm(X,Y_t)$: \
$\gm_t$ is non-degenerate. Similarly one proves that also $\gmp$ is.
\ep

\noindent
{\bf Remarks:} \ i) The non-degeneracy of $\gmp^{-1}$ (or, equivalently, of  
$\gm_t^{-1}$) is not only sufficient, but also necessary
to ensure that $\Omp\cap\Omega_t=\{0\}$. In fact,
if $\gmp^{-1}$ is degenerate  there is a nonzero $\omega\in\Omp$
such that $0=\gmp^{-1}(\omega,\Omp)=\gm^{-1}(\omega,\Omp)$,
hence $\omega$ belongs to $\Omega_t$ as well. \
ii) Similarly, the non-degeneracy of $\gmp$ (or, equivalently, of  
$\gm_t$) is necessary for $\Xip\cap\Xi_t=\{0\}$. \
iii) While $\Xi_t$ is a Lie subalgebra of $\Xi$, in general $\Xi_{\scriptscriptstyle{\perp}}$ is not. \
iv) In general $\Xi_{\scriptscriptstyle{\perp}},\Omega_t$, and therefore also the orthogonal projections $\Pp,\Pt$, are not $U\Xi_t$-equivariant; for this reason in section \ref{TwistedMetric} we are able to deform (pseudo)Riemannian geometry only via twists based on $\k\subset\Xi_t$. \ v) We refer to elements of 
$\Xi_{\scriptscriptstyle{\perp}},\Omp $  and $\Omega_t$ as \textit{normal}  vector fields, \textit{normal} 1-forms and  \textit{tangent} 1-forms.

\medskip
As said, we  identify 
$\Xi_t\subset\Xi$ with the Lie subalgebra of smooth vector fields tangent to {\it all} $M_c$ 
($c\in f({\cal D}_f)$) at all points, because $X(f^a)=0$ implies  $X(f_c^a)=0$;
and $\XiM\subset\Xi$ defined in (\ref{quotient'}) with  the Lie subalgebra of smooth vector fields tangent to $M$ at all points. 
Similarly, we can identify 
$\Omega_t$ with  the subbimodule of $\Omega$ tangent to {\it all} $M_c$ 
($c\in f({\cal D}_f)$) at all points.  \
We find \ $\Omega_t\subset\Omega_\C:=\left\{\omega\in\Omega \:\:\: |\:\:\:  \la \Xip,\omega\ra\subset\C  
\right\}. 
$ \
Let \ $\Omega_{\C\C}:=\bigoplus_{a=1}^k f^a\Omega=\bigoplus_{a=1}^k \Omega f^a\subset\Omega_\C$. \
It fulfills \ $ \la \Xi,\Omega_{\C\C}\ra\subset\C$. \
We can identify the $\XM$-bimodule of 1-forms $\Omega_{\scriptscriptstyle M}$ on $M$ with the quotient
\be
\Omega_{\scriptscriptstyle M}=\Omega_\C/\Omega_{\C\C}=\left\{[\omega]=\omega+\Omega_{\C\C}\:\: |\:\: \omega\in\Omega_\C \right\}.
\ee

\begin{prop}\label{prop05}
For all $X\in\Xi_\C$, $\omega\in\Omega_\C$, the tangent projections $X_t\in\Xi_t$,
$\omega_t\in\Omega_t$  belong to $[X]\in\XiM$ and $[\omega]\in\Omega_{\scriptscriptstyle M}$ respectively; similarly for multivector fields and higher rank forms.
 \label{XCrepr}
\end{prop}
Consequently, we can represent every element of $\Xi_{\scriptscriptstyle M},\Omega_{\scriptscriptstyle M}$, or more generally $\Xi_{\scriptscriptstyle M_c},\Omega_{\scriptscriptstyle M_c}$, resp. by an element of $\Xi_t,\Omega_t$; etc. 
In the appendix we prove Proposition \ref{XCrepr}, as well as the relations
\bea
\ba{l}
\Omp=\left\{\omega\in\Omega  \:|\:   \la \Xi_t,\omega\ra=0 \right\},\qquad \Omp\subset\Omega_{\scriptscriptstyle \Box},\qquad\mbox{where}\\[6pt]
 \Omega_{\scriptscriptstyle \Box}\: :=\: \left\{\omega\in\Omega \:\: |\:\:   \la \Xi_\C,\omega\ra\in\C  
\right\}\: =\: \left\{\omega\in\Omega \:\:  |\:\:  \la \Xi_t,\omega\ra\in\C  \right\}.
\ea \label{extra}
\eea
We call the restriction $\gm_t$ in (\ref{eq21}) of the metric map $\gm$ 
\textit{first fundamental form} for the family of manifolds 
$M_c\subset{\cal D}_f'$, \ $c\!\in\!f\big({\cal D}_f'\big)$. It is
$\X$-linear in both arguments and further satisfies $\gm_t(X\cdot h,Y)=\gm_t(X,h\cdot Y)$ for all $X,Y\in\Xi_t$ and $h\in\X$ ({\it middle-linearity}). Since
$\gm_t$ is uniquely determined (via the pairing) by   the tangent projection
\ $\tilde\gm_t=(\mathrm{pr}_t\ot\mathrm{pr}_t)(\gm)\in\Omega_t\ot\Omega_t$ \ 
%
of the metric $\gm\in\Omega\otimes\Omega$,  
when there is no risk of confusion we will drop the tilde and with a slight abuse of notation denote $\tilde\gm_t$ by $\gm_t$. It is a symmetric element, i.e. $\tau(\gm_t)=\gm_t$.
The \textit{first fundamental form} (induced metric) {\it on $M$} is obtained by
the further projection $\X\to\XM$, which amounts to choosing the $c=0$
manifold $M$ out of the family. The same prescription will hold for
the the Levi-Civita connection, curvature, etc., on $M$.
Applying the decomposition of $\Xi$ in tangent and normal vector fields to the
restriction of the Levi-Civita connection
\begin{equation}
    \nabla|_{\Xi_t\otimes\Xi_t}=\nabla\!_t+II\colon\Xi_t\ot\Xi_t\rightarrow\Xi
\end{equation}
we obtain the \textit{projected Levi-Civita connection} 
and the \textit{second fundamental form} for  the family of manifolds 
$M_c$:
\begin{equation}
    \nabla\!_t:=\mathrm{pr}_t\circ\nabla|_{\Xi_t\otimes\Xi_t}
    \colon\Xi_t\otimes\Xi_t\to\Xi_t,\qquad 
    II:=\mathrm{pr}_{\scriptscriptstyle{\perp}}\circ\nabla|_{\Xi_t\otimes\Xi_t}
    \colon\Xi_t\otimes\Xi_t\to\Xi_{\scriptscriptstyle{\perp}}.  \label{eq07}
\end{equation}
\begin{prop}\label{lem01}
The first  fundamental form $\gm_t$, the second fundamental form $ II$
and  the projected Levi-Civita 
covariant derivative $\nabla\!_t$
are $U\k$-equivariant maps.
\end{prop}
\begin{proof}{}
As compositions of $U\mathfrak{k}$-equivariant maps, $\gm_t,\nabla\!_t$ and $II$ are
$U\mathfrak{k}$-equivariant.
\end{proof}

\medskip
By the Leibniz rule for $\nabla$ and the $\X$-linearity of $\Pt$, $\Pt(hZ)=h\,\Pt(Z)$ for all $h\in\X$, $Z\in\Xi$, \
 $\nabla\!_t$ is $\X$-linear in the first argument, \ 
$\nabla\!_{t,hX}Y=h\nabla\!_{t,X}Y$, \ and fulfills
the Leibniz rule
\bea \label{Nabla_tLeibniz}
 \nabla\!_{t,X}(hY)
    =\mathrm{pr}_t\big[X(h)Y+h\nabla_XY\big]
    =X(h)\mathrm{pr}_t(Y)+h\mathrm{pr}_t(\nabla_XY)=X(h)Y+h\nabla\!_{t,X}Y
\eea
in the second argument, for all $h\in\X$ and $X,Y\in\Xi_t$. Similarly,
we find that $II$ is $\X$-linear in both arguments.
By applying the further projection $\X\to\XM$, which amounts to choosing the $c=0$ manifold $M$ out of the $M_c$ family, one finally obtains  the expected $\XM$-linearity of the
first and second  fundamental form on $M$, as well as the expected
 $\XM$-linearity in the first argument and Leibniz rule in the second
for  the Levi-Civita connection on $M$ (see e.g. \cite{Kobayashi1996}~Chapter~3).
Clearly, if $\gm$ is Riemannian also the first fundamental form on $M$ is.

Of course, one can do the same for any other $M_c$ by a different choice of $c$.

\noindent
The second fundamental form yields the extrinsic curvature of the $M_c$'s. The intrinsic curvature $\rR_t$ is related to the curvature $\rR$ of $\nabla$ on $\RR^n$ by
 the Gauss equation (valid for all \ $X,Y,Z,W\in\Xi_t$)
\bea\label{GaussClassic}
\gm\big(\rR(X,\!Y)\!Z,W\big)=\gm\!\left(\rR_t(X,\!Y)\!Z,W\right)+\gm\big(II(X,\!Z),II(Y,\!W)\big)
-\gm\big(II(Y,\!Z),II(X,\!W)\big).
\eea

\subsubsection{Decompositions in bases of $\Omega,\Xi$;  Euclidean, Minkowski metrics}
\label{Bases}

In this section we explicitly determine the geometry (in particular, the decompositions  (\ref{decoXi'}) and the associated projections $\Pt,\Pp$) 
in terms of bases of $\Omega,\Omega_{\scriptscriptstyle{\perp}},\Omega_t$
and  $\Xi,\Xip,\Xi_t$ for a generic metric $\gm$,  specializing to the Euclidean and Minkowski metrics at the end. 

Let $(x^1,...,x^n)$ be  a $n$-ple of Cartesian coordinates; 
we lower and raise indices $i,j,...$ using the metric components
$g_{ij}:=\gm(\partial_i,\partial_j)$ and  the dual ones $g^{ij}=g^{-1}{}_{ij}=\gm^{-1}(dx^i,dx^j)$
respectively: $dx_i=g_{ij}dx^j$, $Y_i=g_{ij}Y^j$, $\partial^i=g^{ij}\partial_j$, etc.
Thus we can write the metric and its dual  in the form
\be
\gm=dx^i\otimes dx_i
, \qquad\qquad \gm^{-1}=\partial^i\otimes\partial_i
,\label{EuclMetric1}
\ee
implying, for all vector fields \ $X=X^i\partial_i$, \ $Y=Y^i\partial_i$ and 1-forms
\  $\alpha=\alpha_idx^i$, \ $\omega=\omega_i dx^i$,
\be
\gm( X,Y)=X^iY_i, \qquad \qquad \gm^{-1}(\alpha,\omega)=\alpha_i\omega^i. \label{EuclMetric2}
\ee

\noindent
On ${\cal D}_f'\subseteq{\cal D}_f$ the  $k\times k$  matrix  defined by \ 
$\E^{ab}=\gmp^{-1}(df^a,df^b)$ \
 ($\E^{ab}=f^{ai}f^b_i=\big(Jg^{-1}J^{\scriptscriptstyle T}\big)^{ab}$,
in terms of Cartesian coordinates)  is symmetric 
and invertible, by (\ref{decoOmegadiamond}), (\ref{D_f'});  
we denote  its inverse by \ $K:=\E^{-1}$. \  
If the metric  $\gm$ is {\it Riemannian}, then $\E$ is also 
positive-definite on ${\cal D}_f'\!={\cal D}_f$. 
Let
\be
\Np^a:=K^{ab}\:\,\gm^{-1}(df^b ,dx^i)\:\partial_i=K^{ab}f^{bi}\partial_i       \label{DefNpa}
\ee
and, for all \ $\omega\in\Omega$,  \ $X\in\Xi$,
\bea
\omega_{\scriptscriptstyle \perp}:= df^a \, K^{ab} \:\,\gm^{-1}(df^b ,\omega),
\qquad\qquad X_{\scriptscriptstyle \perp}:=\gm(X ,\Np^a)\: \E^{ab}\Np^b,   \label{XiOmegaPerp}
\eea 
or, explicitly in terms of the decompositions  \ $\omega=\omega_i dx^i$, \ $X=X^i\partial_i$,
\bea
\omega_{\scriptscriptstyle \perp}= df^a \, K^{ab} f^{bh} \omega_h,
\qquad\qquad X_{\scriptscriptstyle \perp}=X^if_i^a\, \Np^a       \label{XiOmegaPerp'}
\eea 
(sum over repeated indices: $h,i,j,...$ run over $1,...,n$, while $a,b,c,d,...$ run over $1,...,k$). 

\begin{prop} \ 
$\N_{\scriptscriptstyle \perp}\!:=\!\{\Np^a\}_{a=1}^k$,  \
 $\B_{\scriptscriptstyle{\perp}}\!:=\!\{df^a\}_{a=1}^k$ \ are
bases resp. of   $\Xip,\Omp$   dual to each other, in the sense
\be
 \la \Np^a,df^b\ra=\Np^a(f^b)=\delta^{ab}, \qquad a,b\in\{1,...,k\}.          \label{dualperp}
\ee
$\gm^{-1}(df^a,df^b)=\E^{ab}$, $\gm(\Np^a,\Np^b)=K^{ab}$. \
The $df^a,\Np^a$ as well the $\E^{ab},K^{ab}$ are $\mathfrak{k}$-invariant. \ 
On  $X\in\Xi$,  $\omega\in\Omega$  \ the action of
the projections $\Pp,\Pt$ explicitly reads \ $\Pp(X)=X_{\scriptscriptstyle \perp}$, \ $\Pt(X)= X_t:=X-X_{\scriptscriptstyle \perp}$,  \ $\Pp(\omega)=\omega_{\scriptscriptstyle \perp}$, \ $\Pt(\omega)= \omega_t:=\omega-\omega_{\scriptscriptstyle \perp}$.
\label{XiOmegaDeco'}
\end{prop}

\bp{}
We have already proved in Proposition \ref{Propdiamond} that $\B_{\scriptscriptstyle{\perp}}$ is a basis
of $\Omp$. As a consequence, $\omega_{\scriptscriptstyle \perp}\in\Omp$.
From the definition we find $\gm(X,\Np^a)=K^{ab}X^if^b_i=K^{ab}X(f^b)=0$ for all $X\in\Xi_t$ and
$a=1,...,k$, whence  $\Np^a\in\Xip$; moreover, 
$\Np^a(f^b)=K^{ac}f^{ci}\partial_i(f^a)=K^{ac}\E^{cb}=\delta^{ab}$, and $\N_{\scriptscriptstyle \perp}$ is the  basis  of $\Xip$ dual to $\B_{\scriptscriptstyle{\perp}}$. As a consequence, $X_{\scriptscriptstyle \perp}\in\Xip$. $g\trc df^a=0$ for all $g\in\Xi_t$ holds in particular for  $g\in\mathfrak{k}$.
By Proposition \ref{lem01} $ g\trc\Np^a\in\Xip$ for all $g\in\mathfrak{k}$, and therefore 
$ g\trc\Np^a=C^a_c(g)\Np^c$ with some coefficients $C^a_c(g)$.
Applying $ g\trc$ to both sides of (\ref{dualperp}) 
and  using the $\Xi$-equivariance of the pairing we thus find  the $\mathfrak{k}$-invariance also of the $\Np^a$:
$$
 \la g\trc\Np^a,df^b\ra=0\quad\Rightarrow\quad
0=C^a_c(g)\la \Np^c,df^b\ra=C^a_b(g)\:\:\:\forall a,b\quad\Rightarrow\quad g\trc\Np^a=0.
$$
Checking $\gm^{-1}(df^a,df^b)=\E^{ab}$, 
$\gm(\Np^a,\Np^b)=K^{ab}$ is  a straightforward computation; their 
$\mathfrak{k}$-invariance follows from that of $df^a$ and the $U\mathfrak{k}$-equivariance of $\gm$;
in fact, \ $\forall g\!\in\! U\mathfrak{k}$
$$
g\trc\E^{ab}=g\trc \gm^{-1}(df^a,df^b)=
\gm^{-1}\!\left(\!g_{(1)}\trc df^a,g_{(2)}\trc df^b\!\right)=\varepsilon(g)\,\gm^{-1}(df^a,df^b)
=\varepsilon(g)\E^{ab}.
$$
The linear maps $X\mapsto X_{\scriptscriptstyle \perp}\in\Xip$, 
 $\omega\mapsto \omega_{\scriptscriptstyle \perp}\in\Omp$  indeed realize the projection $\Pp$, because
\bea
(X_{\scriptscriptstyle \perp})_{\scriptscriptstyle \perp}=\gm(X_{\scriptscriptstyle \perp} ,\Np^a)\, \E^{ab}\Np^b=\gm(X ,\Np^c)\, \E^{cd}\gm(\Np^d ,\Np^a)\, \E^{ab}\Np^b=
\gm(X ,\Np^a)\, \E^{ab}\Np^b=X_{\scriptscriptstyle \perp},\qquad\nn[6pt]
(\omega_{\scriptscriptstyle \perp})_{\scriptscriptstyle \perp}=
df^a  K^{ab} \gm^{-1}(df^b \!,\omega_{\scriptscriptstyle \perp})
=df^a  K^{ab} \gm^{-1}(df^b \!,df^c)K^{cd} \gm^{-1}(df^d\!,\omega)
=df^a  K^{ab} \gm^{-1}(df^b \!,\omega)
=\omega_{\scriptscriptstyle \perp};
\nonumber
\eea
hence also the linear maps \ $X\mapsto X_t\!:=\!X\!-\!X_{\scriptscriptstyle \perp}$, \
 $\omega\mapsto \omega_t\!:=\!\omega\!-\!\omega_{\scriptscriptstyle \perp}$ \ realize the projection  $\Pt$.
 \ep


\medskip
\noindent
{\bf Remark.} \ If  $\gm$ is {\it Riemannian},
setting $\mathcal{H}:=\E^{-1/2}$, $\theta^a:=\mathcal{H}^{ab}df^b$, 
$\Up^a:=\mathcal{H}^{ab}f^{bi}\partial_i$, one finds that
$\{\Up^a\}_{a=1}^k$, $\{\theta^a\}_{a=1}^k$ are othonormal bases of 
 \   $\Xip$, $\Omp$, respectively and are dual to each other, i.e.
\be
\gm(\Up^a,\Up^b)=\delta^{ab}, \qquad \gm^{-1}(\theta^a,\theta^b)=\delta^{ab},
\qquad \la \Up^a,\theta^b\ra=\delta^{ab}.          \label{perpdual}
\ee
The  $\mathfrak{k}$-invariance of $\theta^a,\Up^a$ follows from that
of $df^a,\Np^a$ and of  $\E$. 
In terms of the bases $\{\Up^a\}_{a=1}^k$, $\{\theta^a\}_{a=1}^k$ the normal components
of $X\in\Xi$, $\omega\in\Omega$ read
\bea
\omega_{\scriptscriptstyle \perp}= \theta^a \:   \gm^{-1}(\theta^a,\omega),
\qquad\qquad X_{\scriptscriptstyle \perp}=\gm(X ,\Up^a)\:\Up^a.\label{XiOmegaPerp''}
\eea 
Even if  $\gm$ is {\it not Riemannian} one can find in ${\cal D}_f'$
a $k\times k$ symmetric matrix $\mathcal{H}$, such that $\theta^a:=\mathcal{H}^{ab}df^b$  
$\Up^a:=\mathcal{H}^{ab}f^{bi}\partial_i$ are $\mathfrak{k}$-invariant, make up  bases 
$\{\Up^a\}_{a=1}^k$, $\{\theta^a\}_{a=1}^k$  of
   $\Xip$, $\Omp$ respectively  that are othonormal up to 
suitable signs $\epsilon_a=\pm 1$ and   dual to each other, in the sense
\be
\gm(\Up^a,\Up^b)=\zeta^{ab}, \qquad 
\gm^{-1}(\theta^a,\theta^b)=\zeta^{ab},
\qquad \la \Up^a,\theta^b\ra=\delta^{ab},         \label{perpdual'}
\ee
where $\zeta^{ab}=\zeta_{ab}:=\epsilon_a\delta^{ab}$ (no sum over $a$). 
The normal components of $X\in\Xi$, $\omega\in\Omega$ read
\bea
\omega_{\scriptscriptstyle \perp}= \theta^a \zeta_{ab}\:   \gm^{-1}(\theta^b,\omega),
\qquad\qquad X_{\scriptscriptstyle \perp}=\gm(X ,\Up^a)\:\zeta_{ab}\Up^b.\label{XiOmegaPerp'''}
\eea

\medskip
If $\gm$ is the  {\it Euclidean metric} ($g_{ij}=\delta_{ij}$)  
the associated Levi-Civita connection on $\RR^n$ is
\be
 \nabla=dx^i\otimes\mathcal{L}_{\partial_i},  \qquad \mbox{e.g.}\qquad\nabla_{\!\!X} Y=X^i\partial_i(Y^j)\partial_j. \label{LeviCivita}
\ee
We endow $M\subset \RR^n$ with the induced metric $\gm_t$. 
Using $X,Y,Z\in\Xi_t$ as representatives of elements of $\XiM$,  the Levi-Civita connection on $(M,\gm_t)$  is $\nabla\!_{t,X} Y:=(\nabla_X Y)_t$:
(\ref{LC}), (\ref{Killing0}) hold with $\gm,\nabla,\tT,\rR$ replaced by $\gm_t,\nabla\!_t,\tT_t,\rR_t$. 
Deriving the identities $Y(f_c)=Y^jf_j^a=0$ we find that $\partial_i(Y^j)f_j^a=-Y^jf_{ij}^a$,
where we have abbreviated $f_{ij}^a:=\partial_i \big(\partial_j (f^a)\big)$; 
thus, the second fundamental form  \ $II(X,Y):=\left( \nabla_X Y\right)_{\scriptscriptstyle \perp}$ \  
takes the explicit form
\bea
II(X,Y)&=& X^i\partial_i(Y^j)f_j^a  \Np^a=
-X^iY^jf_{ij}^a  \Np^a.
\eea
Replacing this result and $\rR=0$ in the Gauss equation (\ref{GaussClassic}),
 we find for the intrinsic curvature 
$$
\left[\rR_t(X,\!Y)\!Z)\right]^m W^m=f^a_{ij}K^{ab}f^b_{lm}(Y^iX^l-X^iY^l)Z^jW^m 
$$
on all \ $X,Y,Z,W\in\Xi_t$. \ Finally, $Z\in\Xi_t$  is a Killing vector field on $(M,\gm_t)$ if\footnote{In fact, 
$lhs\!=\! Z^h\partial_h\left(X^iY_i\right)\!-\! \left[Z^h\partial_h(X^i)\!-\! X^h\partial_h(Z^i)\right]Y_i \!-\!  \left[Z^h\partial_h(Y^i)\!-\! Y^h\partial_h(Z^i)\right]X_i\!=\!  [X^hY_i\!+\! X_iY^h]\,\partial_h(Z^i)\!=\! rhs$.}
\bea
Z\bigg(\gm( X,Y)\bigg)-\gm\big( [Z, X],Y\big)-\gm\big(X,[Z, Y]\big)=X^hY^i(\partial_hZ_i+\partial_iZ_h)=0
\qquad \forall X,Y\in\Xi_t.       \label{Killing}
\eea
In fact, this condition guarantees that $Z$ is Killing on $(M_c,\gm_t)$ for all $c$.
The Killing vector fields close the Lie algebra $\k=\h\cap \Xi_t$  of the group of isometries 
$\mathfrak{K}$  of the $M_c$'s; $\mathfrak{K}$ is a subgroup   of the   group  $\mathfrak{H}$
of isometries of $\RR^n$, i.e. of the Euclidean group 
(every element of  $\mathfrak{H}$ is a composition of a rotation, a translation and possibly an inversion of axis).

\medskip
If  $\gm$ is the   {\it  Minkowski  metric}  
[$g_{ij}=g^{ij}= \eta_{ij}=\mbox{diag}(1,...,1,-1)$],
the associated Levi-Civita connection on $\RR^n$ is again as in (\ref{LeviCivita}).
Endowing $M_c\subset {\cal D}_f'$ with the induced metric $\gm_t$ and 
using $X,Y,Z\in\Xi_t$ as representatives of elements of $\Xi_{\Mcs}$,  the Levi-Civita connection on $(M_c,\gm_t)$  is again $\nabla\!_{t,X} Y:=(\nabla_X Y)_t$:
(\ref{LC}), (\ref{Killing0}) hold with $\gm,\nabla,\tT,\rR$ replaced by $\gm_t,\nabla\!_t,\tT_t,\rR_t$. In terms of components 
the condition for $Z\in\Xi_t$  to be a Killing vector field on  $(M_c,\gm_t)$  remains (\ref{Killing}).

\bigskip\noindent
{\bf  Bases and complete sets of  $\Xi_t$, $\Omega_t$}
\smallskip

\noindent
As seen,  $\B_{\scriptscriptstyle{\perp}}:=\left\{df^a\right\}^k_{a=1}$,  $\N_{\scriptscriptstyle\perp}:=\left\{\Np^a\right\}^k_{a=1}$ are globally defined  bases
of the $\X$-bimodules  $\Omega_{\scriptscriptstyle{\perp}}$,  $\Xip$  respectively.
Also the $\Vp^a:=f^{ai}\partial_i=\E^{ab}\Np^b$ make a basis of  $\Xip$.
The globally  defined sets
\bea
\Theta_t:=\left\{\vartheta^j\right\}^n_{j=1}, \quad S_W:=\left\{W_j\right\}^n_{j=1},
\qquad \quad \mbox{with }\quad \vartheta^j\!:=\Pt(\xi^j), 
 \:\: W_j\!:=\Pt(\partial_j),
\eea
are respectively complete in $\Omega_t$, $\Xi_t$, but are not bases, because of the linear dependence relations
\be
\vartheta^jf^a_j=0, \qquad\qquad f^{aj}W_j=0, \qquad\qquad a=1,...,k.   
\ee
The above definition of $\B_{\scriptscriptstyle{\perp}}$ does not involve any specific metric, 
as the  definition (\ref{defOmegadiamond}) 
of  $\Omega_{\scriptscriptstyle{\perp}}$ itself. Similarly, as  the  definition (\ref{defeqXis}) 
of  $\Xi_t$ does not involve any metric, there should be
some alternative complete set in $\Xi_t$ with the same feature. To determine it we
start with the case $k=1$, i.e. with a $(n\!-\!1)$-dimensional (hyper)surface $M\subset{\cal D}_f$ determined by  a single equation
\be
f(x)=0.                                                          \label{defeq}
\ee
Rescaling $S_W$ by the factor $f^if_i$ we obtain 
another complete set:  \ $S_V:=\left\{V_j\right\}^n_{j=1}$, with $V_j:
=(f^if_i)\partial_j- f_j\Vp$. \
A third complete set (of globally  defined vector fields) in $\Xi_t$ is
\bea
S_L:=\left\{L_{ij}\right\}_{i,j=1}^n, \qquad L_{ij}:=f_i\partial_j-f_j\partial_i.
\eea
In fact,  $L_{ij}$ annihilate $f$; $S_L$  is complete because 
$V_j=f^iL_{ij}$.  This is the searched set, because its definition does not involve the metric.
Clearly $L_{ij}=-L_{ji}$, so at most  $n(n\!-\!1)/2$ of the $L_{ij}$ (e.g. those
with  $i<j$) are linearly independent over $\RR$. Obviously, both $S_V,S_L$
are of rank $n\!-\!1$ over $\X$;  they are respectively characterized by
the dependence relations  
\be
f^iV_i=0, \qquad\qquad f_{[i}L_{jk]}=0                          \label{DepRel}
\ee
(here and below square brackets enclosing indices mean a complete  antisymmetrization of the latter).
As known, if $M$ is not parallelizable there is no basis (i.e. complete
set of just $(n\!-\!1)$ elements) of $\Xi_t$ consisting
of globally defined vector fields: redundancy is unavoidable.
In the case of spheres $f\equiv (x^ix^i-R^2)/2=0$ the  $n(n\!-\!1)/2$ 
$L_{ij}:=x^i\partial_j-x^j\partial_i$ ($i<j$) are the usual generators of rotations (angular momentum components), i.e. span $so(n)$. The $L_{ij}$ are antihermitean under the $*$-structure (\ref{natural*}), namely $L_{ij}^*=-L_{ij}$.

By an explicit computation we find that their Lie brackets are
\bea
[L_{ij},L_{hk}]
& = & f_{jh}L_{ik}-f_{ih}L_{jk}-
f_{jk}L_{ih}+f_{ik}L_{jh}.       \label{comm}
\eea

\medskip
Now we consider the  general $k$ case. The globally defined vector fields 
\be 
L_{i_1i_2...i_{k+1}}:=  f^1_{[i_1}f^2_{i_2}...
f^k_{i_k}\,\partial_{i_{k+1}]}            \label{genL}
\ee
are antihermitean,  fulfill  \ $L_{i_1i_2...i_{k+1}}f^a=0$
for all $a\!=\!1,...,k$, are completely antisymmetric with respect to $(i_1,i_2,...,i_{k+1})$, 
and make up a  set $S_L$ complete (over $\X$) in $\Xi_t$, independently of the metric.
The $L_{i_1i_2...i_{k+1}}$ with $i_1<i_2<...<i_{k+1}$,
or a subset thereof, is linearly  independent over $\CC$. 
Even the latter may be linearly dependent over $\X$, 
because 
$f^a{}_{[j} L_{i_1i_2...i_{k+1}]}=0$ for all $a$. \
We do not compute their Lie brackets here.

\subsubsection{Differential calculus algebras $\Q^\bullet,\QM^\bullet$ on $\RR^n$, $M_c$}
\label{QQM} 

Henceforth we abbreviate $\xi^i:=dx^i$. \   
Let  $S=\{e_\alpha\}_{\alpha=1}^{A}$   be a  set 
 of vector fields, globally defined on ${\cal D}_f$ that 
is  complete in $\Xi$. The $e_\alpha,\xi^i$ fulfill  relations of the type
\be
\ba{l}
\sum_{\alpha=1}^A t_l^\alpha\, e_\alpha=0, \qquad l=1,...,A-n,\\[4pt]
e_\alpha e_\beta- e_\beta e_\alpha-  C^\gamma_{\alpha\beta}\,e_\gamma=0,\\[4pt]  
e_\alpha\xi^i-\xi^ie_\alpha=0, \qquad 
\xi^i\xi^j+\xi^j\xi^i=0                             
\ea \label{DCrel2}
\ee
(with suitable $t_l^a,C^\gamma_{\alpha\beta}\in\X$). 
The first line  contains possible linear dependence relations among the $e_\alpha$,
like (\ref{DepRel}). If  
we choose  $S=\{\partial_1,...,\partial_n\}$ this is empty, while 
in the second line $C^\gamma_{\alpha\beta}\equiv 0$. \
Clearly the coefficients in the decomposition $X=X^\alpha e_\alpha\in \Xi$
are defined up to shifts $X^\alpha \mapsto X^\alpha + \sum_l h^l t_l^\alpha$,
with $h^l\in\X$. 
Consider  the unital 
algebra   $\Q^\bullet$  over $\CC$ consisting of  polynomials
in $\xi^i,e_\alpha$ with (left or right) coefficients in $\X$, modulo 
relations (\ref{DCrel2}) and the ones
\bea 
&&\ba{l}
h\xi^i-\xi^i h=0, \qquad
e_\alpha h -h\, e_\alpha- e_\alpha(h)=0
\ea\qquad\quad \forall h\in\X;         \label{DCrel1}  
\eea
 $\Q^\bullet$  is a $U\Xi$-equivariant $\X$-bimodule. It is easy to check that  a different choice of $S$ changes
 (\ref{DCrel2}-\ref{DCrel1}), but leads to an equivalent definition of  $\Q^\bullet$ (one could
choose also a different basis of 1-forms, but we will not consider this here).
We shall name $\Q^\bullet$  {\it differential calculus algebra  on ${\cal D}_f$}. 
The elements of 
$\Q^\bullet$ can be considered as differential-operator-valued 
inhomogenous forms. Relations  (\ref{DCrel2}-\ref{DCrel1}) encode all the information about the differential calculus 
and allow to order  the $\xi^i,e_\alpha$ in any prescribed way, with the
coefficient functions at the left, center, or right - as one wishes. $\Q^\bullet$ 
admits  $\X,\Omega^\bullet,{\cal H}$ as  subalgebras; the {\it enlarged Heisenberg algebra} 
${\cal H}$ is  the  component of form degree zero. While $\Q^\bullet,\Omega^\bullet$ are graded by the form degree, $\Q^\bullet,{\cal H}$ are filtered by the degree $r$ in the $e_\alpha$; $r$ gives the order of an element of ${\cal H}$ seen as a differential operator on $\X$. 
Note that within $\Q^\bullet$  also the action of a generic vector field $X=X^\alpha e_\alpha$
on a function $h$ can be expressed as a commutator:
\ $[X,h]=[X^\alpha e_\alpha,h]=X^\alpha [e_\alpha,h]=X(h)$.
In the $\Q^\bullet$ framework $Xh=hX+X(h)$ is the inhomogeneous first order differential operator sum of a first order part (the vector field  $hX$) and a zero order part  (the multiplication operator by $X(h)$); it must not be confused with the  product of  $X$ by $h$ from the right, which  is equal to  $hX$ and in the previous sections has been denoted in the same way.  In the $\Q^\bullet$ framework we denote the latter by $X \ltlc  h$ (of course
$(X \ltlc  h)(h')=X(h') h=hX(h')$,  $X \ltlc  (hh')= hh'X$ remain valid).
We endow $\Q^\bullet$  with the  natural $*$-struture  defined  by
\be
f^*(x)=\overline{f(x)},\qquad \partial_i^*=-\partial_i,\qquad \xi^i{}^*=\xi^i.
\label{natural*}
\ee 

If one chooses $S$ so that a subset $S_t:=\{e_\alpha\}_{\alpha=1}^B$
($B\!:=\!A\!-\!k$)  
 is  complete in $\Xi_t$ (e.g. it consists of the $L_{i_1i_2...i_{k+1}}$), 
while $e_{B+a}:= s^b_a\Np^b$, with some  matrix $s_a^b(x)$  ($a,b\in\{1,...,k\}$)
invertible everywhere, then if
$\alpha,\beta\leq B$  the sum in (\ref{DCrel2})$_2$ is extended over 
$\gamma\leq B$.
The differential calculus algebra  $\QM^\bullet$   on $M_c$
is  the ${\cal X}^{\scriptscriptstyle M_c}$-bimodule generated by the $\xi^1,...,\xi^n,e_1,...,e_B$
modulo the relations (\ref{DCrel2}-\ref{DCrel1}) (with $\alpha,\beta\leq B$) and the ones
\bea 
&&\ba{l}
f_c^a\equiv f^a\!-\!c^a\1=0,\qquad
df^a\equiv\xi^hf_h^a=0,     
\ea     \qquad\quad a=1,...,k.              \label{DCMcrel}
\eea

\subsection{Twisted differential geometry of manifolds embedded in $\RR^n$}\label{SecTwistingVF}

Using a twist \ $\F\in \big(U\Xi_t\otimes U\Xi_t\big) [[\nu]]$ \ and following
the general twisting approach we deform the differential geometry
on ${\cal D}_f$ in a way compatible with the embeddings, i.e. so that it projects
to the twist deformation of the differential geometry on the submanifolds $M_c$, $c\in\RR^n$.
Equivalently, we deform the differential calculus algebra
$\Q^\bullet$ on ${\cal D}_f$ into an associated $\Q^\bullet_\star$  
in a way compatible with the embeddings, i.e. encoding
through projections all deformations $\QM^\bullet\leadsto\QMst^\bullet$. 
Unless explicitly stated,
we still denote by $X\star h=(\bR_1\trc h)\star(\bR_2\trc X)$ the vector field
that is $\star$-product of the one $X$ by the function $h$ from the right, as done so far.

To state the twisted analog of Proposition \ref{Propdiamond}
we first define a $\X_\star$-subbimodule 
$\Omega_{\scriptscriptstyle{\perp}\star}\subset \Omega_\star$:
\begin{equation}
\Omega_{\scriptscriptstyle{\perp}\star}:=
\{\omega\in\Omega_\star \:|\: \langle\Xi_{t\star},\omega\rangle_{\star}=0\} .
\end{equation}
%


\begin{prop}\label{prop07}
Equipped with the $\star$-Lie bracket $[\:,\:]_\star$ $\Xi_{t\star},\Xi_{\C\star}$ are  $\star$-Lie subalgebras of $\Xi_\star$, and $\Xi_{\C\C\star}$ is an ideal of $\Xi_{\C\star}$. Another $\star$-Lie algebra is thus
\be
\XiM{}_{\star}:=\Xi_{\C\star}/\Xi_{\C\C\star}\equiv\big\{\, [X]:=X+\Xi_{\C\C\star} \:\: |\:\: X\in\Xi_{\C\star}\big\}.                                \label{quotient'star}
\ee 
Moreover,  $\Xi_{t\star},\Xi_{\C\star}, \Xi_{\C\C\star},\XiM{}_{\star},
\Omega_{\scriptscriptstyle{\perp}\star}$
resp. coincide  with  $\Xi_t[[\nu]],\Xi_\C[[\nu]], \Xi_{\C\C}[[\nu]],\XiM[[\nu]],\Omega_{\scriptscriptstyle{\perp}}[[\nu]]$ as  $\CC[[\nu]]$-modules. 
$\Xi_{t\star}$,  $\XiM{}_{\star}$, $\Omega_{\scriptscriptstyle{\perp}\star}$ and the corresponding exterior algebras $\Xi_{t\star}^\bullet=\bigwedge_\star\!\!\!\!^\bullet\,\Xi_{t\star}$,  $\XiM^\bullet{}_{\star}=\bigwedge_\star\!\!\!\!^\bullet\,\XiM$,  
 $\Omega^\bullet_{\scriptscriptstyle{\perp}\star}=\bigwedge_\star\!\!\!\!^\bullet\,\Omega_{\scriptscriptstyle{\perp}\star}$  are  $U\Xi_t^\f$-equivariant $\X_\star$-bimodules. 
 $\Omega_{\scriptscriptstyle{\perp}\star}$ can be explicitly decomposed as
\be
\Omega_{\scriptscriptstyle{\perp}\star}=\bigoplus_a\X_\star\star\mathrm{d}f^a=\bigoplus_a\mathrm{d}f^a \star\X_\star.  \label{decoOmegadiamondstar}
\ee
\label{Propdiamondstar}
\end{prop}
\vskip-.5cm

{\bf Proof} \ \
These are direct consequences of the
following properties.
By Proposition \ref{Propdiamond}, for all \ $h\in\X[[\nu]]$, \ $g\in U\Xi_t[[\nu]]$, \ $X,X'\in\Xi_t[[\nu]]$, \ $Y,Y'\in\Xi_\C[[\nu]]$, $Z\in\Xi_{\C\C}[[\nu]]$, \ $\omega=\omega_adf^a\in\Omega_{\scriptscriptstyle{\perp}}[[\nu]]$:  

\medskip
$\bullet$   $h\star X,\:\: X\star h,\:\: g\trc X$ \ and \ 
$[X,X']_\star$ \ belong to \  $\Xi_t[[\nu]]$.

\smallskip
$\bullet$   $h\star Y,\:\:Y\star h,\:\: g\trc Y$  \ and  \ $[Y,Y']_\star$ \ belong to \  $\Xi_\C[[\nu]]$.

\smallskip
$\bullet$   $h\star Z,\:\:Z\star h,\:\:g\trc Z$ \ and $[Y,Z]_\star$ \ belong to \  $\Xi_{\C\C}[[\nu]]$ \ (because $\Xi_{\C\C}$ is an ideal in $\Xi_\C$).

\smallskip
$\bullet$   $h\star [Y],\:\:[Y]\star h,\:\:g\trc  [Y]$ \ and $\big[ [Y], [Y']\big]_\star$ \ belong to \  $\XiM[[\nu]]$.

\smallskip
$\bullet$   $h\star df^a=h df^a=(df^a)\star h$ \ and \ 
$g\trc \omega=(g\trc \omega_a)df^a $ \ belong to \  $\Omega_{\scriptscriptstyle{\perp}}[[\nu]]$, \ by 
(\ref{twistcond}) and the relation $g\trc df^a=\varepsilon(g)df^a$.

\smallskip
$\bullet$   $\langle X,\omega\rangle_\star=\left\langle \bF_1\trc X,\bF_2\trc\omega\right\rangle=0$, \
because \ $\bF_1\trc X\in\Xi_t[[\nu]]$ \ and \ $\bF_2\trc\omega\in\Omega_{\scriptscriptstyle{\perp}}[[\nu]]$. \hfill $\Box$

\bigskip
\noindent
This means in particular that taking the quotient commutes with twisting. 
To build explicit examples of twist-deformed submanifolds we recall that several known types of Drinfel'd twists (as the ones  mentioned in section \ref{TwistSym}) are based on finite-dimensional Lie algebras. When does the infinite-dimensional 
$\Xi_t$ admit a finite-dimensional Lie subalgebra $\g$ over $\RR$,
so that we can choose $\F\in (U\g\otimes U\g)[[\nu]]$?
Given a set $S$ of vector fields that is complete in $\Xi_t$,  the question is which combinations   (with coefficients in $\X$)   of them, if any,
close a finite-dimensional Lie algebra $\g$. An easy answer is available
for the quadrics in $\RR^n$, see section \ref{Examples}.
If $\RR^n$ endowed with a metric admits a family $M_c$ of (pseudo)Riemannian submanifolds 
manifestly symmetric under a  
Lie group $\mathfrak{K}$ (its group of isometries)\footnote{For instance, the sphere $S^{n-1}$ is 
$SO(n)$ invariant; a cylinder in $\RR^3$ is invariant under $SO(2)\times \RR$; 
the hyperellipsoid of equation $(x^1)^2\!+\!(x^2)^2\!+\!2[(x^3)^2\!+\!(x^4)^2] =1$  is invariant under $SO(2)\times SO(2)$; etc.} ,  then a nontrivial $\g$ exists and 
contains the  (Killing) Lie algebra $\mathfrak{k}$ of  $\mathfrak{K}$ (if $M_c$ is maximally 
 symmetric  then $\mathfrak{k}$ is even complete - over $\X$ - in $\Xi_t$).
In the next subsections we consider such a case and stick to deformations induced 
by a twist $\mathcal{F}$ based on $\mathfrak{k}\subset \Xi_t$; under these assumptions 
the deformation is compatible with the geometry. $\Xi_{t\star},\Omega_{\scriptscriptstyle{\perp}\star}$
appear as addends in  the decomposition of $\Xi_\star$ in tangent and orthogonal vector fields.
In Section~\ref{Basesstar}  we first give explicit results for a generic metric, then specialize the discussion to the Euclidean  and  Minkowski metric.

\subsubsection{Twisted metric, Levi-Civita connection, intrinsic and extrinsic curvatures}
\label{TwistedMetric}

As seen in section \ref{Metric}, 
endowing ${\cal D}_f\subseteq\mathbb{R}^n$ with a (non-degenerate) metric $\gm$  makes all the $M_c\subset {\cal D}_f'$ into (pseudo)Riemannian submanifolds;  ${\cal D}_f'\subseteq{\cal D}_f$ is where the restriction 
$\gmp^{-1}$ is  non-degenerate.
For a generic twist $\F\in (U\Xi_t\otimes U\Xi_t)[[\nu]]$ we introduce the $\X$-subbimodules 
\bea
\Xips:=\{X\in\Xi_\star \:\:\:|\:\:\: \gm_\star(X,\Xi_{t\star})=0\}, \qquad
\Omega_{t\star}:=\{\omega\in\Omega_\star \:\:\: |\:\:\:  \gm^{-1}_\star(\omega,\Omps)=0  \},
\eea
Again, let $\k\subset\Xi_t$ the Lie subalgebra of Killing vector fields of $\gm$ that are also tangent to the submanifolds $M_c\subset {\cal D}_f'$. 
The twisted version of Proposition \ref{XiOmegaDeco} reads

\begin{prop}\label{prop08}
If $\F\in (U\mathfrak{k}\otimes U\mathfrak{k})[[\nu]]$
the $\star$-Lie algebra $\Xi_\star$ of 
smooth vector fields and the $\X_\star$-bimodule $\Omega_\star$ of 1-forms
on ${\cal D}_f'$ split into the direct sums of  $\X_\star$-subbimodules 
\bea
\Xi_\star=\Xi_{t\star}\oplus\Xi_{{\scriptscriptstyle{\perp}}\star},\qquad 
\quad \Omega_\star=\Omega_{t\star}\oplus\Omega_{\scriptscriptstyle{\perp}\star} \label{decoXi'star}
\eea
orthogonal with respect to the twisted metrics $\gm_\star$ and $\gm^{-1}_\star$ respectively. $\Xi_{t\star}$ is a $\star$-Lie subalgebra of $\Xi_\star$. 
$\Omega_{t\star},\Xips$ are orthogonal with respect to the $\star$-pairing, 
$\Omega_{t\star}\!=\!\{\omega\!\in\!\Omega_\star \, |\, \la \Xips,\omega\ra_\star\!=\!0\}$. \ 
%
%
Also the  restrictions  of $\gm^{-1}_\star$ (resp. $\gm$) to the {\rm tangent} and {\rm normal} 1-forms (resp. vector fields)
\be
\gmps^{-1}:=\gm^{-1}_\star|_{\Omps\ots\Omps}\colon\Omps\ots\Omps\to\X_\star, \qquad
\gm^{-1}_{t\star}:=\gm^{-1}_\star|_{\Omega_{t\star}\ots\Omega_{t\star}}\colon\Omega_{t\star}\ots\Omega_{t\star}\to\X_\star, 
\ee
\be
\gmps:=\gm_\star|_{\Xips\otimes_\star\Xips}:\Xips\otimes_\star\Xips\to\X_\star,
  \qquad 
   \gm_{t\star}:=\gm_\star|_{\Xi_{t\star}\otimes_\star\Xi_{t\star}}\colon\Xi_{t\star}\otimes_\star\Xi_{t\star}\to\X_\star \label{gmtstar}
\ee
are non-degenerate.
$\Xi_{t\star},\Omega_{{\scriptscriptstyle{\perp}}\star},\Xi_{{\scriptscriptstyle{\perp}}\star},\Omega_{t\star}$
resp. coincide  with  $\Xi_t[[\nu]],\Omp[[\nu]],\Xip[[\nu]],\Omega_t[[\nu]]$ as  $\CC[[\nu]]$-modules. Similarly for $\star$-tensor powers of the former.
The orthogonal projections
\ $\Pps: \Xi_\star\to \Xips$, \ $\Pts:\Xi_\star\to \Xi_{t\star}$,  \ $\Pps:\Omega_\star\to \Omega_{{\scriptscriptstyle{\perp}}\star}$, \ $\Pts:\Omega_\star\to \Omega_{t\star}$
are uniquely extended as projections to the bimodules of multivector fields and higher rank forms through the rules
$\Pps(\omega\star\omega')=\Pps(\omega) \star\Pp(\omega')$, $\Pts(\omega\star\omega')=\Pts(\omega) \star\Pts(\omega')$,...:
\bea
\Pps: \Omega^p_\star\to \Omega^p_{{\scriptscriptstyle{\perp}}\star}, \quad \Pts:  \Omega^p_\star\to \Omega^p_{t\star}, \quad
\Pps: \bigwedge_\star\nolimits^p\Xi_\star\to \bigwedge_\star\nolimits^p\Xi_{{\scriptscriptstyle{\perp}}\star}, \quad \Pts: \bigwedge_\star\nolimits^p\Xi_\star\to \bigwedge_\star\nolimits^p\Xi_{t\star}.
\eea
$\Pps,\Pts$ are the $\CC[[\nu]]$-linear extensions of $\Pp,\Pt$. \
 $\Xi_{t\star},\Xi_{{\scriptscriptstyle{\perp}}\star},\Omega_{t\star},\Omega_{{\scriptscriptstyle{\perp}}\star}$, their $\star$-exterior powers and the projections $\Pps,\Pts$  are $U\mathfrak{k}^\f$-equivariant. 
\label{XiOmegaDecostar}
\end{prop} 
Again we stress that, while $\Xi_{t\star}$ is a $\star$-Lie subalgebra of $\Xi_\star$, in general $\Xi_{{\scriptscriptstyle{\perp}}\star}$ is not.
Furthermore, as $\Xi_{{\scriptscriptstyle{\perp}}\star},\Omega_{t\star}$ are not $U\Xi_t^\f$-equivariant, also the orthogonal projections $\Pps,\Pts$ are not. 

\bp{}
By Proposition \ref{Propdiamondstar} \
 $\Xi_{t\star}$ is a $\star$-Lie subalgebra of $\Xi_\star$ and a $U\Xi_t^\f$-equivariant 
$\X_\star$-subbimodule;
in particular it is  $U\mathfrak{k}^\f$-equivariant.
Moreover, according to Proposition~\ref{prop02},
\begin{equation}
    \gm_\star(X,Y)
    =\gm(\overline{\F}_1\trc X,\overline{\F}_2\trc Y)
    =\gm(X,Y)+\mathcal{O}(\nu)
    \text{ for all }X,Y\in\Xi.    \nonumber
\end{equation}
If $X\in\Xi_{\scriptscriptstyle{\perp}}$ (i.e. $\gm(X,Y)=0$ for all $Y\in\Xi_t$) it
follows that
\begin{equation}
    \gm_\star(X,Y)
    =\gm(\underbrace{\overline{\F}_1\trc X}_{\in\Xi_{\scriptscriptstyle{\perp}}[[\nu]]}
    ,\underbrace{\overline{\F}_2\trc Y}_{\in\Xi_t[[\nu]]})
    =0           \nonumber
\end{equation}
for all $Y\in\Xi_t$, i.e. $\Xi_{{\scriptscriptstyle{\perp}}}[[\nu]]\subseteq
\Xi_{{\scriptscriptstyle{\perp}}\star}$. On the other hand, for every
$X=\sum_{n=0}^\infty\nu^nX_n\in\Xi_{{\scriptscriptstyle{\perp}}\star}$ with
$X_n\in\Xi$ it follows that $0=\gm_\star(X,Y)=\gm(X_0,Y)+\mathcal{O}(\nu)$
for all $Y\in\Xi_t$, i.e. $\gm(X_0,Y)=0$ for all $Y\in\Xi_t$. In other words
$X_0\in\Xi_{\scriptscriptstyle{\perp}}$. Also $X_1\in\Xi_{\scriptscriptstyle{\perp}}$,
since
\begin{equation}
    0=\gm_\star(X,Y)
    =\underbrace{\gm(X_0,Y)}_{=0}
    +\nu\bigg(
    \gm(X_1,Y)
    +\underbrace{\gm(\overbrace{\overline{F}_1^1\trc X_0}^{\in\Xi_{\scriptscriptstyle{\perp}}},
    \overbrace{\overline{F}_2^1\trc Y}^{\in\Xi_t})}_{=0}
    \bigg)
    +\mathcal{O}(\nu^2) \nonumber
\end{equation}
for all $Y\in\Xi_t$, where $\overline{\F}=\sum_{n=0}^\infty\nu^n
\overline{F}^n_1\otimes\overline{F}^n_2$ and $\overline{F}^n_1\otimes\overline{F}^n_2\in
U\k\otimes U\k$.
Inductively $X_n\in\Xi_{\scriptscriptstyle{\perp}}$ for all $n\geq 0$, implying
$\Xi_{{\scriptscriptstyle{\perp}}}[[\nu]]=\Xi_{{\scriptscriptstyle{\perp}}\star}$, as claimed.
This also implies the equality 
\begin{equation}
    \mathrm{pr}_{{\scriptscriptstyle{\perp}}\star}
    =\mathrm{pr}_{{\scriptscriptstyle{\perp}}}
    \colon\Xi[[\nu]]\rightarrow\Xi_{{\scriptscriptstyle{\perp}}}[[\nu]].
\end{equation}
Note that $\gm^{-1}_\star(\omega,\alpha)=
\gm^{-1}(\overline{\mathcal{F}}_1\trc \omega,\overline{\mathcal{F}}_2\trc \alpha)$
for all $\omega,\alpha\in\Omega$, since $\mathcal{F}$ is based on Killing vector fields. 
Now assume that $X\in\Xips$ (resp. $X\in\Xi_{t\star}$) fulfills $\gmps(X,\Xips)=0$
 (resp. $\gm_{t\star}(X,\Xi_{t\star})=0$). Expanding $X$ and $\gmps$ 
 (resp. $\gm_{t\star}$)
in $\nu$-powers and arguing as above, we find   $X=0$, whence
the non-degeneracy of $\gmps$ (resp. $\gm_{t\star}$).
By employing (\ref{lerest}), Proposition~\ref{XiOmegaDeco} and the equivariance of the $\star$-pairing and $\gm^{-1}_\star$, one similarly proves that $\Omega_{t\star}=\Omega_t[[\nu]]$,
 $\Omega_{\scriptscriptstyle{\perp}\star}=\Omega_{\scriptscriptstyle{\perp}}[[\nu]]$
and the non-degeneracy of  $\gmps^{-1},\gm_{t\star}^{-1}$ on $\Omega_\star$.
Let $X\in\Xi_{{\scriptscriptstyle{\perp}}\star}$,  $\omega\in\Omega_{t\star}$,
$\xi\in U\mathfrak{k}^\f $. Then
\bea
 &&    \gm_\star(\xi\trc X,Y)
    =\xi_{\widehat{(1)}}\trc\gm_\star(X,S_\f(\xi_{\widehat{(2)}})\trc Y)    =0\quad \forall Y\in\Xi_{t\star}
\qquad\Rightarrow \qquad \xi\trc X\in\Xi_{{\scriptscriptstyle{\perp}}\star} \nn[6pt]
 &&   \la X,\xi\trc \omega\ra_\star
    =\xi_{\widehat{(1)}}\trc\la S_\f(\xi_{\widehat{(2)}})\trc X, \omega \ra_\star
    =0\quad  \forall X\in\Xi_{{\scriptscriptstyle{\perp}}\star}\qquad\Rightarrow \qquad \xi\trc \omega\in\Omega_{t\star} \nonumber
\eea
since $\gm_\star$ and the $\star$-pairing are equivariant under 
the action of $U\mathfrak{k}^\f $
and $\Xi_{t\star}$ is a $U\mathfrak{k}^\f $-equivariant $\X_\star$-bimodule.
This proves  that also
$\Xi_{{\scriptscriptstyle{\perp}}\star},\Omega_{t\star}$ are $U\k$-equivariant $\X$-bimodules. 
To verify that $\Pps$ is $U\mathfrak{k}^\f$-equivariant let 
$X=X_{t\star}+X_{{\scriptscriptstyle{\perp}}\star}\in\Xi_\star$ be the decomposition
(\ref{decoXi'star}) with $X_{t\star}\in\Xi_{t\star}$ and $X_{{\scriptscriptstyle{\perp}}\star}\in\Xi_{{\scriptscriptstyle{\perp}}\star}$. Then
$\xi\trc X=\xi\trc X_{t\star}+\xi\trc X_{{\scriptscriptstyle{\perp}}\star}$ and according to the $U\k^\f$-invariance of $X_{t\star},X_{{\scriptscriptstyle{\perp}}\star}$
\begin{align*}
    \mathrm{pr}_{{\scriptscriptstyle{\perp}}\star}(\xi\trc X)
    =\mathrm{pr}_{{\scriptscriptstyle{\perp}}\star}(\xi\trc X_{t\star}+\xi\trc X_{{\scriptscriptstyle{\perp}}\star})
    =\xi\trc X_{{\scriptscriptstyle{\perp}}\star}
    =\xi\trc\mathrm{pr}_{{\scriptscriptstyle{\perp}}\star}(X)
\end{align*}
for all $\xi\in U\mathfrak{k}^\f$.
Similarly one argues  with $\Pts$ on $\Omega_\star$ and on the $\star$-exterior
powers of $\Xi_{\star},\Omega_{\star}$. 
Finally, 
$\Omega_{t\star}\!=\!\{\omega\!\in\!\Omega_\star \, |\, \la \Xips,\omega\ra_\star\!=\!0\}$  follows from its undeformed counterpart and the previous results.
\ep

As in the undeformed case, we  identify 
$\Xi_{t\star}\subset\Xi_\star$ with the $\star$-Lie subalgebra of smooth vector fields tangent to {\it all} $M_c$ ($c\in f({\cal D}_f')$) at all points, because $X(f^a)=0$ implies  $X(f_c^a)=0$;
and $\XiM{}_\star\subset\Xi_\star$ defined in (\ref{quotient'star}) with  the twisted Lie subalgebra of smooth vector fields tangent to $M$ at all points. Similarly, we   identify 
$\Omega_{t\star}$ with  the subbimodule of $\Omega_\star$ tangent to {\it all} $M_c$ 
($c\in f({\cal D}_f')$) at all points.  \
We find \ $\Omega_{t\star}\subset\Omega_{\C\star}:=\left\{\omega\in\Omega_\star \:\:\: |\:\:\:  \la \Xi_{{\scriptscriptstyle{\perp}}\star},\omega\ra_\star\subset\C[[\nu]] 
\right\}. 
$ \
Let \ $\Omega_{\C\C\star}:=\bigoplus_{a=1}^k f^a\star\Omega_\star=\bigoplus_{a=1}^k \Omega_\star \star f^a\subset\Omega_{\C\star}$. \
It fulfills \ $ \la \Xi_\star,\Omega_{\C\C\star}\ra_\star\subset\C[[\nu]]$. \
We can identify the $\XM_\star$-bimodule of 1-forms $\Omega_{{\scriptscriptstyle M}\star}$ on $M$ with the quotient
\be
\Omega_{{\scriptscriptstyle M}\star}=\Omega_{\C\star}/\Omega_{\C\C\star}=\left\{[\omega]=\omega+\Omega_{\C\C\star}\:\: |\:\: \omega\in\Omega_{\C\star} \right\}.
\ee

\begin{prop} 
For all \ $X\in\Xi_{\mathcal{C}\star}$, \ $\omega\in\Omega_{\mathcal{C}\star}$ \ the tangent projections \ $X_{t\star}:=\mathrm{pr}_{t\star}(X)\in\Xi_{t\star}$, \
$\omega_{t\star}:=\mathrm{pr}_{t\star}(\omega)\in\Omega_{t\star}$ \
respectively belong to \ $[X]\in\Xi_{\Ms\star}$ \ and \ $[\omega]\in\Omega_{\Ms\star}$
 \label{XCreprstar}
\end{prop}
Consequently, we can represent every element of $\Xi_{{\scriptscriptstyle M}\star}$ (resp. $\Omega_{{\scriptscriptstyle M}\star}$) by an element of $\Xi_{t\star}$ (resp. $\Omega_{t\star}$). Similarly for multivector fields and higher rank forms.

\bp{} 
By Propositions~\ref{prop07}, \ref{prop08} the twist-deformed spaces can be
identified with formal power series of the undeformed ones and the twisted projections are
given by the $[[\nu]]$-linear extensions of the undeformed ones. The claim follows
as a corollary of Proposition~\ref{prop05}.
\ep

Motivated from the classical situation we define the \textit{twisted first} and
\textit{second fundamental form} on the family of submanifolds $M_c$ by
\begin{equation}\label{eq08}
\begin{split}
    \gm_{t\star}:=&\gm_\star|_{\Xi_{t\star}\otimes_\star\Xi_{t\star}}
    \colon\Xi_{t\star}\ot_\star\Xi_{t\star}\to\X_{\star},\\
    II^\f_\star:=&\mathrm{pr}_{\scriptscriptstyle{\perp}\star}
    \circ\nabla^\f|_{\Xi_{t\star}\otimes_\star\Xi_{t\star}}\colon
    \Xi_{t\star}\otimes_\star\Xi_{t\star}\rightarrow\Xi_{{\scriptscriptstyle{\perp}}\star},
\end{split}    
\end{equation}
as well as the \textit{twisted projected Levi-Civita connection} on the family of submanifolds $M_c$
\begin{equation}
    \nabla\!_t^{\,\f}:=\mathrm{pr}_{t\star}
    \circ\nabla^\f|_{\Xi_{t\star}\otimes\Xi_{t\star}}
    \colon\Xi_{t\star}\ot\Xi_{t\star}\to\Xi_{t\star}.
\end{equation}
In the following proposition we clarify the relation of these objects to their
classical counterparts. In particular, that twist deformation and
projection to the submanifold commute.
\begin{prop}\label{lem01star}
$\nabla\!_t^{\,\f}$ is a twisted covariant
derivative on the family of submanifolds $M_c$. The twisted first fundamental form $\gm_{t\star}$ is a metric on the family,
with corresponding twisted Levi-Civita covariant derivative $\nabla\!_t^{\,\f}$.
They, as well as the second fundamental form, are $U\mathfrak{k}^\f$-equivariant.
In terms of the undeformed objects we obtain
\begin{equation}\label{eq22}
    \gm_{t\star}(X,Y)
    =\gm_t(\bF_1\trc X,\bF_2\trc Y),
\end{equation}
\begin{equation}\label{II^F}
    II^\f_\star(X,Y)
    =II(\overline{\F}_1\trc X,\overline{\F}_2\trc Y),
\end{equation}
and
\begin{equation}\label{eq23}
    \nabla\!_{t,X}^{\,\f}Y
    =\nabla\!_{t,\overline{F}_1\trc X}(\bF_2\trc Y)
\end{equation}
for all $X,Y\in \Xi_{t\star}=\Xi_t[[\nu]]$. Furthermore
\begin{equation}
    \nabla^\f_X 
    =\nabla\!_t^{\,\f}+II^\f_\star
    \colon\Xi_{t\star}\ot\Xi_{t\star}\to\Xi_\star.
\end{equation}
\end{prop}
\begin{proof}{} 
As a composition of $U\mathfrak{k}^\f$-equivariant maps, 
$\gm_{t\star}$, $II^\f_\star$ and $\nabla\!_t^{\,\f}$ also are. Eq.
(\ref{eq22}) follows from (\ref{twistedmetric2}).
Since $\nabla^\f_X=\nabla_{\overline{\f}_1\trc X}\overline{\F}_2\trc$ for all
$X\in\Xi$ we find  (\ref{II^F}) and (\ref{eq23})
(see also \cite{AschieriCastellani2009}~eq.~129). Then it follows from 
Proposition~\ref{lem01} that $\gm_{t\star}$ is a (non-degenerate) 
metric on the $M_c$'s
 with twisted Levi-Civita covariant derivative given by $\nabla\!_t^{\,\f}$.
\end{proof}

\bigskip
\noindent 
A generalization of
Proposition~\ref{lem01star} to braided commutative geometry is in \cite{TWeber2019}~Proposition~4.4.

The twisted second fundamental form (\ref{eq08})
yields the twisted extrinsic curvature of $M$.  The twisted intrinsic curvature $\rR_{t\star}^\f$ is related to the twisted
curvature $\rR^\f_\star$ of $\nabla^\f$ on $\RR^n$ by the following quantum analogue of the Gauss equation (see appendix~\ref{proofQGauss} for the proof):

\begin{prop}
For all $X,Y,Z,W\in\Xi_{t\star}$  the following twisted Gauss equation holds:
\bea\label{GaussQuantum}
    \gm_\star(\rR^\f_\star(X,Y)Z,W)
    &=&\gm_\star(\rR_{t\star}^\f(X,Y)Z,W)
    +\gm_\star(II^\f_\star(X,\overline{\mathcal{R}}_1\trc Z),
    II^\f_\star(\overline{\mathcal{R}}_2\trc Y,W))\nn[8pt]
    &&-\gm_\star(II^\f_\star(\overline{\mathcal{R}}_{1\widehat{(1)}}\trc Y,
    \overline{\mathcal{R}}_{1\widehat{(2)}}\trc Z),
    II^\f_\star(\overline{\mathcal{R}}_2\trc X,W)).
\eea
\label{QuantumGauss}
\end{prop}
The twisted first and second fundamental forms,  Levi-Civita connection, 
curvature tensor, Ricci tensor, Ricci scalar {\it on $M$} are 
finally obtained from the above objects
by applying the further projection $\X_\star\to\XM_\star$, which amounts to choosing the $c=0$ manifold $M$ out of the $M_c$ family. 
Of course,  by a different choice of $c$ one can do the same on any other $M_c$.

\subsubsection{Decompositions in bases of $\Omega_\star,\Xi_\star$;  Euclidean, Minkowski metrics}
\label{Basesstar}

In this section we explicitly determine the twisted geometry induced by a twist  $\F\in U\k\otimes U\k[[\nu]]$ (in particular, the decompositions  (\ref{decoXi'star}))
in terms of bases of $\Omega_\star,\Omega_{\scriptscriptstyle{\perp}\star},\Omega_{t\star}$
and  $\Xi_\star,\Xi_{{\scriptscriptstyle{\perp}}\star},\Xi_{t\star}$ for a generic  metric  
$\gm$ on $\RR^n$, specializing to the Euclidean and Minkowski metric at the end, as done in section \ref{Bases}. 

By Proposition \ref{XiOmegaDecostar} 
$\Xi_{t\star},\Omega_{{\scriptscriptstyle{\perp}}\star},\Xi_{{\scriptscriptstyle{\perp}}\star},\Omega_{t\star}$ are $U\mathfrak{k}^\f$-equivariant, and the 
projections $\Pps,\Pts$  are $U\mathfrak{k}^\f$-equivariant.
Here is the twisted analogue of Proposition \ref{XiOmegaDeco'} and the remark following it:

\begin{prop} \ 
$df^a,\Np^a,\theta^a,\Up^a$  are $U\mathfrak{k}^\f$-invariant.
$\N_{\scriptscriptstyle\perp}:=\left\{\Np^a\right\}^k_{a=1}$,
$\B_{\scriptscriptstyle{\perp}}:=\left\{df^a\right\}^k_{a=1}$  are  $\star$-dual bases  of  
$\Xi_{{\scriptscriptstyle \perp}\star},\Omega_{{\scriptscriptstyle \perp}\star}$ respectively:
\be
 \la \Np^a,df^b\ra_\star=\delta^{ab}, \qquad a,b\in\{1,...,k\}.          \label{dualperpstar}
\ee
$\gm^{-1}_\star(df^a,df^b)=\E^{ab}$, $\gm_\star(\Np^a,\Np^b)=K^{ab}$. \
$\{\Up^a\}_{a=1}^k$, $\{\theta^a\}_{a=1}^k$ are   $\star$-dual,  othonormal 
(possibly up to signs $\epsilon_a=\pm 1$) bases  of 
 \   $\Xi_{{\scriptscriptstyle \perp}\star},\Omega_{{\scriptscriptstyle \perp}\star}$ respectively, in the sense 
\be
\gm_\star(\Up^a,\Up^b)=\epsilon_a\delta^{ab}, \qquad \gm^{-1}_\star(\theta^a,\theta^b)=\epsilon_a\delta^{ab},
\qquad \la \Up^a,\theta^b\ra_\star=\delta^{ab}.          \label{perpdualstar}
\ee
On  $X\in\Xi_\star$,  $\omega\in\Omega_\star$  \ the actions of
the projections $\Pps,\Pts$ explicitly read $\Pps(X)=X_{\scriptscriptstyle \perp}$,
$\Pps(\omega)=\omega_{\scriptscriptstyle \perp}$,
and \ $\Pts(X)=X_t =X-X_{\scriptscriptstyle \perp}$,  \  $\Pts(\omega)= \omega_t=\omega-\omega_{\scriptscriptstyle \perp}$; 
the normal components explicitly read
\bea
\ba{l}
\omega_{\scriptscriptstyle \perp}= df^a \star K^{ab} \star 
\gm^{-1}_\star(df^b ,\omega)
= \gm^{-1}_\star(\omega,df^a) \star K^{ab} \star df^b, \\[8pt]
X_{\scriptscriptstyle \perp}=\gm_\star(X ,\Np^a) \star \E^{ab} \star\Np^b
=\Np^a \star \E^{ab} \star\gm_\star(\Np^b,X)  
\ea\label{XiOmegaPerpstar}
\eea 
\label{XiOmegaDeco'star}
in terms of the mentioned bases, twisted product and metric.
\end{prop}

\bp{}
All statements but the last one are straightforward consequences of the choice of the twist and
of Propositions \ref{XiOmegaDeco}, \ref{XiOmegaDeco'}. 
As $\Pps,\Pts$ are just the $\CC[[\nu]]$-linear extensions of $\Pp,\Pt$ (see Proposition \ref{XiOmegaDecostar}), \  then $\Pps(\omega)=\omega_{\scriptscriptstyle \perp}$,
$\Pps(X)=X_{\scriptscriptstyle \perp}$, with the right-hand sides as defined in  (\ref{XiOmegaPerp}),
and \ $\Pts(X)=X_t =X-X_{\scriptscriptstyle \perp}$,  \  $\Pts(\omega)= \omega_t=\omega-\omega_{\scriptscriptstyle \perp}$. Eq. (\ref{XiOmegaPerpstar}) holds because
any twist $\star$-product boils down to an ordinary product if one of the two factors
is $U\mathfrak{k}$-invariant, and similarly $\gm^{-1}_\star(\omega ,\omega')
=\gm^{-1}(\omega ,\omega')$,  and $\gm_\star(X ,X')=\gm(X ,X')$,
if one of the arguments is $U\mathfrak{k}$-invariant,  
by eq. (\ref{starprod}), (\ref{twistedmetric2})  (\ref{twistcond});  the order of the factors and 
of the arguments
of $\gm_\star,\gm^{-1}_\star$ can be freely changed, for the same reason and the symmetry of metric.
 \ep

An equivalent alternative to (\ref{XiOmegaPerpstar}) is
\bea
\ba{l}
\omega_{\scriptscriptstyle \perp}=\theta^a \star   \zeta_{ab}\gm^{-1}_\star(\theta^b,\omega)=  
\gm^{-1}_\star(\omega,\theta^a) \zeta_{ab}\star \theta^b,  \\[8pt]
X_{\scriptscriptstyle \perp}=\gm_\star(X ,\Up^a)\zeta_{ab}\star\Up^b=
\Up^a \star \zeta_{ab}\gm_\star(\Up^b,X),  
\ea\label{XiOmegaPerpstar''}
\eea 
where $\zeta_{ab}=\epsilon_a\delta^{ab}$. 
By the $\star$-bilinearity of $\gm_\star$ the above equations imply in particular
\bea
\ba{l}
\omega_{\scriptscriptstyle \perp}= df^a \star K^{ab} \star 
\gm^{-1}_\star(df^b ,dx^i)\star\check\omega_i
=\hat\omega_i \star  \gm^{-1}_\star(dx^i,df^a) \star K^{ab} \star df^b, \\[8pt]
X_{\scriptscriptstyle \perp}=\hat X^i\star\gm_\star(\partial_i ,\Np^a) \star \E^{ab} \star\Np^b
=\Np^a \star \E^{ab} \star\gm_\star(\Np^b,\partial_i) \star\check X^i,   
\ea\label{XiOmegaPerpstar'}
\eea 
in terms of the left  and right decompositions  \ $\omega=\hat\omega_i \star dx^i=dx^i\star\check\omega_i\in\Omega_\star$, \ $X=\hat X^i\star\partial_i=\partial_i \star\check X^i\in\Xi_\star$ 
  in the bases  $\{dx^i\}_{i=1}^n$, $\{\partial_i\}_{i=1}^n$. 
In the latter formulae one can decompose  $df^a,\Np^a,\theta^a,\Up^a$ themselves in the same way,
if one wishes.

By the previous propositions, every complete set of $\Xi_t$, e.g. $\Theta_t$, is also a complete set  of $\Xi_{t\star}$; similarly, every complete set of $\Omega_t$, e.g.  $S_V$, or $S_L$,
is also a complete set  of  $\Omega_{t\star}$.

\medskip
If the metric is Euclidean  ($g_{ij}=\delta_{ij}$) or Minkowski 
[$g_{ij}=g^{ij}= \eta_{ij}=\mbox{diag}(1,...,1,-1)$] 
\bea
\ba{ll} \gm^{-1}_\star(dx^i,df^a)= \gm^{-1}(dx^i,df^a)=f^{ai},\qquad
&\gm^{-1}_\star(df^a,dx^i)=f^{ai},\\[8pt]
\gm_\star(\partial_i ,\Np^a) =\gm(\partial_i ,\Np^a)= K^{ab}f^b_i= K^{ab}\star f^b_i,\qquad
&\gm_\star(\Np^a,\partial_i)= K^{ab}\star f^b_i;
\ea  
\eea 
replacing the right-hand sides in (\ref{XiOmegaPerpstar'}) makes the latter more explicit.

\subsubsection{Twisted differential calculus algebras  $\Q^\bullet_\star,\QMst^\bullet$}
\label{QQMstar}

The twist deformation of the differential calculus algebra $\Q^\bullet$ on $\RR^n$
introduced in section \ref{QQM} gives the one $\Q^\bullet_\star$, with the same
generators $e_\alpha,\xi^i$  and relations
\bea
&&\ba{l}
\sum_{\alpha=1}^A (\F_1\trc t_l^\alpha)\star (\F_2\trc e_\alpha)=0, \qquad l=1,...,A-n,\\[4pt]
e_\alpha \star e_\beta- (\bR_1\trc e_\beta) \star (\bR_2\trc e_\alpha)-  C^\gamma_{\star\alpha\beta}\star e_\gamma=0,\\[4pt]  
e_\alpha\star\xi^i-(\bR_1\trc \xi^i)\star (\bR_2\trc e_\alpha)=0, \\[4pt] 
\xi^i\star\xi^j+(\bR_1\trc \xi^j)\star(\bR_2\trc \xi^i)=0,                         
\ea \label{DCrel2st}\\[10pt]
&&\ba{l}
h\star \xi^i-(\bR_1\trc \xi^i) \star (\bR_2\trc h)=0, \\[4pt] 
e_\alpha \star h -\underbrace{(\bR_1\trc h)\star  (\bR_2\trc e_\alpha)}_{e_\alpha\ltlc_\star  h}- e_{\alpha\star}(h)=0,
\ea\qquad \forall h\in\X_\star,         \label{DCrel1st}  
\eea
where $C^\gamma_{\star\alpha\beta}\in\X_\star$ are defined by the decomposition 
$\left[e_\alpha,e_\beta\right]_\star\equiv\left[\bF_1\trc e_\alpha,\bF_2\trc e_\beta\right]=C^\gamma_{\star\alpha\beta}\star e_\gamma$.
$\Q^\bullet_\star$ is a $U\Xi^\f$-equivariant $\X_\star$-bimodule.
We endow $\Q^\bullet_\star$  with the  $*_\f$-structure. 

Note the change of notation: in the $\Q^\bullet_\star$ framework 
$X\star h=(\bR_1\trc h)\star (\bR_2\trc X) +X_\star(h)$, hence $X\star h$
 has a different meaning with respect to the previous sections, where
it stood just for the first term at the right-hand side, i.e. for the $\star$-product of the vector field $X$
by the function $h$ from the right; in  
the $\Q^\bullet_\star$ framework, we denote the latter by  \ $X \ltlc_\star  h:=(\bR_1\trc h)\star (\bR_2\trc X)$, \ so that we can abbreviate
$X\star h=X \ltlc_\star  h+X_\star(h)$. \
Of course  \ $(X \ltlc_\star  h)_\star(h')=[X_\star(\bR_1\trc h')]\star (\bR_2\trc h)$, \ 
 $(X \ltlc_\star  h)\ltlc_\star  h'=X\ltlc_\star  (h\star h')$ remain valid.

If one chooses $S$ so that a subset $S_t:=\{e_\alpha\}_{\alpha=1}^B$
($B\!:=\!A\!-\!k$)  
 is  complete in $\Xi_{t\star}$ (e.g. it consists of the $L_{i_1i_2...i_{k+1}}$), 
while $e_{B+a}:=V_{\scriptscriptstyle{\perp}}^a$, then, if
$\alpha,\beta\leq B$,  the sum in (\ref{DCrel2st})$_2$ is extended over 
$\gamma\!\leq\! B$.
The twisted differential calculus algebra  $\QM^\bullet{}_\star$  on $M_c$
is  the ${\cal X}^{\scriptscriptstyle M_c\star}$-bimodule generated by the $\xi^1,...,\xi^n,e_1,...,e_B$,
modulo the relations (\ref{DCrel2st}-\ref{DCrel1st}) with $\alpha,\beta\leq B$ and the ones
\bea 
&&\ba{l}
f_c^a\equiv f^a\!-\!c^a\1=0,\\[6pt] 
df^a\equiv\xi^h\star f_h^a=0,     
\ea     \qquad a=1,...,k.              \label{DCMcrelst}
\eea

\section{Examples of twisted algebraic submanifolds of $\RR^3$}
\label{Examples}

We can apply the whole machinery developed in the previous two sections
to twist deform 
algebraic manifolds embedded in $\RR^n$, provided we adopt
$\X=\mbox{Pol}^\bullet(\RR^n)$, etc. everywhere.
We can assume without loss of generality that the $f^a$ be irreducible
polynomial functions\footnote{If for some $a=1,...,k$ $f^a$ is reducible, i.e. $f^a(x)=g^a(x)h^a(x)$ with 
$g^a,h^a\in\X$ of positive degree, then 
$M=M^g\cup M^h$, where the manifold $M^g$ is  defined by 
the equations $g^a(x)=0$ and $f^h(x)=0$ if $h\neq a$, while  $M^h$
is  defined by $h^a(x)=0$ and $f^h(x)=0$ if $h\neq a$. 
If $k=1$, $f(x)=g(x)h(x)$, we find \ 
$L_{ij}=h(x)[g_i\partial_j-g_j\partial_i]+g(x)[h_i\partial_j-h_j\partial_i]$; \ \
on $M_g$ the second term vanishes and the first is tangent to $M_g$, as it must be;
and similarly on $M_h$. Having assumed the 
Jacobian everywhere of maximal rank  $M_g, M_h$ have empty intersection
and can be analyzed separately. Otherwise  $L_{ij}$ vanishes on
$M_g\cap M_h\neq\emptyset$ (the singular part of $M$),
so that  on the latter a twist built using the $L_{ij}$ will reduce to the identity, and the $\star$-product to the pointwise product (see the conclusions).}.
Following subsection \ref{SecTwistingVF},
it is interesting to ask for which algebraic submanifolds $M_c\subset\RR^n$ the infinite-dimensional Lie algebra  $\Xi_t$ 
admits a nontrivial finite-dimensional subalgebra $\g$ over $\RR$,
so that we can build concrete examples
of twisted $M_c$ by choosing a twist $\F\in (U\g\otimes U\g)[[\nu]]$ of a known type. As said, manifestly symmetric $M_c$ are of this type.

We can easily answer this question when $k=1$
and the $L_{ij}$ themselves close a finite-dimensional Lie algebra $\g$
over $\RR$. This means that  in   (\ref{comm}) $f_{ij}=$const, hence $f(x)$ is a quadratic polynomial,
and $M$ is either a quadric or the union of two hyperplanes (reducible case); 
moreover $\g$ is a  Lie subalgebra of the affine Lie algebra ${\sf aff}(n)$ of $\RR^n$.
More explicitly, if
\be
f(x)\equiv \frac 12 a_{ij}x^ix^j+a_{0i}x^i+\frac 12 a_{00}=0   \label{quadricsn}
\ee
with some real constants $a_{\mu\lambda}\!=\!a_{\lambda\mu}$ ($\mu,\lambda=0,1,...,n$),
$f_i=a_{ij}x^j\!+\!a_{i0}$,   $f_{ij} =a_{ij}$ are constant, and (\ref{comm}) has already the desired form 
\bea
[L_{ij},L_{hk}]=a_{jh}L_{ik}-a_{ih}L_{jk}-
a_{jk}L_{ih}+a_{ik}L_{jh};     \label{comm'}
\eea
$L_{ih}\trc $ act as linear transformations of the coordinates $x^k$:
\be
L_{ij}\trc x^h
=(a_{ik}x^k\!+\!a_{0i})\delta^h_j-(a_{jx}x^k\!+\!a_{j0})\delta^h_i.
 \label{L_ij_su_x^h}
\ee
By an Euclidean transformation 
(this is also an affine one)  one can always 
make the $x^i$ canonical coordinates for the quadric, so 
that $a_{ij}=a_i\delta_{ij}$ (no sum over $i$), $b_i:=a_{0i}=0$ if  $a_i\neq 0$.
In  \cite{FioFraWebquadrics} the authors first classify $\g$ and derive some general results  from the only assumptions $\X=\mbox{Pol}^\bullet(\RR^n)$ and 
$\g\subset {\sf aff}(n)$; in particular, that the {\it global} description of differential geometry on $\RR^n,M_c$ in terms of generators and relations extends to their
twist deformations, in such a way to preserve the subspaces of the differential calculus algebras
consisting of polynomials of any fixed degrees in the coordinates $x^i$, differential $dx^i$ and  vector fields chosen as generators. 
Then they analyze in detail the twisted quadrics embedded in $\RR^3$.

\medskip
Here we just present  two families of the latter as examples of 
applications of the formalism developed in the previous sections.
We analyze (referring for  details to  \cite{FioFraWebquadrics}) few twist deformations of the following classes of quadrics in $\RR^3$:  
(a) elliptic cylinders; (c) elliptic cone;
(b) 1-sheet and (d) 2-sheet hyperboloids.
As usual, we identify two quadric surfaces if they can be translated into each other via an
Euclidean transformation. By a suitable one we can make the equation $f(x)=0$
take a canonical (i.e. simplest) form, which we use to identify the class.
In Figure \ref{QuadricSummarySmall} we summarize the characterizing signs,
rank, associated symmetry Lie algebra $\g$, and  type of twist deformation that
we perform; an example in each class is plotted in Figure \ref{SomeQuadricSurfaces}. 
The elliptic cylinders of  class (a) make up a family  $M_c$ parametrized by $c>0$ (the axis of the cylinder $M_0$ is the ${\cal E}_f$ of the family), while classes (b), (c), (d)  altogether give a single family $M_c$ parametrized by $c\equiv -a_{00}\in\RR$. 
This splits into a class of connected manifolds 
(the 1-sheet hyperboloids) and a class of disconnected  ones (the 2-sheet hyperboloids and the cone, which has two nappes separated by the apex - a singular point); all are closed, except the cone, whose apex gives the ${\cal E}_f$ of the family. In either case ${\cal E}_f$ is  an algebraic variety. We note that, since
the $L_{ij}=f_i\partial_j-f_j\partial_i\in \g$ involved in the twist  vanish on  ${\cal E}_f$, the deformation automatically disappears on it, and the twisted algebraic variety is well-defined
as the undeformed. For other examples of submanifold algebras
that are not algebras of functions on smooth manifolds we refer the reader
to the recent paper \cite{Dan19}. 
We devote a subsection to each family 
and a proposition to each twist deformation; propositions are proved in
 \cite{FioFraWebquadrics}, where  twist deformations also of the
other classes of quadrics are discussed in detail.
Throughout this section the $\star$-product $X\star h$
of a vector field $X$ by a function $h$ from the right is understood in the
$\Q_\star,\QM{}_\star$ sense \ $X\star h=X \ltlc_\star  h+X_\star(h)$
\  (see section \ref{QQMstar}).

\begin{figure}[t]
\begin{subfigure}{.48\textwidth}
\centering
    \includegraphics[width=0.46\textwidth, angle=0]{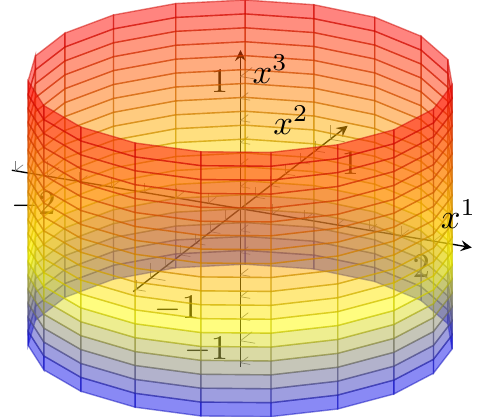}
    \caption{Elliptic cylinder with $a_1=\frac{1}{2}$, $a_2=2$} 
    \label{Elliptic cylinder}
\end{subfigure}
\begin{subfigure}{.48\textwidth}
\centering
    \includegraphics[width=0.44\textwidth, angle=0]{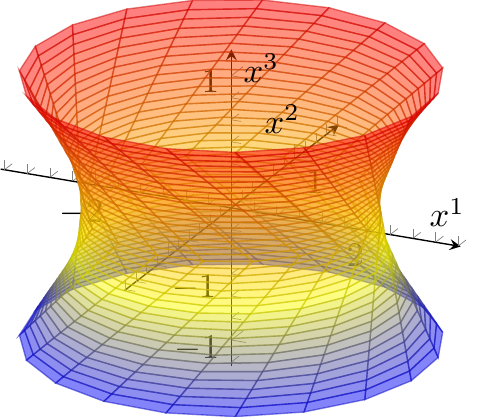}
    \caption{1-sheet hyperboloid with $a_1=\frac{1}{2}$, \\ $a_2=-a_3=2$} 
    \label{Elliptic hyperboloid of one sheet}
\end{subfigure}
\vskip.25cm
\begin{subfigure}{.48\textwidth}
\centering
    \includegraphics[width=0.46\textwidth, angle=0]{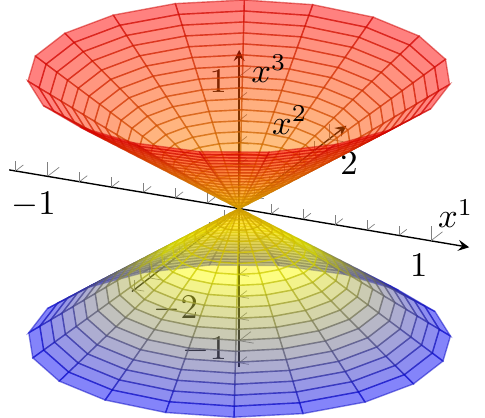}
    \caption{Elliptic cone with $a_1=-a_3=2$, $a_2=\frac{1}{2}$} 
    \label{Elliptic cone}
\end{subfigure}
\begin{subfigure}{.48\textwidth}
\centering
    \includegraphics[width=0.46\textwidth, angle=0]{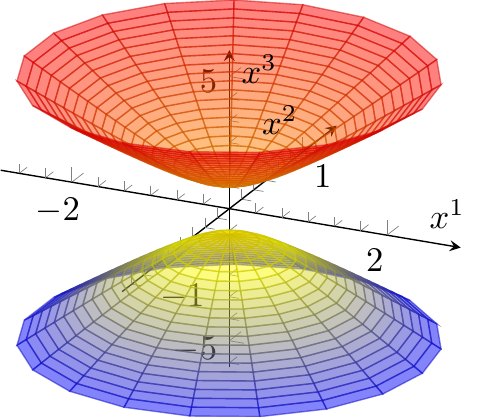}
    \caption{2-sheet hyperboloid  with $a_1=8$, $a_2=32$, \\ $a_3=-2$} 
    \label{Elliptic hyperboloid of two sheets}
\end{subfigure}
\caption{The  irreducible quadric surfaces of $\mathbb{R}^3$ that we twist-deform here.}
\label{SomeQuadricSurfaces}
\end{figure}

\begin{figure}[h!]
\begin{center}
\begin{tabular}{|c|c|c|c|c|c|c|c|c|c|c|}
\hline
& $a_1$ & $a_2$ & $a_3$ & $a_{03}$ & $a_{00}$ & $r$ &  quadric  &$\mathfrak{g}\simeq$ & Abelian & Jordanian \\
\hline
(a) & $+$ & $+$ & 0  & 0 &  $-$ & 3 & elliptic cylinder & 
\begin{tabular}{c}
$\mathfrak{so}(2) \cross \RR^2$\\ 
$\mathfrak{so}(2) \times \RR$ 
\end{tabular}
& \begin{tabular}{c}
Yes\\ 
Yes
\end{tabular}
 &  \begin{tabular}{c}
No\\ 
No
\end{tabular} \\
\hline
(b) & $+$ & $+$ & $-$  & 0 & $-$ & 4 &  1-sheet hyperboloid
& $\mathfrak{so}(2,1)$ & No & Yes \\
\hline
(c) & $+$ & $+$ & $-$  & 0 & 0 & 3 &  elliptic cone
& $\mathfrak{so}(2,\!1)\!\times\!\RR$ & Yes & Yes \\
\hline
(d) & $+$ & $+$ & $-$  & 0 &+& 4 & 2-sheet hyperboloid
& $\mathfrak{so}(2,1)$ & No & Yes \\
\hline
\end{tabular}
\end{center}
\caption{Signs of the coefficients  of the canonical equations, 
rank,  associated symmetry Lie algebra $\g$, type of twist deformation.
The cone (c) consists of two open components (nappes) disconnected by the apex
(a singular point);
we build an abelian twist for it using also the generator of dilatations.}
\label{QuadricSummarySmall}
\end{figure}

\subsection{(a)   Family of elliptic cylinders in Euclidean $\RR^3$}

Their equations in canonical form (with $a_1\!=\!1$, $a_3\!=\!a_{0i}\!=\!0$) 
are parametrized by \
$c\equiv -a_{00}>0$, $a\equiv a_2>0$ and read
\be\label{EllCyleq}
    f_c(x):=\frac{1}{2}\big[(x^1)^2+a(x^2)^2\big] -c=0.
\ee
For every $a>0$, $\{M_c\}_{c\in\RR^+}$ is a foliation of $\RR^3\setminus \vec{z}$,
where $\vec{z}$ is the axis of equations $x^1\!=\!x^2\!=\!0$.
The  vector fields $L_{12}=x^1\partial_2-ax^2\partial_1$, 
$L_{13}=x^1\partial_3$, $L_{23}=a x^2\partial_3$ generate   \ $\g\simeq \mathfrak{so}(2)\cross \RR^2$: 
\bea\label{eq04'}
[L_{12},L_{13}]=-L_{23}, \qquad [L_{12},L_{23}]=aL_{13}, 
\qquad  [L_{13},  L_{23}]=0. 
\eea
The commutation relations \
 $[L_{ij},x^h]=L_{ij}\trc x^h$, \ $[L_{ij},\partial_h]=L_{ij}\trc \partial_h$, \
$[L_{ij},\xi^h]=0$ \ hold in $\Q^\bullet$. \

\begin{prop}
$\F=\exp(i\nu L_{13}\otimes L_{23})$ is a unitary abelian twist 
 inducing  a twist deformation of $U\g$,  of $\Q^\bullet$ on $\RR^3$
and of $\QM^\bullet$  on the elliptic cylinders (\ref{EllCyleq}).
The basic relations  characterizing the $U\g^\f$-module $*$-algebra  $\QMst^\bullet$
keep the same form, in particular 
\bea \label{characterizingEllCyl}
f_c(x)\equiv\frac{1}{2}(x^1\!\star\! x^1\!+\!a x^2\!\star\! x^2)\!-\!c=0,\quad
df\equiv \xi^1\!\star\! x^1\!+\!a\, \xi^2\!\star\! x^2=0,
\quad \epsilon^{ijk} f_i \star L_{jk}=0.
\eea
\end{prop}

\noindent
Alternatively, as a complete set in $\Xi_t$ instead of $\{L_{12},L_{13},L_{23}\}$
we can use $S_t=\{L_{12},\partial_3\}$, which is actually a basis of $\Xi_t$;
the Lie algebra $\g\simeq \mathfrak{so}(2)\times \RR$ generated by the latter is abelian.

\begin{prop}\label{prop03}
$\F=\exp(i\nu \partial_3\otimes L_{12})$   is a unitary abelian twist
 inducing  a twist  deformation of $U\g$,  of $\Q^\bullet$ on $\RR^3$
and of $\QM^\bullet$  on the elliptic cylinders (\ref{EllCyleq}). 
The basic relations  characterizing the $U\g^\f$-module $*$-algebra  $\QMst^\bullet$
keep the same form, in particular 
\bea
f_c(x)\equiv\frac{1}{2}(x^1\!\star\! x^1+a x^2\!\star\! x^2)-c=0,\quad
df_c\equiv \xi^1\!\star\! x^1+a\, \xi^2\!\star\! x^2=0, \quad \epsilon^{ijk} f_i \star L_{jk}=0.
\eea
\end{prop}

In \cite{FioFraWebquadrics} the reader can find the  relations characterizing these two deformations in detail.

In the case  $a_1\!=\!a_2\!=\!1$
\ (circular cylinder of radius $R=\sqrt{2c}$ embedded in the Euclidean $\RR^3$) 
$S:=\{L,\partial_3,\Np\}$ is an orthonormal basis of $\Xi$  alternative to $S'\!:=\!\{\partial_1,\partial_2,\partial_3\}$
and such that $S_t\!:=\!\{L,\partial_3\}$, $S_{\scriptscriptstyle \perp}\!:=\!\{\Np\}$
are orthonormal bases of $\Xi_t$, $\Xi_{\scriptscriptstyle \perp}$,  respectively;
here  $L:=L_{12}/R$, $\Np=f^i\partial_i/R=(x^1\partial_1+x^2\partial_2)/R$. 
The Killing Lie algebra $\k$ is abelian and spanned (over $\RR$) by $S_t$. \ 
 $\nabla_XY=0$ for all $X,Y\!\in\! S'$, whereas the only non-zero $\nabla_XY$, with 
$X,Y\in S$ are
\be
\nabla_LL=-\frac 1 R \Np,\qquad\nabla_L\Np=\frac 1 R L,\qquad
\nabla_{\Np}L=\frac 1 R L,\qquad\nabla_{\Np}\Np=\frac 1 R\Np.
\ee
The second fundamental form is  explicitly given by \ $II(X,Y)=\left( \nabla_X Y\right)_{\perp}=-\frac{\tilde X\,\tilde Y }{R}\,\Np$, \   for all $X,Y\in \Xi_t$; \ 
here we are using the decomposition 
$ Z=\tilde Z L+Z^3\partial_3$ of the generic $Z\in \Xi_t$.
Thus $II$ is diagonal in the basis $S_t$,
with diagonal elements (i.e. principal curvatures) \ $\kappa_1=0, \kappa_2=-1/R$. \
Hence the Gauss (i.e. intrinsic) curvature $K=\kappa_1\kappa_2$ vanishes; 
$\rR_t=0$ easily follows also from $\rR=0$ using the Gauss theorem. The mean (i.e. extrinsic) curvature is $H=(\kappa_1\!+\!\kappa_2)/2=-1/2R$. \ 
The Levi-Civita covariant derivative $\nabla\!_{t,X}$ on $M_c$ is 
$$
\nabla\!_{t,X}Y=\mathrm{pr}_t(\nabla_XY)=\nabla_XY-II(X,Y)=\nabla_XY +\tilde X\,\tilde Y \,\Np/R.
$$
The deformation via the abelian twist $\F=\exp(i\nu \partial_3\otimes  L_{12})$  yields \ $\nabla^\f_{X} =\nabla^{}_X$ \  for all \ $X\in S\cup S'=
\{\partial_1,\partial_2,\partial_3,L,\Np\}$, \
$\nabla\!_{t,X}^{\,\f}Y = \mathrm{pr}_t(\nabla_XY)=\nabla\!_{t,X}Y$ \ for all
\ $X,Y\in\{\partial_3,L\}$, \
because $\partial_3$ commutes with all such $X$,
so that $\bF_1\trc X\otimes\bF_2=X\otimes\1$, \ and the projections $\mathrm{pr}_{\scriptscriptstyle \perp},\mathrm{pr}_t,$ stay undeformed, as shown in 
Proposition \ref{prop08}.  \ These results determine $ \nabla^\f_{X}Y$
for all $X,Y\in\Xi_\star$ and  $\nabla\!_{t,X}^{\,\f}Y=\nabla\!_{t,X}Y$ for all $X,Y\in\Xi_{t\star}$ via  the function left $\star$-linearity
in $X$ and the deformed Leibniz rule for $Y$. 
The twisted curvatures $\rR^{\f},\rR^{\f}_t$ vanish.
Furthermore, for all $X,Y\in S_t$
\begin{equation}\label{II^F=II}
    II^\f_\star(X,Y)
    \stackrel{(\ref{II^F})}{=}II(\mathcal{F}_1^{-1}\trc X,\mathcal{F}_2^{-1}\trc Y)
    =\gm(\nabla_{\f_1^{-1}\trc X}(\mathcal{F}_2^{-1}\trc Y),\Np)
    \Np    =II(X,Y),
\end{equation}
i.e. the principal curvatures  
\ $\kappa_1=0, \kappa_2=1/R$, \
Gauss and mean curvatures are  undeformed.

\subsection{(b-c-d)  Family of  hyperboloids and cone in Minkowski $\RR^3$}

Their equations in canonical form (with $a_1\!=\!1$) are parametrized by
$a=a_2>0$, $b=-a_3>0$, $c=-a_{00}$ ($c>0$, $c<0$ resp. for the  $1$-sheet  
and the $2$-sheet hyperboloids, $c=0$ for the cone) and read
\begin{equation}\label{eq19}
    f_c(x) :=\frac{1}{2}[(x^1)^2+a(x^2)^2-b(x^3)^2]-c=0.
\end{equation}
For all $a,b>0$, $\{M_c\}_{c\in\RR\setminus\{0\}}$ is a foliation of $\RR^3\setminus M_0$,
where $M_0$ is the cone of equation $f_0=0$.
The Lie algebra  $\g$ is spanned by  $L_{12}=x^1\partial_2-ax^2\partial_1$, 
$L_{13}=x^1\partial_3+bx^3\partial_1$,
$L_{23}=ax^2\partial_3+bx^3 \partial_2$, which  fulfill \
$[L_{12},L_{13}]=-L_{23}$, \ $ [L_{12},L_{23}]=aL_{13}$, \ 
$[L_{13},L_{23}]    =bL_{12}$. \
Setting $H:=\frac{2}{\sqrt{b}}L_{13}$,
$E:=\frac{1}{\sqrt{a}}L_{12}+\frac{1}{\sqrt{ab}}L_{23}$
and $E':=\frac{1}{\sqrt{a}}L_{12}-\frac{1}{\sqrt{ab}}L_{23}$,
we obtain 
\be
[H,E]=2E,\qquad [H,E']=-2E',\qquad [E,E']=-H,       \label{so(2,1)}
\ee 
showing that the corresponding symmetry Lie algebra is
$ \mathfrak{g}\simeq\mathfrak{so}(2,1)$. 
The commutation relations \
 $[L_{ij},x^h]=L_{ij}\trc x^h$, \ $[L_{ij},\partial_h]=L_{ij}\trc \partial_h$, \
$[L_{ij},\xi^h]=0$ \ hold in $\Q^\bullet$.
To compute the action of $\mathcal{F}$ on  functions it is convenient 
to adopt as new coordinates the eigenvectors of $H\trc$ 
\be
    y^1    =x^1+\sqrt{b}x^3,\qquad
    y^2=x^2,\qquad
    y^3=x^1-\sqrt{b}x^3;
\ee
with eigenvalues $\lambda_1=2$, $\lambda_2=0$ and $\lambda_3=-2$. \ 
 Also  $\eta^i\equiv dy^i, \:\tilde{\partial}_i\equiv\frac{\partial}{\partial y^i}$
are eigenvectors of $H\trc$.
%
\begin{prop}
 $\F=\exp(H/2\otimes\log(\1+i\nu E))$ is a unitary twist
 inducing a twist deformation of $U\g$,
 of $\Q^\bullet$ on $\RR^3$
and of $\QM^\bullet$  on the elliptic hyperboloids and cone (\ref{eq19}). 
The basic relations  characterizing the $U\g^\f$-module 
$*$-algebra  $\QMst^\bullet$ are  
\be  \label{QMstarEllHyp}
\ba{lll}
    0 &=&f_c(y)\equiv
    \frac{1}{2}y^3\star y^1+\frac{a}{2}y^2\star y^2-c,\\[6pt]
  0&=& \mathrm{d}f =\frac{1}{2}(y^3\star\eta^1+\eta^3\star y^1)    +ay^2\star\eta^2,\\[6pt]
  0&= &y^3\star E-y^1\star E'-\sqrt{a}\,y^2\star H+i\nu y^1\star H-2i\nu (1+i\nu) y^1\star E.
\ea\ee
%
\end{prop}
In Ref. \cite{FioFraWebquadrics} the reader can find, first the detailed actions of $H,E,E'$ on  $y^i,\partial_i, \eta^i$, then the twisted coproducts and antipodes of $H,E,E'$,  the star products and commutation relations among the generators $H,E,E', y^i,\partial_i, \eta^i$ arising from this twist.

\medskip
Let us now focus on the case $1=a_1=a=b$, i.e.
$f_c(x)=\frac{1}{2}((x^1)^2+(x^2)^2-(x^3)^2)-c$. This covers the
circular cone, the circular hyperboloid of one and two sheets.
We endow $\RR^3$ with the Minkowski metric $\gm:=\eta_{ij}\mathrm{d}x^i\otimes\mathrm{d}x^j=\mathrm{d}x^1\otimes\mathrm{d}x^1
+\mathrm{d}x^2\otimes\mathrm{d}x^2-\mathrm{d}x^3\otimes\mathrm{d}x^3$,
whence $\gm(\partial_i,\partial_j)=\eta_{ij}$. \ $\gm$
is equivariant with respect to $U\mathfrak{g}$, where
$\mathfrak{g}$ is the Lie $^*$-algebra spanned by the vector fields
$L_{ij}$, tangent to $M_c$. $\E =\gm^{-1}(df,df)=x^ix_i=2c$, which vanishes for $c=0$.
Therefore ${\cal D}_f'=\RR^3\setminus M_0$ ($M_0$ is the cone).
The induced metric on the remaining $M_c\subset{\cal D}_f'$ (or first fundamental form) $\gm_t$ makes 
$M_c$ Riemannian if $c<0$, Lorentzian if $c>0$ (whereas it is degenerate on the cone $M_0$); moreover, in any basis $S_t:=\{v_1,v_2\}$ of $\Xi_t$ we find
\ $II(v_\alpha,v_\beta)=-g_{\alpha\beta}\,\Vp/2c$, \ 
where  $\alpha,\beta\in\{1,2\}$, $g_{\alpha\beta}:=\gm(v_\alpha,v_\beta)$, 
$\Vp=f_j\eta^{ji}\partial_i=x^i\partial_i$. 
We find the components of the curvature,  Ricci tensors,  the Ricci scalar  
 (or {\it  Gauss curvature}) on $M_c$
\bea
\rR_t{\,}^{\delta}_{\alpha\beta\gamma}=
\frac{g_{\alpha\gamma}\delta_\beta^\delta-g_{\beta\gamma}\delta_\alpha^\delta}{2c},\qquad \ric_{t}{}_{\beta\gamma}=\rR_{t}{\,}_{\alpha\beta\gamma}^{\alpha}
=-\frac {g_{\beta\gamma}}{2c},\qquad \mathfrak{R}_{t}
=\ric_{t}{}_{\beta\gamma}g^{\beta\gamma}=-\frac 1{c}  \label{curvatures}
\eea
applying the Gauss theorem.
All diverge as $c\to 0$ (i.e. in  the cone $M_0$ limit). $M_c$ is therefore
anti-de Sitter space $AdS_2$ if $c>0$, the union of two copies of de Sitter space
$dS_2$ if $c<0$.

Under twist deformation the curvature (and Ricci) tensor on $\RR^3$ remain zero.
By propositions \ref{prop08}, \ref{prop02}, \ref{lem01star} on $M_c$  
the first and second fundamental forms, as well as the
 curvature and Ricci tensors, stay undeformed as elements of the corresponding
tensor spaces:   $\gm^\f_{t}=\gm_t\in\Omega_t\otimes\Omega_t[[\nu]]$, ... \ 
Only the associated multilinear maps of twisted tensor products 
$\gm_{t\star}:\Xi_{t\star}\ots \Xi_{t\star}\to\X_\star$, ...,
 `feel' the twist; they are related to the undeformed maps
through formulae (\ref{twistedmetric2}), (\ref{II^F}) and
\be
\ric^\f_{t\star}(X,Y)=\ric_t\left(\bF_1\trc X,\bF_2\trc Y\right)
\ee 
[compare also to \cite{AschieriCastellani2009}~Theorem~7 and eq.~(6.138)],
and similarly for $\rR^\f_{t\star}$.
Also the Ricci scalar (or Gauss curvature) $\mathfrak{R}_t^\f$ remains
the undeformed one $-1/c$.

\medskip
Finally, one can elaborate \cite{FioFraWebquadrics} also abelian twist deformations for the elliptic cone, (\ref{eq19}) with $c=0$, enlarging the Lie algebra $\g$
by a generator $D=
x^i\partial_i=y^i\tilde{\partial}_i$ of dilatations, which commutes with all $L_{ij}$ and
is also tangent to the cone (only),  as $D(f)=2f$: $D\in\XiM$.

\section{Outlook and final remarks}
\label{Outlook}

Considering a generic
embedded submanifold $M\subset\RR^n$ that consists of  the solutions $x$
of a set of $k$ equations $f^a(x)=0$ ($a=1,...,k$), where 
$f:{\cal D}_f\subset\RR^n\to \RR^k$ is a $k$-ple of smooth  functions with Jacobian matrix of rank $k$, in this work we have 
explicitly built its noncommutative analogue in the framework of Drinfel'd (cocycle) twist \cite{Drinfeld1983}
deformation of differential geometry  \cite{Aschieri2006,AschieriCastellani2009}.
This can be considered as a successful result, also in the broader framework of
deformation quantization \cite{BayFlaFroLicSte,Kon97}, in the sense that the deformed
algebra $\X_\star$ of functions on ${\cal D}_f$ and the one $\XM_\star$  of functions 
on $M$ both coincide with their 
undeformed counterparts $\X[[\nu]],\XM[[\nu]]$ as $\CC[[\nu]]$-modules ($\nu$ is the deformation parameter), because the same occurs for the ideal $\C$ of functions vanishing on $M$ and its deformed counterpart, $\C_\star$, \ $\C_\star=\C[[\nu]]$; \ only the
pointwise product is deformed into a (in general noncommutative) one, the
$\star$-product. In other words, taking the quotient and performing the deformation commute:
$\XM_\star=(\X/\C)_\star=\X_\star/\C_\star$. \  
The key point has been to perform the deformation 
using a Drinfeld twist $\F$ based on the Lie subalgebra $\Xi_t$ (\ref{defeqXis}) 
of vector fields on ${\cal D}_f$ that are tangent to {\it every} manifold $M_c$ of the family
of level sets of $f$ (the latter is parametrized by $c\in f\left({\cal D}_f\right)\subset\RR^k$; $M_c$ consists of the solutions of $f^a(x)-c^a=0$, 
$a=1,...,k$), rather than on the Lie algebra $\XiM$ of vector
fields tangent  to $M$ only; this has given for free the deformation of the whole
family by the same twist. Every vector field in the $\star$-Lie algebra $\XiMst$
can be represented by an element
of the $\star$-Lie algebra $\Xi_{t\star}$, as it occurs in the undeformed case.
The whole  twisted Cartan calculus is automatically equivariant  under the non-cocommutative Hopf
algebra $U\Xi_t^\f$; the latter may be interpreted as the quantum group of (small)
diffeomorphisms of the deformed submanifolds.
The dimensions of $\Xi_{\star}, \Xi_{t\star}$ as $\X_\star$ bimodules, as well as 
of their duals $\Omega_\star,\Omega_{t\star}$, 
remain undeformed, contrary to what happens to the quantum group
bicovariant or equivariant differential calculi mentioned in the introduction. This is because we consider 2-cocycles twists, but could change
with more general twists leading to quasitriangular Hopf algebras or quasi-Hopf algebras,
or twists in the category of bialgebroids.
We have also shown that,
when $\RR^n$ is endowed with a connection $\nabla$, 
taking the tangent projection (from $\RR^n$ to $M$) of $\nabla$ and the associated torsion, curvature, commutes with performing the deformation, provided  $\F$ is based on the equivariance Lie subalgebra $\mathfrak{e}\subset\Xi_t$  
[see (\ref{equivLieAlg})].
When $\RR^n$ is endowed with a metric $\gm$  the same holds for 
$\gm$ itself, the associated  Levi-Civita connection, the intrinsic and extrinsic
curvatures  (while the torsions remain zero), only if  $\F$ is based on the  Lie algebra $\k\subset\Xi_t$ of Killing vector fields of the metric.

All our results are global, in that we have determined global (i.e., defined on all of $M$) bases - or complete sets - of all the relevant 
$\X_\star$- and $\XM_\star$-bimodules from their undeformed counterparts: $\C_\star$ is spanned by the globally defined functions $f^a$, $\Xi_{t\star}$ by some 
complete set $\{e_\alpha\}$
of globally defined 
vector fields [e.g. (\ref{genL})]; these  fulfill some linear dependence relations], the $\X_\star$-bimodule $\Xips\subset\Xi_\star$ of twisted vector fields normal to the $M_c$'s (with respect to the metric $\gm$) is spanned by the globally defined 
vector fields (\ref{DefNpa}), and similarly the dual
ones $\Omega_{t\star},\Omega_{{\scriptscriptstyle \perp}\star}$ of 1-forms, their tensor or wedge powers,.... This means that both in the undeformed and deformed context
these bimodules/algebras can be formulated in terms of (the mentioned) generators and polynomial relations, with elements in $\X,\X_\star$ as coefficients. 

In the polynomial setting, if the polynomial functions $f^a(x)$ 
fulfill suitable irreducibility conditions, then
also $\X,\X_\star,\XM,\XM_\star$ can be defined in terms of  generators $x^i$  (the Cartesian coordinates) and  polynomial relations \cite{FioFraWebquadrics}. The procedure can be potentially applied to a large number of algebraic manifolds, starting from
algebraic hypersurfaces ($k=1$),  in particular quadrics; 
one can use the examples of cocycle twists available in the literature (tipically based on finite-dimensional Lie algebras $\g$) to build concrete  deformations of these submanifolds. In \cite{FioFraWebquadrics}  
the authors discuss in detail deformations of all families
of quadric surfaces embedded in $\RR^3$ that are induced by unitary twists of
the abelian \cite{Reshetikhin1990} or Jordanian \cite{Ogievetsky1992,Ohn1992} type,
except the ellipsoids.  Here (section \ref{Examples}) we have only presented the results for the
elliptic (in particular, circular) cylinders, hyperboloids and cone.
Endowing $\RR^3$ with the Euclidean (resp. Minkowski) metric
we have found twisted circular (i.e. maximally symmetric) cylinders, hyperboloids and cone  $M_c$ that are (pseudo)Riemannian and equivariant  under 
a non-trivial Hopf algebra $U\mathfrak{k}^{\f}$ (``quantum group of isometries");
the twisted Levi-Civita connection on all $M_c$ equals the projection of the twisted  Levi-Civita connection on $\RR^3$ (the exterior derivative), 
while the twisted intrinsic curvature can be expressed 
in terms of the twisted second fundamental form (or extrinsic curvature)
via the twisted Gauss theorem;
the twisted  curvatures are the same constants as their undeformed counterparts. The twisted hyperboloids with $c<0$ (resp. $c>0$)
can be thus considered as twisted de Sitter spaces $dS_2$  (resp.  anti-de Sitter spaces $AdS_2$).

We recall that the higher-dimensional generalizations of the latter manifolds 
play a prominent role in present cosmology and theoretical physics as  maximally symmetric cosmological solutions to the Einstein field equations of general relativity with a nonzero cosmological constant $\Lambda$; in particular, de Sitter spacetime ($\Lambda>0$) 
can describe a universe with accelerating expansion rate (see e.g. \cite{Dod03}), while anti-de Sitter  spacetimes 
 ($\Lambda<0$) are at the base  of the socalled Ads/CFT correspondence \cite{Mal98}.
Interpreting Minkowski $\RR^{2+1}$ as a relativistic momentum, rather than
position, space ($x^1,x^2$ playing the role of
components of the momentum, $x^3$ of energy), the equations   (\ref{eq19}) as dispersion relations
for relativistic particles, and performing the deformations, we should regard
the $x^3>0$ component of the twisted 2-sheet hyperboloid ($c<0$) 
as the twisted mass shell of a particle of mass $\sqrt{|2c|}$;
similarly, the $x^3>0$ nappe of the cone $c=0$ would do for a massless particle,
while $c>0$ would do for a tachyon. 

Generalizing the framework to submanifolds of $\CC^n$
looks straightforward and should make things even simpler, as we
drop  $*$-structures and the related constraints on the twist. For instance,
there are no abelian twist deformations of the ellipsoids in $\RR^3$, because
the corresponding $\g\simeq\mathfrak{so}(3)$ is simple;
neither are there Jordanian ones, because $\mathfrak{so}(3)$  (over $\RR$)
contains no elements $E,H$ fulfilling $[H,E]=2E$.
If  we considered $f(x)\equiv x^ix^i-1$ as a polynomial function $f:\CC^3\to\CC$,
 then such $E,H\in\g\simeq\mathfrak{so}(3,\CC)$ would exist, and we could perform a Jordanian deformation
also of the complex ellipsoid $M\subset\CC^3$  solution of $f(x)=0$.

\smallskip
Finally, 
in \cite{FioPis18,FioPis19JPA,FioPis19LMP} an alternative approach to introduce
noncommutative (fuzzy) embedded submanifolds $S$ in $\RR^n$
was proposed and applied to the spheres;
it is obtained projecting the algebra of observables of a quantum particle
in $\RR^n$, in a confining potential with a very sharp
minimum on $S$, to the Hilbert subspace with energy below a certain cutoff.

\bigskip\noindent
{\bf Acknowledgments.} 
We are indebted with F. D'Andrea for critical advice, fruitful discussions  and suggestions at all stages of the work. We thank D. Franco for stimulating remarks, P. Aschieri for clarifying some crucial aspects of Ref.
\cite{Aschieri2006,AschieriCastellani2009}. T. Weber would like
to thank P. Aschieri also for his kind hospitality at the
University of Eastern Piedmont "Amedeo Avogadro".

\section{Appendix}

\subsection{Proof of Proposition \ref{DecoC}}

As the inclusion $\C\supset\bigoplus_{a=1}^k \X f^a$ is trivial, we need to prove the converse one.
For all $\bar x\in M$ we can find 
a local smooth change of coordinates $\phi:x\in V_{\bar x}\mapsto z\in U_{\bar z}$ of the form $\phi(x)=(f,y)\equiv(f^1,...,f^k,y^1,...,y^{n-k})$, where 
$\bar z\equiv\phi(\bar x)=(0_k,\bar y)$  ($0_k$ stands for the row with $k$ zeroes),
$U_{\bar z}\subset\RR^n$ is an open ball with center $\bar z$, and
$V_{\bar x}=\phi^{-1}\left(U_{\bar z}\right)\subset\RR^n$; one can choose the extra coordinates $y^h$ e.g. as a subset of the $x^j$ themselves\footnote{
Consider in fact the set of equations in the variables $(x,c)\in\RR^{n+k}$
\be
l^a(x,c):= f^a(x)-c^a=0 ,\qquad a=1,2,...,k<n.   \label{eql}
\ee
The Jacobian matrix of $l=(l^1,...,l^k)$ is
the $k\times(n\!+\!k)$-matrix $(J|-I_k)$, where $J=\partial f/\partial x$ has rank $k$.   $M$ consists of the  points $x$ such that $(x,0_k)$ solves (\ref{eql}). Fixed a $\bar x\in M$, we can always permute the coordinates so that 
the $k\times k$-matrix $A:=\left(\partial f^a/\partial x^b\right)_{a,b=1}^k$ is invertible in $\bar x$.
By the implicit function theorem there exists an open ball $U_{\bar z}\subset\RR^n$ centered at $\bar z:=(0_k,\bar x^{k+1},...,\bar x^n)$  and  smooth functions 
 $x^a(z)$ of
$ z:=(c^1,...,c^k,x^{k+1},...,x^n)\in U_{\bar z}$ such that $x^a(\bar z)=\bar x^a$, and 
$l\big(x^1(z),...,x^k(z),x^{k+1},...,x^n,c^1,...,c^k\big)=0$; thus we can set \ $y^1= x^{k+1}$, ..., $y^{n-k}= x^n$. 
}.
For all $h\in\X$ the function defined on $U_{\bar z}$ by $\hat h(z)=h(x)$ is 
smooth as well. 
In terms of the new coordinates the points of $V_{\bar x}\cap M$ belong to the 
hyperplane $z^1=...=z^k=0$. For all $z=(c,y)\in U_{\bar z}$
we denote as $z':=(0_k,y)$ its projection on this hyperplane; the segment $zz'$ is contained in the ball $U_{\bar z}$.
Applying Hadamard's lemma 
to the dependence of $\hat h(z)$ on the first $k$ coordinates [considering $y$ as parameters] we find 
$\hat h(z)=\hat h(z')+\sum_{a=1}^kc^a\hat h^a(z)$  in  $U_{\bar z}$,
with smooth $\hat h^a$; more explicitly, 
$$
\hat h^a(z)=\int^1_0 \frac{\partial \hat h}{\partial z^a}(tc,y)\: dt.
$$
Equivalently,  $h(x)=h(x')+\sum_{a=1}^kf^a(x)h^a_{\bar x}(x)$ in  $V_{\bar x}$, where
$x'=\phi^{-1}(z')\in V_{\bar x}\cap M$ and $h^a_{\bar x}$ are  defined 
by $h^a_{\bar x}(x)=\hat h^a(z)$ and smooth in 
$V_{\bar x}$. 
If $h\in\C$ then $h(x')=0$,  and $h(x)=\sum_{a=1}^kf^a(x)h^a_{\bar x}(x)$.
 This is the desired decomposition, but only locally. To make it
global, consider the open cover of  ${\cal D}_f$
$$
\O=\O'\cup \{V_1\}\cup ...\cup\{V_k\}, \qquad \O':=\{V_{\bar x}\: |\:\bar x\in M\},\quad
V_a:={\cal D}_f \setminus \overline{M_a},
$$
where $\overline{M_a}$ is the closure of the hypersurface $M_a$, which is the level
set of $f^a$ ($M=\bigcap_{a=1}^kM_a$). Since ${\cal D}_f$  is paracompact (as so is the metric space $\RR^n$), there is a smooth partition of unity 
subordinated to $\O$, i.e. there exist: a function $\rho_a$  with support contained in 
$V_a$, for all $a\in\{1,...,k\}$, and 
a function $\rho_{\bar x}\in\X$  with support contained in $V_{\bar x}$,
for all $\bar x\in  M$, such that  for all $x\in {\cal D}_f$ \
$\sum_{\bar x\in  M}\rho_{\bar x}(x)+\sum_{a=1}^k\rho_a(x)=1$, with only a finite
number of non-zero terms in the sum.  The functions
$$
\tilde h^a_{\bar x}(x):=\left\{\!\ba{lr}
h^a_{\bar x}(x)\rho_{\bar x}(x) &\mbox{if }x\in V_{\bar x},\\[6pt]
0 &\mbox{if }x\in \RR^n\setminus V_{\bar x},\ea\right. \qquad 
\tilde h^a(x):=\left\{\!\ba{lr}
\displaystyle h(x)\frac{\rho_a(x)}{f^a(x)} &\mbox{if }x\in V_a,\\[4pt]
0 &\mbox{if }x\in \RR^n\setminus V_a,\ea\right.
$$
belong to $\X$ and fulfill \ $\sum_{a=1}^kf^a(x)\tilde h^a_{\bar x}(x)=
h(x)\rho_{\bar x}(x)$, \ $f^a(x)\tilde h^a(x)= h(x)\rho_a(x)$; \
hence \ $h^a:=\sum_{\bar x\in  M}\tilde h^a_{\bar x}+
\tilde h^a\in\X$ are  the coefficients needed for (\ref{DecoCeq}) to hold. 
In fact,  for all $x\in {\cal D}_f$
$$
\sum_{a=1}^k h^a(x)f^a(x)=\sum_{a=1}^k f^a(x)h^a(x)=
h(x)\left[\sum_{\bar x\in M}\rho_{\bar x}(x)+\sum_{a=1}^k\rho_a(x)\right]
=h(x).    \label{DecohinC}
$$

\subsection{More on twists}

We write in a compact notation
the  inverse of (\ref{cocycle}) and its consequences
\be\ba{ll}
\bF_{(12)3}\bF_{12}=\bF_{1(23)}\bF_{23}, \qquad\qquad &
\bF_{(123)4}\bF_{(12)3}=\bF_{(12)(34)}\bF_{34},\\[8pt]
\bF_{(123)4}\bF_{1(23)}=\bF_{1(234)}\bF_{(23)4},\qquad\qquad &
\bF_{(12)(34)}\bF_{12}=\bF_{1(234)}\bF_{2(34)},\ea \label{bla}
\ee
obtained applying $\Delta$ on the first, second, third tensor factor
and recalling that $\Delta$ is cocommutative; the bracket encloses tensor factors
obtained from one by application of $\Delta$.   To denote the
decomposition of $\F_{(12)3}$ we have used a Sweedler-type notation
$
\F_{(12)3}\equiv (\Delta\ot\id)(\F)=\F_{1(1)}\ot \F_{1(2)}
\ot\F_2,
$
and similarly for $\F_{1(23)},\bF_{(12)3}...$. Several proofs are based on these 
formulae.

\subsection{Isomorphism of twisted Hopf $*$-algebras for unitary twists}
\label{AppendixD}

We use the notation of section \ref{TwistAlgStruc}.
Assume that $(H,*)$ is a Hopf $*$-algebra. We now prove 

\begin{prop}
If $\F$ is unitary, then $D:(H_\star,*_\star)\to(H^\f,*)$ is an isomorphism of  triangular
Hopf $*$-algebras; in particular, $D\circ *_\star= *\circ D$.
 \label{IsomorHopf}
\end{prop}

\begin{proof}{}  Via (\ref{bla}) one can prove the relation (see e.g.  Lemma 2.2. in \cite{Majid1994} or eq. (126) in \cite{Fiore2010})
\be 
\Delta(\beta)=\F^{-1}(\beta \ot \beta )[(S\ot S)\F^{-1}_{21}]
=\F^{-1}_{21}(\beta \ot \beta )[(S\ot S)\F^{-1}].   \label{deltabeta}
\ee
$D\circ *_\star= *\circ D$ is almost the same as eq. (31) in \cite{Fiore2010}. 
We prove it again using (\ref{deltabeta}):
\bea
D(\xi^{*_\star}) \!&\!=\!&\! D\left[S(\beta)\trc \xi^*\right]= 
D\left[S\!\left(\beta_{(1)}\right)  \xi^* \beta_{(2)} \right]=
\F_1 S\left(\beta_{(1)}\right)  \xi^* \beta_{(2)}  S(\F_2) \beta^{-1}\stackrel{(\ref{deltabeta})}{=}\! S(\beta)S(\bF_2)\, \xi^*\,\bF_1\nn
\!&\!=\!&\! S(\beta)\left[\F_1\, \xi \,S(\F_2)\right]^*
=S(\beta)\left[D( \xi)\beta\right]^*=S(\beta)S\big(\beta^{-1}\big)\left[D( \xi)\right]^*
=\left[D( \xi)\right]^*. 
\eea
As particular consequences,
$\Delta_\star\circ *_\star=(*_\star\otimes *_\star)\circ\Delta_\star $, 
$S_\star\circ *_\star\circ S_\star\circ *_\star=\id$ follow from
$\Delta_\f\circ *=(*\ot *)\circ\Delta_\f$, 
$S_\f\circ *\circ S_\f\circ *=\id$.
\end{proof}

\subsection{Proof of Proposition \ref{TwistedNabla}}

First of all, $\nabla^{\f}$ is well-defined, since
$U\g[[\nu]]\trc\TT^{p,q}_\star\subseteq\TT^{p,q}_\star$.
Eq. (\ref{eq04}) easily follows from the properties of the classical covariant derivative:
$\nabla^{\f}_Xh
    =\nabla_{\overline{\f}\!_1\trc X}(\overline{\mathcal{F}}\!_2\trc h)
    =\mathcal{L}_{\overline{\mathcal{F}}\!_1\trc X}(\overline{\mathcal{F}}\!_2\trc h)
    =\mathcal{L}^\star_X h.
$
Furthermore, for every $g\in U\mathfrak{e}^\f$ we obtain
\begin{align*}
    g\trc(\nabla^\f_XT)
    =\nabla_{(g_{(1)}\overline{\f}_1)\trc X}((g_{(2)}\overline{\F}_2)\trc T)
    =\nabla_{(\overline{\f}_1g_{\widehat{(1)}})\trc X}((\overline{\F}_2g_{\widehat{(2)}})\trc T)
    =\nabla^\f_{g_{\widehat{(1)}}\trc X}(g_{\widehat{(2)}}\trc T),
\end{align*}
where $X\in\Xi_\star$ and $T\in\mathcal{T}_\star$ are arbitrary. In other words,
$\nabla^\f$ is $U\mathfrak{e}^\f$-equivariant.
If $\mathcal{F}$ is based on $U\mathfrak{e}$, $\nabla^\f$ is equivariant with respect
to the action of any leg of $\mathcal{F}$ or $\R$  (and their inverses).
By the linearity properties of the classical
covariant derivative and (\ref{bla}) we obtain
\begin{align*}
    \nabla^{\f}_{h\star X}T
    =&\nabla_{\overline{\mathcal{F}}\!_1\trc((\overline{\mathcal{F}}'_1\trc h)
    (\overline{\mathcal{F}}'_2\trc X))}(\overline{\mathcal{F}}\!_2\trc T)
    =((\overline{\mathcal{F}}\!_{1(1)}\overline{\mathcal{F}}'_1)\trc h)
    \nabla_{(\overline{\mathcal{F}}\!_{1(2)}\overline{\mathcal{F}}'_2)\trc X}
    (\overline{\mathcal{F}}\!_2\trc T)\\
    =&(\overline{\mathcal{F}}\!_{1}\trc h)
    \nabla_{(\overline{\mathcal{F}}\!_{2(1)}\overline{\mathcal{F}}'_1)\trc X}
    ((\overline{\mathcal{F}}\!_{2(2)}\overline{\mathcal{F}}'_2)\trc T)
    =(\overline{\mathcal{F}}\!_{1}\trc h)(\overline{\mathcal{F}}\!_2\trc
    (\nabla_{\overline{\mathcal{F}}'_1\trc X}
    (\overline{\mathcal{F}}'_2\trc T)))
    =h\star\nabla^{\f}_XT,
\end{align*}
which, together with \ $\nabla^{\f}_{Z+Z'}T=\nabla_{\overline{\mathcal{F}}\!_1\trc Z+\overline{\mathcal{F}}\!_1\trc Z'}(\bF_2\trc T)
=\nabla_{\overline{\mathcal{F}}\!_1\trc Z}(\bF_2\trc T)+\nabla_{\overline{\mathcal{F}}\!_1\trc Z'}(\bF_2\trc T)=\nabla^{\f}_{Z}T+\nabla^{\f}_{Z'}T$ gives  (\ref{eq03}).
By the Leibniz rule of the classical covariant derivative,
the inverse $2$-cocycle property, the equivariance of $\nabla$ and the Lie derivative
we obtain (\ref{ddsDuhv}).
The  rule (\ref{rightLeibniz})
\begin{align*}
    &\nabla^\f_X(T\star h)
    =(\mathcal{L}^\star_X(\overline{\R}_1\trc h))\star(\overline{\R}_2\trc T)
    +((\overline{\R}'_1\overline{\R}_1)\trc h)
    \star(\nabla^\f_{\overline{\r}'_2\trc X}(\overline{\R}_2\trc T))\\
    =&((\overline{\R}'_1\overline{\R}_2)\trc T)
    \star(\overline{R}'_2\trc\mathcal{L}^\star_X(\overline{\R}_1\trc h))
    +(\overline{\R}''_1\trc\nabla^\f_{\overline{\r}'_2\trc X}(\overline{\R}_2\trc T))
    \star((\overline{\R}''_2\overline{\R}'_1\overline{\R}_1)\trc h)\\
    =&((\overline{\R}'_1\overline{\R}_2)\trc T)
    \star(\mathcal{L}^\star_{\overline{\r}'_{2(1)}\trc X}
    ((\overline{\R}'_{2(2)}\overline{\R}_1)\trc h))
    +(\nabla^\f_{(\overline{\r}''_{1(1)}\overline{\r}'_2)\trc X}
    ((\overline{\R}''_{1(2)}\overline{\R}_2)\trc T))
    \star((\overline{\R}''_2\overline{\R}'_1\overline{\R}_1)\trc h)\\
    =&((\overline{\R}'_1\overline{\R}''_1\overline{\R}_2)\trc T)
    \star(\mathcal{L}^\star_{\overline{\r}'_{2}\trc X}
    ((\overline{\R}''_{2}\overline{\R}_1)\trc h))
    +(\nabla^\f_{(\overline{\r}'''_{1}\overline{\r}'_2)\trc X}
    ((\overline{\R}''_{1}\overline{\R}_2)\trc T))
    \star((\overline{\R}''_2\overline{\R}'''_2\overline{\R}'_1\overline{\R}_1)\trc h)\\
    =&(\overline{\R}'_1\trc T)\star(\mathcal{L}^\star_{\overline{\r}'_{2}\trc X}(h))
    +(\nabla^\f_{X}T)\star h
\end{align*}
holds for all $X\in\Xi$, $T\in\mathcal{T}^{p,q}$ and $h\in\X$.
The proof of (\ref{Leibddsg}-\ref{compleredd}) is
\begin{align*}
    &\nabla^{\f}_X( T\ots T')
    =\nabla_{\overline{\mathcal{F}}_1\trc X}
    (((\overline{\mathcal{F}}_{2(1)}\overline{\mathcal{F}}'_1)\trc T)
    \otimes((\overline{\mathcal{F}}_{2(2)}\overline{\mathcal{F}}'_2)\trc T'))\\
    =&(\nabla_{(\overline{\mathcal{F}}_{1(1)}\overline{\mathcal{F}}'_1)\trc X}
    ((\overline{\mathcal{F}}_{1(2)}\overline{\mathcal{F}}'_2)\trc T))
    \otimes(\overline{\mathcal{F}}_{2}\trc T')
    +((\overline{\mathcal{F}}_{2(1)}\overline{\mathcal{F}}'_1)\trc T)
    \otimes(\nabla_{\overline{\mathcal{F}}_1\trc X}
    ((\overline{\mathcal{F}}_{2(2)}\overline{\mathcal{F}}'_2)\trc T'))\\
    =&(\overline{\mathcal{F}}_1\trc\nabla_{\overline{\mathcal{F}}'_1\trc X}
    (\overline{\mathcal{F}}'_2\trc T))
    \otimes(\overline{\mathcal{F}}_{2}\trc T')
    +((\overline{\mathcal{F}}_{1(2)}\overline{\mathcal{F}}'_2
    \R'_1\overline{\R}_1)\trc T)
    \otimes((\nabla_{\overline{\mathcal{F}}_{1(1)}\overline{\mathcal{F}}'_1
    \r'_2\overline{\r}_2)\trc X}
    (\overline{\mathcal{F}}_{2}\trc T'))\\
    =&\nabla^{\f}_{X}(T)\ots T'
    +((\overline{\mathcal{F}}_{1(1)}\overline{\mathcal{F}}'_1
    \overline{\R}_1)\trc T)
    \otimes((\nabla_{\overline{\mathcal{F}}_{1(2)}\overline{\mathcal{F}}'_2
    \overline{\r}_2)\trc X}
    (\overline{\mathcal{F}}_{2}\trc T'))\\
    =&\nabla^{\f}_{X}(T)\ots T'
    +(\overline{\mathcal{R}}_1\trc T)\ots
    \nabla^{\f}_{\overline{\r}_2\trc X} T'  \hskip8cm  (\ref{Leibddsg'})\\
    =&\nabla^{\f}_{(\overline{\r}_1\overline{\r}'_2)\trc X}
    ((\overline{\R}'''_1\overline{\mathcal{R}}''_2)\trc T)
    \ots((\overline{\mathcal{R}}'''_2\overline{\mathcal{R}}_2\overline{\mathcal{R}}'_1
    \overline{\mathcal{R}}''_1)\trc T')
    +(\overline{\mathcal{R}}_1\trc T)\ots
    \nabla^{\f}_{\overline{\r}_2\trc X} T'\\
    =&\nabla^{\f}_{(\overline{\r}_{1\widehat{(1)}}\overline{\r}'_2)\trc X}
    ((\overline{\R}_{1\widehat{(2)}}\overline{\mathcal{R}}''_2)\trc T)
    \ots((\overline{\mathcal{R}}_2\overline{\mathcal{R}}'_1
    \overline{\mathcal{R}}''_1)\trc T')
    +(\overline{\mathcal{R}}_1\trc T)\ots
    \nabla^{\f}_{\overline{\r}_2\trc X} T'\\
    =&(\overline{\mathcal{R}}_1\trc\nabla^{\f}_{\overline{\r}'_2\trc X}
    (\overline{\mathcal{R}}''_2\trc T))
    \ots((\overline{\mathcal{R}}_2\overline{\mathcal{R}}'_1
    \overline{\mathcal{R}}''_1)\trc T')
    +(\overline{\mathcal{R}}_1\trc T)\ots
    \nabla^{\f}_{\overline{\r}_2\trc X} T'
\end{align*}
and (in complete analogy)
\begin{align*}
    \nabla^{\f}_X\la Y,\omega\ra_\star 
    =&\la\nabla^{\f}_X(Y),\omega\ra_\star + 
\la\bR_1\trc Y,\nabla^{\f}_{\br_2 \trc X} \omega\ra_\star \hskip6.5cm  (\ref{compleredd'})\\
=&\la\overline{\mathcal{R}}_1\trc(\nabla^{\f}_{\overline{\r}'_2\trc X}
    (\overline{\mathcal{R}}''_2\trc Y)),
    (\overline{\mathcal{R}}_2\overline{\mathcal{R}}'_1\overline{\mathcal{R}}''_1)
    \trc\omega\ra_\star
    +\la\overline{\mathcal{R}}_1\trc Y,
    \nabla^{\f}_{\overline{\r}_2\trc X} \omega\ra_\star,
\end{align*}
for all $X,Y\in\Xi$, $\omega\in\Omega^1$ and $T,T\in\mathcal{T}^{p,q}$.
In particular, we proved  (\ref{Leibddsg'}) and (\ref{compleredd'}) on
the way. 
Note that we further used the
cocommutativity of $\Delta$,
the equivariance property (\ref{eq01}), the (inverse)
$2$-cocycle condition, as well as the relations $(\Delta_\f\ot\id)\R=\R_{13}\R_{23}$ and
$\R^{-1}=\R_{21}$.

\subsection{Proof of Proposition \ref{TwistedLC}  }

First of all we prove that $\k$ is a Lie subalgebra of the equivariance Lie algebra
(\ref{equivLieAlg}) of $\nabla$, which implies that $\nabla^{\f}$ is a well-defined
covariant derivative according to Proposition~\ref{TwistedNabla}. Let $\xi\in\k$. Then,
making use of the Koszul formula, we obtain
\begin{align*}
    2\gm&(\nabla_{\xi_{(1)}\trc X}(\xi_{(2)}\trc Y),Z)\\
    =&(\xi_{(1)}\trc X)(\gm(\xi_{(2)}\trc Y,Z))
    +(\xi_{(2)}\trc Y)(\gm(Z,\xi_{(1)}\trc X))
    -Z(\gm(\xi_{(1)}\trc X,\xi_{(2)}\trc Y))\\
    &-\gm(\xi_{(1)}\trc X,[\xi_{(2)}\trc Y,Z])
    +\gm(\xi_{(2)}\trc Y,[Z,\xi_{(1)}\trc X])
    +\gm(Z,[\xi_{(1)}\trc X,\xi_{(2)}\trc Y])\\
    =&\xi_{(1)}\trc\left\{
    X(\gm(Y,S(\xi_{(2)})\trc Z))
    +Y(\gm(S(\xi_{(2)})\trc Z,X))
    -(S(\xi_{(2)})\trc Z)(\gm(X,Y))\right.\\
    &\left. -\gm(X,[Y,S(\xi_{(2)})\trc Z])
    +\gm(Y,[S(\xi_{(2)})\trc Z,X])
    +\gm(S(\xi_{(2)})\trc Z,[X,Y])
    \right\}\\
    =&\xi_{(1)}\trc(2\gm(\nabla_XY,S(\xi_{(2)})\trc Z))
    =2\gm(\xi\trc\nabla_XY,Z)
\end{align*}
for all $X,Y,Z\in\Xi$, where we further employed the $U\k$-equivariance of $\gm$ and 
of the pairing of vector fields with
forms,
as well as the cocommutativity of $U\k$ and the antipode properties.
Since $\gm$ is non-degenerate it follows that
$\xi\trc\nabla_XY=\nabla_{\xi_{(1)}\trc X}(\xi_{(2)}\trc Y)$, i.e. 
$\xi$ is an element of the equivariance Lie algebra of $\nabla$. 
Thus we have shown the inclusion $\k\subset\g$.
If \ $\mathcal{F}\in U\k\otimes U\k[[\nu]]$, \ then $\bF_2\trc \gm= \varepsilon(\bF_2) \gm$,
and using (\ref{twistcond}), (\ref{startensor}) we immediately find (\ref{twistedmetric2}). In fact
\begin{align*}
    \gm_\star(X,Y)
    =&\langle X,\langle Y,\gm^A\rangle_\star\star\gm_A\rangle_\star
    =\langle X,\langle Y,\gm^A\rangle\gm_A\rangle_\star
    =\langle\overline{\mathcal{F}}_1\trc X,
    \langle\overline{\mathcal{F}}_2\trc Y,\gm^A\rangle\gm_A\rangle\\
    =&\left\la  \: \bF_1\trc( X\ots Y),\bF_2\trc\gm \:\right\ra
    =\left\la\:  X\ots Y,\gm \:\right\ra,
\end{align*}
where in the last two equations the pairing is extended to double tensor products, see
(\ref{ExtPairing}). This reduces to
the undeformed $\gm(X,\!Y)$ if  $\bF\!=\!\1 \ot \1$.  
The twisted metric $\gm_\star$ is right $\X_\star$-linear, since
\begin{align*}
    \gm_\star(X,Y\star f)
    =&\langle X,\langle Y\star f,\gm^A\rangle_\star\star\gm_A\rangle_\star
    =\langle X,\langle Y,
    \overline{\mathcal{R}}_1\trc\gm^A\rangle_\star
    \star(\overline{\mathcal{R}}'_1\trc\gm_A)
    \star((\overline{\mathcal{R}}'_2\overline{\mathcal{R}}_2)\trc f)\rangle_\star\\
    =&\langle X,\langle Y,\gm^A\rangle_\star\star\gm_A\rangle_\star\star f
    =\gm_\star(X,Y)\star f
\end{align*}
for all $f\in\X_\star$. Next we prove that
$\nabla^{\f}$ is torsion-free with respect to the twisted torsion
and metric compatible with respect to the twisted  metric. The first
property holds since
\begin{align*}
    \tT^{\f}_\star(X,Y)
    =&\nabla_{\overline{\mathcal{F}}\!_1\trc X}(\overline{\mathcal{F}}\!_2\trc Y)
    -\nabla_{(\overline{\mathcal{F}}\!_1\mathcal{R}_2)\trc Y}
    ((\overline{\mathcal{F}}\!_2\mathcal{R}_1)\trc X)
    -[\overline{\mathcal{F}}\!_1\trc X,\overline{\mathcal{F}}\!_2\trc Y]\\
    =&\nabla_{\overline{\mathcal{F}}\!_1\trc X}(\overline{\mathcal{F}}\!_2\trc Y)
    -\nabla_{\overline{\mathcal{F}}\!_2\trc Y}
    (\overline{\mathcal{F}}\!_1\trc X)
    -[\overline{\mathcal{F}}\!_1\trc X,\overline{\mathcal{F}}\!_2\trc Y]
    =\tT_\star(\overline{\mathcal{F}}\!_1\trc X,\overline{\mathcal{F}}\!_2\trc Y)
    =0,
\end{align*}
while the second one holds because
\begin{align*}
    \mathcal{L}^\star_X&\gm_\star(Y,Z)
    =(\overline{\mathcal{F}}\!_1\trc X)
    (\gm((\overline{\mathcal{F}}\!_{2(1)}\overline{\mathcal{F}}'_1)\trc Y,
    (\overline{\mathcal{F}}\!_{2(2)}\overline{\mathcal{F}}'_2)\trc Z))\\
    =&\gm(\nabla_{\overline{\mathcal{F}}\!_1\trc X}
    ((\overline{\mathcal{F}}\!_{2(1)}\overline{\mathcal{F}}'_1)\trc Y),
    (\overline{\mathcal{F}}\!_{2(2)}\overline{\mathcal{F}}'_2)\trc Z)
    +\gm((\overline{\mathcal{F}}\!_{2(1)}\overline{\mathcal{F}}'_1)\trc Y,
    \nabla_{\overline{\mathcal{F}}\!_1\trc X}
    ((\overline{\mathcal{F}}\!_{2(2)}\overline{\mathcal{F}}'_2)\trc Z))\\
    =&\gm(\nabla_{(\overline{\mathcal{F}}\!_{1(1)}\overline{\mathcal{F}}'_1)\trc X}
    ((\overline{\mathcal{F}}\!_{1(2)}\overline{\mathcal{F}}'_2)\trc Y),
    \overline{\mathcal{F}}\!_{2}\trc Z)
    +\gm((\overline{\mathcal{F}}\!_{1(2)}\overline{\mathcal{F}}'_2)\trc Y,
    \nabla_{(\overline{\mathcal{F}}\!_{1(1)}\overline{\mathcal{F}}'_1)\trc X}
    (\overline{\mathcal{F}}\!_{2}\trc Z))\\
    =&\gm(\overline{\mathcal{F}}\!_1\trc(\nabla_{\overline{\mathcal{F}}'_1\trc X}
    (\overline{\mathcal{F}}'_2\trc Y)),
    \overline{\mathcal{F}}\!_{2}\trc Z)
    +\gm((\overline{\mathcal{F}}\!_{1(2)}
    \overline{\mathcal{F}}'_1\mathcal{R}_2)\trc Y,
    \nabla_{(\overline{\mathcal{F}}\!_{1(1)}
    \overline{\mathcal{F}}'_2\mathcal{R}_1)\trc X}
    (\overline{\mathcal{F}}\!_{2}\trc Z))\\
    =&\gm_\star(\nabla^{\f}_XY,Z)
    +\gm((\overline{\mathcal{F}}\!_{1(1)}
    \overline{\mathcal{F}}'_1\mathcal{R}_2)\trc Y,
    \nabla_{(\overline{\mathcal{F}}\!_{1(2)}
    \overline{\mathcal{F}}'_2\mathcal{R}_1)\trc X}
    (\overline{\mathcal{F}}\!_{2}\trc Z))\\
    =&\gm_\star(\nabla^{\f}_XY,Z)
    +\gm((\overline{\mathcal{F}}\!_{1}\mathcal{R}_2)\trc Y,
    \nabla_{(\overline{\mathcal{F}}\!_{2(1)}
    \overline{\mathcal{F}}'_1\mathcal{R}_1)\trc X}
    ((\overline{\mathcal{F}}\!_{2(2)}\overline{\mathcal{F}}'_2)\trc Z))\\
    =&\gm_\star(\nabla^{\f}_XY,Z)
    +\gm((\overline{\mathcal{F}}\!_{1}\mathcal{R}_2)\trc Y,\overline{\mathcal{F}}\!_2\trc(
    \nabla_{(\overline{\mathcal{F}}'_1\mathcal{R}_1)\trc X}
    (\overline{\mathcal{F}}'_2\trc Z)))\\
    =&\gm_\star(\nabla^{\f}_XY,Z)
    +\gm_\star(\mathcal{R}_2\trc Y,\nabla^{\f}_{\mathcal{R}_1\trc X}Z)
\end{align*}
for all $X,Y,Z\in\Xi$, which is equivalent to $\nabla^{\f}_X\gm=0$.
The statements about the twisted curvature and torsion are proven in 
\cite{AschieriCastellani2009}~Theorem~7, while the uniqueness of $\nabla^\f$
is given by \cite{AschieriCastellani2009}~Theorem~5.
We prove that the twisted torsion and curvature are right $\star$-linear in the last
argument if $\mathcal{F}$ is based on Killing vector fields.
Let $X,Y,Z\in\Xi$ and $h\in\X$. Then
\begin{align*}
    \tT^\f_\star(X,Y\star h)
    =&\nabla^\f_X(Y\star h)
    -\nabla^\f_{\overline{\r}_1\trc(Y\star h)}(\overline{\R}_2\trc X)
    -[X,Y\star h]\\
    =&\nabla^\f_X(Y)\star h
    +(\overline{\R}_1\trc Y)\star\mathcal{L}^\star_{\overline{\r}_2\trc X}(h)
    -\nabla^\f_{(\overline{\r}_{1(1)}\trc Y)\star(\overline{\r}_{1(2)}\trc h)}
    (\overline{\R}_2\trc X)\\
    &-[X,Y]\star h
    -(\overline{\R}_1\trc Y)\star\mathcal{L}^\star_{\overline{\r}_2\trc X}(h)\\
    =&\nabla^\f_X(Y)\star h
    -\nabla^\f_{\overline{\r}_{1}\trc Y}
    ((\overline{\R}'_1\overline{\R}''_2\overline{R}_2)\trc X)
    \star((\overline{\R}'_2\overline{\R}''_{1})\trc h)
    -[X,Y]\star h\tT^\f_\star(X,Y)\star h
\end{align*}
proves right $\X_\star$-linearity of the twisted torsion.
Finally,
\begin{align*}
    \rR^\f_\star&(X,Y,Z\star h)
    =\nabla^\f_X\nabla^\f_Y(Z\star h)
    -\nabla^\f_{\overline{\mathcal{R}}_1\trc Y}
    \nabla^\f_{\overline{\mathcal{R}}_2\trc X}(Z\star h)
    -\nabla^\f_{[X,Y]}(Z\star h)\\
    =&\nabla^\f_X((\nabla^\f_YZ)\star h
    +(\overline{\mathcal{R}}_1\trc Y)\star(\mathcal{L}^\star_{\overline{R}_2\trc X}h))
    -\nabla^\f_{\overline{\mathcal{R}}_1\trc Y}(
    (\nabla^\f_{\overline{\mathcal{R}}_2\trc X}Z)\star h\\
    &+(\overline{\mathcal{R}}'_1\trc Z)
    \star(\mathcal{L}^\star_{(\overline{\mathcal{R}}'_2\overline{\mathcal{R}}_2)\trc X}h))
    -(\nabla^\f_{[X,Y]}Z)\star h
    -(\overline{\mathcal{R}}_1\trc Z)
    \star(\mathcal{L}^\star_{\overline{\mathcal{R}}_2\trc[X,Y]}h)\\
    =&\rR^\f_\star(X,Y,Z)\star h
    +(\overline{\mathcal{R}}_1\trc(\nabla^\f_YZ))
    \star(\mathcal{L}^\star_{\overline{\mathcal{R}}_2\trc X}h)
    +\nabla^\f(\overline{\mathcal{R}}_1\trc Y)\star(\mathcal{L}^\star_{\overline{R}_2\trc X}h)\\
    &+((\overline{\mathcal{R}}'_1\overline{\mathcal{R}}_1)\trc Y)
    (\mathcal{L}^\star_{\overline{\mathcal{R}}'_2\trc X}
    (\mathcal{L}^\star_{\overline{R}_2\trc X}h))
    -(\overline{\mathcal{R}}'_1\trc(\nabla^\f_{\overline{\mathcal{R}}_2\trc X}Z))
    \star(\mathcal{L}^\star_{(\overline{\mathcal{R}}'_2\overline{\mathcal{R}}_1)\trc Y}h)\\
    &-(\nabla^\f_{\overline{\mathcal{R}}_1\trc Y}(\overline{\mathcal{R}}'_1\trc Z))
    \star(\mathcal{L}^\star_{(\overline{\mathcal{R}}'_2\overline{\mathcal{R}}_2)\trc X}h)
    -((\overline{\mathcal{R}}''_1\overline{\mathcal{R}}'_1)\trc Z)
    \star\mathcal{L}^\star_{(\overline{\mathcal{R}}''_2\overline{\mathcal{R}}_1)\trc Y}
    (\mathcal{L}^\star_{(\overline{\mathcal{R}}'_2\overline{\mathcal{R}}_2)\trc X}h)\\
    &-(\overline{\mathcal{R}}_1\trc Z)
    \star(\mathcal{L}^\star_{\overline{\mathcal{R}}_2\trc[X,Y]}h)
    =\rR^\f_\star(X,Y,Z)\star h,
\end{align*}
where in the last equation the eighth term cancels with the fourth and seventh term, the second and sixth cancel each other, and so do the third and fifth terms.

\subsection{Proof of Proposition \ref{XCrepr}  and eq. (\ref{extra})}

Decompose \ $X=X_t+X_{\scriptscriptstyle \perp}$, $ X_{\scriptscriptstyle \perp}=  X^a\Np^a$.  \ Then \ $X(f^b)=X^a\Np^a(f^b) 
= X^aK^{ac}(f^{ci}f^{b}_i)=  X^b$ must belong
to $\C$ for all $b=1,..,k$, i.e. must be of the form
$X^a= f^bX^a_b $, for some $X^a_b\in\X$.
Hence $X_{\scriptscriptstyle \perp}= f^b(  X^a_b\Np^a)$
belongs to $\Xi_{\C\C}$, and $X_t\in[X]$.
Decompose \ $\omega=\omega_t+\omega_{\scriptscriptstyle \perp}$.  \ 
One can find an atlas of ${\cal D}_f$, with a pair $\{e_i\}$, $\{\theta^i\}$
of dual frames in each chart,   
such that $\{e_\alpha\}_{\alpha=1}^{n-k}$ is a basis of $\Xi_t$ and 
 $\{\theta^\alpha\}_{\alpha=1}^{n-k}$ is a basis of $\Omega_t$.
Then 
 $\omega_t=\omega_{\alpha}\theta^\alpha$, and  for all $X= X^\alpha e_\alpha\in\Xi_t$ \ it is
 $\la X,\omega\ra=\la X,\omega_t\ra=  X^\alpha \omega_{\alpha}$;
by Theorem \ref{DecoC}, this  belongs to $\C$ for all $(X^\alpha)$ if and only if
\ $\omega_{\alpha}= f^a\omega_{\alpha}^a $, \ for some $\omega_{\alpha}^a\in\X$.
Hence $\omega_{\scriptscriptstyle \perp}=
 f^a\omega_{\alpha}^a\theta^\alpha $
belongs to $\Omega_{\C\C}$, and $\omega_t\in[\omega]$.

In Proposition \ref{Propdiamond} we have   shown  that
$\Omp\subseteq\Omp':=\left\{\omega\in\Omega  \:|\:   \la \Xi_t,\omega\ra=0 \right\}$. 
Conversely, for any $\omega\in\Omp'$  we have
$0=\la X,\omega\ra=\la X,\omega_t\ra$ for all $X\in\Xi_t$, whence $\omega_t=0$
and $\omega=\omega_{\scriptscriptstyle \perp}\in \Omp$. This proves the first equality in 
(\ref{extra}). To prove the last equality, decompose $X=X_t+X_{\scriptscriptstyle \perp}$; this belongs to $\Xi_\C$ if and only if $X_{\scriptscriptstyle \perp}$ is of the form
$X_{\scriptscriptstyle \perp}=
f^bX_{ba}\Np^a$,
whence \ $ \la X,\omega\ra= \la X_t,\omega\ra+
f^bX_{ba}\la\Np^a,\omega\ra$ belongs to $\C$ iff $\la X_t,\omega\ra$ does, for all $X_t\in\Xi_t$.

\subsection{Proof of Proposition \ref{QuantumGauss}}\label{proofQGauss}

We reduce eq.(\ref{GaussQuantum}) to the $\nu$-linear extension of eq. (\ref{GaussClassic}).
For $X,Y,Z,W\in\Xi_t$ we obtain
\begin{align*}
    \gm_\star (\rR_{t\star}^\f(X,Y)Z,W)
    =&\gm(\rR_{t}((\overline{\mathcal{F}}_{1(1)}\overline{\mathcal{F}}'_{1(1)}
    \overline{\mathcal{F}}''_{1})\trc X,
    (\overline{\mathcal{F}}_{1(2)}\overline{\mathcal{F}}'_{1(2)}
    \overline{\mathcal{F}}''_{2})\trc Y)
    ((\overline{\mathcal{F}}_{1(3)}\overline{\mathcal{F}}'_{2})\trc Z),
    \overline{\mathcal{F}}_{2}\trc W)\\
       =&\gm(\rR_{t}((\overline{\mathcal{F}}_{1(1)}
    \overline{\mathcal{F}}''_{1})\trc X,
    (\overline{\mathcal{F}}_{1(2)}
    \overline{\mathcal{F}}''_{2})\trc Y)
    ((\overline{\mathcal{F}}_{2(1)}\overline{\mathcal{F}}'_{1})\trc Z),
    (\overline{\mathcal{F}}_{2(2)}\overline{\mathcal{F}}'_2)\trc W)
\end{align*}
\begin{align*}
  \gm_\star&(II^\f_\star(X,\overline{\mathcal{R}}_1\trc Z),
    II^\f_\star(\overline{\mathcal{R}}_2\trc Y,W))\\
    =&\gm(II((\overline{\mathcal{F}}_{1(1)}\overline{\mathcal{F}}'_{1})\trc X,
    (\overline{\mathcal{F}}_{1(2)}\overline{\mathcal{F}}'_{2}
    \overline{\mathcal{R}}_{1})\trc Z),
    II((\overline{\mathcal{F}}_{2(1)}\overline{\mathcal{F}}''_{1}
    \overline{\mathcal{R}}_{2})\trc Y,
    (\overline{\mathcal{F}}_{2(2)}\overline{\mathcal{F}}''_{2})\trc W))\\
      =&\gm(II(\overline{\mathcal{F}}_{1}\trc X,
    (\overline{\mathcal{F}}_{2(1)}\overline{\mathcal{F}}'_{1}
    \overline{\mathcal{R}}_{1})\trc Z),
    II((\overline{\mathcal{F}}_{2(2)}\overline{\mathcal{F}}'_{2(1)}
    \overline{\mathcal{F}}''_{1}
    \overline{\mathcal{R}}_{2})\trc Y,
    (\overline{\mathcal{F}}_{2(3)}\overline{\mathcal{F}}'_{2(2)}
    \overline{\mathcal{F}}''_{2})\trc W))\\
       =&\gm(II(\overline{\mathcal{F}}_{1}\trc X,
    (\overline{\mathcal{F}}_{2(1)}\overline{\mathcal{F}}'_{1(1)}\overline{\mathcal{F}}''_1
    \overline{\mathcal{R}}_{1})\trc Z),
    II((\overline{\mathcal{F}}_{2(2)}\overline{\mathcal{F}}'_{1(2)}
    \overline{\mathcal{F}}''_{2}
    \overline{\mathcal{R}}_{2})\trc Y,
    (\overline{\mathcal{F}}_{2(3)}\overline{\mathcal{F}}'_{2})\trc W))\\
    =&\gm(II(\overline{\mathcal{F}}_{1}\trc X,
    (\overline{\mathcal{F}}_{2(2)}\overline{\mathcal{F}}'_{1(2)}
    \overline{\mathcal{F}}''_2)\trc Z),
    II((\overline{\mathcal{F}}_{2(1)}\overline{\mathcal{F}}'_{1(1)}
    \overline{\mathcal{F}}''_{1})\trc Y,
    (\overline{\mathcal{F}}_{2(3)}\overline{\mathcal{F}}'_{2})\trc W))\\
     =&\gm(II(\overline{\mathcal{F}}_{1}\trc X,
    (\overline{\mathcal{F}}_{2(2)}\overline{\mathcal{F}}'_{2(1)}
    \overline{\mathcal{F}}''_1)\trc Z),
    II((\overline{\mathcal{F}}_{2(1)}\overline{\mathcal{F}}'_{1})\trc Y,
    (\overline{\mathcal{F}}_{2(3)}
    \overline{\mathcal{F}}'_{2(2)}\overline{\mathcal{F}}''_2)\trc W))\\
       =&\gm(II((\overline{\mathcal{F}}_{1(1)}\overline{\mathcal{F}}'_1)\trc X,
    (\overline{\mathcal{F}}_{2(1)}
    \overline{\mathcal{F}}''_1)\trc Z),
    II((\overline{\mathcal{F}}_{1(2)}\overline{\mathcal{F}}'_{2})\trc Y,
    (\overline{\mathcal{F}}_{2(2)}\overline{\mathcal{F}}''_2)\trc W))
\end{align*}
\begin{align*}
        -\gm_\star&(II^\f_\star(\overline{\mathcal{R}}_{1\widehat{(1)}}\trc Y,
    \overline{\mathcal{R}}_{1\widehat{(2)}}\trc Z),
    II^\f_\star(\overline{\mathcal{R}}_2\trc X,W))\\
    =&-\gm(II((\overline{\mathcal{F}}_{1(1)}\overline{\mathcal{F}}'_{1}
    \overline{\mathcal{R}}_{1\widehat{(1)}})\trc Y,
    (\overline{\mathcal{F}}_{1(2)}\overline{\mathcal{F}}'_{2}
    \overline{\mathcal{R}}_{1\widehat{(2)}})\trc Z),
    II((\overline{\mathcal{F}}_{2(1)}\overline{\mathcal{F}}''_{1}
    \overline{\mathcal{R}}_{2})\trc X,
    (\overline{\mathcal{F}}_{2(2)}\overline{\mathcal{F}}''_{2})\trc W))\\
    =&-\gm(II((\overline{\mathcal{F}}_{1(1)}
    \overline{\mathcal{R}}_{1(1)}\overline{\mathcal{F}}'_{1})\trc Y,
    (\overline{\mathcal{F}}_{1(2)}
    \overline{\mathcal{R}}_{1(2)}\overline{\mathcal{F}}'_{2})\trc Z),
    II((\overline{\mathcal{F}}_{2(1)}\overline{\mathcal{F}}''_{1}
    \overline{\mathcal{R}}_{2})\trc X,
    (\overline{\mathcal{F}}_{2(2)}\overline{\mathcal{F}}''_{2})\trc W))\\
    =&-\gm(II((\overline{\mathcal{F}}_{1(1)}\overline{\mathcal{F}}''_{1(1)}
    \overline{\mathcal{R}}_{1(1)}\overline{\mathcal{F}}'_{1})\trc Y,
    (\overline{\mathcal{F}}_{1(2)}\overline{\mathcal{F}}''_{1(2)}
    \overline{\mathcal{R}}_{1(2)}\overline{\mathcal{F}}'_{2})\trc Z),
    II((\overline{\mathcal{F}}_{1(3)}\overline{\mathcal{F}}''_{2}
    \overline{\mathcal{R}}_{2})\trc X,
    \overline{\mathcal{F}}_{2}\trc W))\\
    =&-\gm(II((\overline{\mathcal{F}}_{1(1)}\overline{\mathcal{F}}''_{2(1)}
    \overline{\mathcal{F}}'_{1})\trc Y,
    (\overline{\mathcal{F}}_{1(2)}\overline{\mathcal{F}}''_{2(2)}
    \overline{\mathcal{F}}'_{2})\trc Z),
    II((\overline{\mathcal{F}}_{1(3)}\overline{\mathcal{F}}''_{1})\trc X,
    \overline{\mathcal{F}}_{2}\trc W))\\
    =&-\gm(II((\overline{\mathcal{F}}_{1(2)}\overline{\mathcal{F}}''_{1(2)}
    \overline{\mathcal{F}}'_{2})\trc Y,
    (\overline{\mathcal{F}}_{1(3)}\overline{\mathcal{F}}''_{2})\trc Z),
    II((\overline{\mathcal{F}}_{1(1)}\overline{\mathcal{F}}''_{1(1)}
    \overline{\mathcal{F}}'_1)\trc X,
    \overline{\mathcal{F}}_{2}\trc W))\\
    =&-\gm(II((\overline{\mathcal{F}}_{1(2)}
    \overline{\mathcal{F}}'_{2})\trc Y,
    (\overline{\mathcal{F}}_{2(1)}\overline{\mathcal{F}}''_{1})\trc Z),
    II((\overline{\mathcal{F}}_{1(1)}
    \overline{\mathcal{F}}'_1)\trc X,
    (\overline{\mathcal{F}}_{2(2)}\overline{\mathcal{F}}''_2)\trc W)),
\end{align*}
using the $2$-cocycle property of $\overline{\mathcal{F}}$ and its consequences (\ref{bla}),
eq. (\ref{inter-2}), the definition of $\R$, as well as the
$U\mathfrak{k}$-equivariance of $\rR^\f_\star$, $\rR_{t\star}^\f$ and $II^\f_\star$.
The sum of these three terms  is the right-hand side (rhs) of (\ref{GaussQuantum}). By (\ref{GaussClassic}) it equals the  left-hand side:
\begin{align*}
    \mbox{rhs}(\ref{GaussQuantum})  =&\gm(\rR_t((\overline{\mathcal{F}}_{1(1)}
    \overline{\mathcal{F}}'_{1})\trc X,
    (\overline{\mathcal{F}}_{1(2)}
    \overline{\mathcal{F}}'_{2})\trc Y)
    ((\overline{\mathcal{F}}_{2(1)}\overline{\mathcal{F}}''_{1})\trc Z),
    (\overline{\mathcal{F}}_{2(2)}\overline{\mathcal{F}}''_2)\trc W)\\
    &+\gm(II((\overline{\mathcal{F}}_{1(1)}\overline{\mathcal{F}}'_1)\trc X,
    (\overline{\mathcal{F}}_{2(1)}
    \overline{\mathcal{F}}''_1)\trc Z),
    II((\overline{\mathcal{F}}_{1(2)}\overline{\mathcal{F}}'_{2})\trc Y,
    (\overline{\mathcal{F}}_{2(2)}\overline{\mathcal{F}}''_2)\trc W))\\
    &-\gm(II((\overline{\mathcal{F}}_{1(2)}
    \overline{\mathcal{F}}'_{2})\trc Y,
    (\overline{\mathcal{F}}_{2(1)}\overline{\mathcal{F}}''_{1})\trc Z),
    II((\overline{\mathcal{F}}_{1(1)}
    \overline{\mathcal{F}}'_1)\trc X,
    (\overline{\mathcal{F}}_{2(2)}\overline{\mathcal{F}}''_2)\trc W))\\
    =&\gm(\rR((\overline{\mathcal{F}}_{1(1)}\overline{\mathcal{F}}'_1)\trc X),
    (\overline{\mathcal{F}}_{1(2)}\overline{\mathcal{F}}'_2)\trc Y)
    (\overline{\mathcal{F}}_{2(1)}\overline{\mathcal{F}}''_1)\trc Z),
    \overline{\mathcal{F}}_{2(2)}\overline{\mathcal{F}}''_2)\trc W)\\
    =&\gm_\star(\rR^\f_\star(X,Y)Z,W). 
\end{align*}

\end{document}